\documentclass[reprint,
superscriptaddress,
bibnotes,
 amsmath,amssymb,
 aps,
]{revtex4-2}

\usepackage{appendix}
\usepackage{import}
\usepackage{graphicx}
\usepackage{dcolumn}%
\usepackage{bm}%
\usepackage{graphicx}
 \usepackage{relsize}
\usepackage{braket}
\usepackage{algorithm,algpseudocode,algorithmicx}
\usepackage{amsthm}
\usepackage{xcolor}
\usepackage{pythontex}
\usepackage{enumitem}
\usepackage{vcell}
\usepackage{booktabs}
\usepackage{array}
\usepackage{hyperref}
\usepackage{cleveref}
\usepackage[bb=boondox]{mathalfa}
\usepackage[most]{tcolorbox}
\usepackage{lipsum}

\newcommand{\red}[1]{\textcolor{red}{#1}}

\newcommand{\commentout}[1]{}

\newcommand{\delete}[1]{}

\newcommand{\brm}[1]{\bm{\mathrm{#1}}}

\newcommand{\hl}{\brm{L}_{k}}

\newcommand{\bk}{\brm{B}_{k}}

\newcommand{\bkk}{\brm{B}_{k+1}}

\newcommand{\UTheta}{\brm{U}_{\brm{\Theta}^{k}}}

\newcommand{\UOmega}{\brm{U}_{\brm{\Omega}^{k}}}

\newcommand{\USigma}{\brm{U}_{\brm{\Sigma}_{k}}}
\newcommand{\USigmaLU}{\brm{U}_{\brm{\Sigma}_{k\pm1}}}
\newcommand{\Om}{\brm{O}_{m_p}}
\newcommand{\ALU}{{\brm{A}_{[\pm]}}}
\newcommand{\UALU}{\brm{U}_{\brm{A}_{[\pm]}/\gamma_1}}

\newtheorem{theorem}{Theorem}
\newtheorem{lemma}{Lemma}

\newtheorem{corollary}{Corollary}
\newtheorem{definition}{Definition}
\newtheorem{problem}{Task}

\Crefname{equation}{Eq.}{Eqs.}
\Crefname{figure}{Fig.}{Figs.}
\Crefname{problem}{Task}{Tasks}

\begin{document}

\title{Efficient Quantum Algorithms for Higher-Order Coupled Oscillators}

\author{Caesnan M. G. Leditto}
\email{caesnan.leditto@monash.edu}
\affiliation{School of Physics and Astronomy, Monash University, Clayton, VIC 3168, Australia}
\affiliation{Quantum Systems, Data61, CSIRO, Clayton, VIC 3168, Australia}

\author{Angus Southwell}
\affiliation{School of Physics and Astronomy, Monash University, Clayton, VIC 3168, Australia}

\author{Muhammad Usman}
\affiliation{Quantum Systems, Data61, CSIRO, Clayton, VIC 3168, Australia}
\affiliation{School of Physics, The University of Melbourne, Parkville, VIC 3052, Australia}
\affiliation{School of Physics and Astronomy, Monash University, Clayton, VIC 3168, Australia}

\author{Kavan Modi}
\email{kavan@quantumlah.org}
\affiliation{Science, Mathematics and Technology Cluster, Singapore University of Technology and Design, 8 Somapah Road, 487372 Singapore}
\affiliation{School of Physics and Astronomy, Monash University, Clayton, VIC 3168, Australia}

\date{\today}%
\begin{abstract}
Higher-order networks with multiway interactions can exhibit collective dynamical phenomena that are absent in traditional pairwise network models. However, analyzing such dynamics becomes computationally prohibitive as their state space grows combinatorially in the multiway interaction order. Here we develop quantum algorithms for two central tasks -- synchronization estimation and certification of the no-phase-locking regime -- in the simplicial Kuramoto model. This model is a higher-order generalization of the celebrated Kuramoto model for coupled oscillators on graph-based networks. Under explicit assumptions on data access and types, and simplicial structure, we derive end-to-end quantum gate complexities and identify regimes with polynomial quantum advantage for synchronization estimation and super-polynomial quantum advantage for no-phase-locking certification over classical methods. More broadly, these results extend quantum algorithms for higher-order networks from structural analysis to nonlinear dynamical diagnostics, easing a major computational bottleneck and opening a route to quantum methods for probing higher-order phenomena beyond the reach of direct classical approaches.
\end{abstract}

\maketitle


\section{Introduction}

Network theory provides a framework for studying complex dynamical systems. While many models rely on pairwise interactions, an increasing body of evidence shows that in some systems, interactions occur in groups of three or more units and cannot be reduced to pairwise effects. For example, a collection of neurons can become active together, large groups of people behave differently than individuals or small groups, and disease can spread through gatherings rather than only through one-to-one contact~\cite{Petri2014HomologicalNetworks,xu2016representing,giusti2016two,Reimann2017Cliques,Patania2017TopologicalData,Iacopini2019SimplicialContagion,Matamalas2020Abruptepidemicspreading}. Whereas the properties of networks with pairwise interactions can be represented on graphs, they cannot represent complex mechanisms stemming from multiway collective behaviour. Higher-order network models~\cite{Bianconi_2021,Bick2023WhatNetworks}, such as hypergraphs and simplicial complexes, are needed to represent multiway interactions. Here, the encoding is not only edges but also triangles, tetrahedra, and their generalisations. This richer description can reveal effects that are hidden in pairwise models~\cite{Battiston2020NetworksDynamics,Battiston2021TheSystems,Majhi2022DinamicsReview}. 

A paradigmatic model for collective dynamics on higher-order networks is the \textit{simplicial Kuramoto model} (SKM)~\cite{Millan2020ExplosiveComplexes,DeVille2021ConsensusSynchronization,Arnaudon2022ConnectingHodgeSakaguchi,Nurisso2024UnifiedSimplicialKuramoto}. The SKM extends the Kuramoto model (KM)~\cite{Kuramoto1975,strogatz2000kuramoto,Acebron2005KuramotoReview,Rodrigues2016KuramotoComplexNetworks}, in which oscillators sit on the nodes of a graph, and pairwise couplings tend to pull their rhythms together, to the setting of higher-order networks. While the KM has successfully described functional biological activities~\cite{breakspear2010generative, cabral2011role, sadilek2015physiologically}, power systems dynamics~\cite{doerfler2013synchronization, motter2013spontaneous, nardelli2014models}, and even social dynamics~\cite{pluchino2005changing, pluchino2006opinion}, its reliance on pairwise interactions limits modeling dynamics on higher-order networks. In the SKM, oscillatory variables can be assigned to simplices---such as edges, triangles, tetrahedra, and their higher-dimensional counterparts---and coupled through adjacent lower- and higher-dimensional simplices. In this way, the model broadens the range of systems that can be described and captures dynamical effects that are invisible to pairwise descriptions~\cite{Ghorbanchian2021Higher-orderSignals,calmon2022dirac,millan2025topology}. However, the computational analysis of these richer dynamics is much harder than that of graphs as the number of oscillators grows combinatorially. As a result, even basic analyses of the SKM dynamics can rapidly become infeasible with standard classical computational techniques.

Recent quantum algorithms have shown promise in extracting information from higher-order networks, including counting high-dimensional holes~\cite{Lloyd2016QuantumData,Gyurik2022TowardsAnalysis,Hayakawa2022QuantumAnalysis,McArdle2022AQubits,hayakawa2024quantumwalkssimplicialcomplexes} and spectral filtering on data with simplicial structure~\cite{Leditto2023topological}. These algorithms serve as quantum subroutines in data analysis pipelines, such as topological data analysis~\cite{Carlsson2009TopologyData} and topological signal processing~\cite{Barbarossa2020TopologicalComplexes}, that are useful for classification and learning tasks. When comparing the resources required by classical and quantum methods, quantum algorithms exhibit a super-polynomial separation across a large family of instances~\cite{berry2023analyzing,leditto2024quantumhodgeranktopologybasedrank}. Such algorithms are also robust to current dequantization techniques~\cite{Chia2020DequantizedQSVT, Gharibian2022DequantizingConjecture}, thereby strengthening their claim to advantage. However, these advances have focused mainly on the topological structure and static data associated with the networks. This raises the question of whether quantum algorithms can make the analysis of dynamical systems on higher-order networks tractable in regimes where classical approaches become computationally prohibitive. 

In this paper, we address this challenge by developing quantum algorithms for two dynamical diagnostics in the SKM. These two tasks are (i) estimating how strongly the oscillator phases align, and (ii) certifying whether the system fails to settle to a shared long-term frequency. We refer to them as synchronization estimation and certification of the no-phase-locking (NPL) regime. Measuring synchronization is important because it provides a concise metric for describing the instantaneous state of a system's emergent behavior, such as a group of neurons firing together~\cite{Haruna2016HodgeNetworks,Sizemore2019TheNeuroscientist,Parastesh2022SynchronizationIH,Majhi2025PatternsON}. Complementary to this, identifying the NPL regime is important for describing the long-term dynamics of a system, and whether they settle into a consistent asymptotic state, without the expense of simulating the system for arbitrarily long times~\cite{FranciChailletPanteleyLamnabhi-Lagarrigue2012, LowetRobertsBonizziKarelDeWeerd2016,WeerasingheDCBB19,ReynerParraHuguet2022}. \Cref{fig:SKM_plus_quantum_tasks} summarizes this setting and the two algorithmic tasks studied in this work. 

\begin{figure*}
    \centering
    \includegraphics[width=1\linewidth]{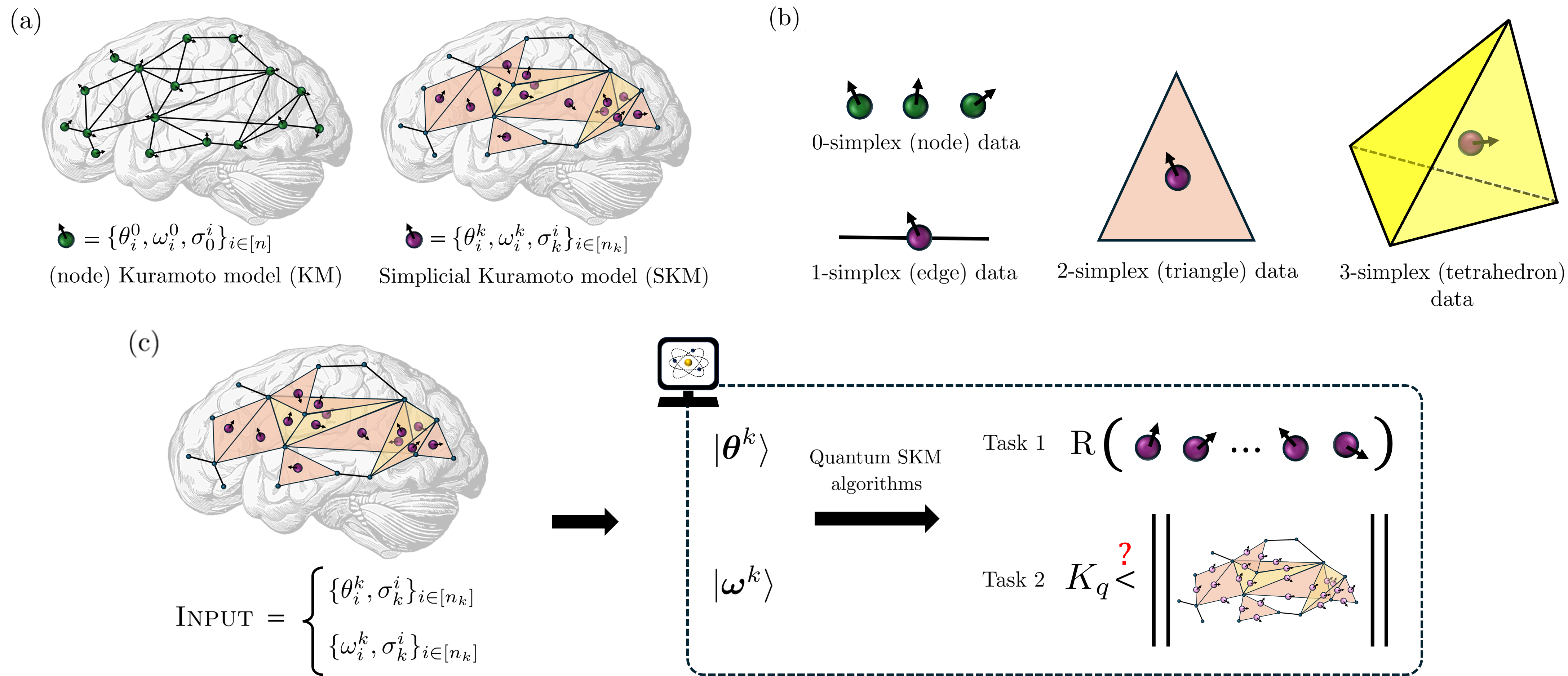}
    \caption{A prominent anticipated application of SKM is the neural signal processing in brain networks. Thus, we use brain networks as a base example to first differentiate between KM and SKM, and then to highlight our achievements. (a) (Left) A standard/node Kuramoto model (KM), where each oscillator is attributed to the node of a graph with node data (green dot) comprised of the node phase, node natural frequency, and their corresponding node $\{\theta^{0}_{i},\omega^{0}_{i},\sigma_{0}^{i}\}_{i\in[n]}$. (Right) A simplicial Kuramoto model (SKM), where each oscillator is attributed to the $k$-simplex in a simplicial complex with the SKM data (purple dot) is given by the simplicial phase, natural frequency, their corresponding simplex $\{\theta^{k}_{i},\omega^{k}_{i},\sigma_{k}^{i}\}_{i\in[n_{k}]}$ for $k>0$. In brain network modeling, KM or SKM is used to represent neuronal activity and communication. In the KM, a node represents a single brain region, and its data correspond to the neural activity signal. In the SKM, the $k$-simplex is a group of $(k+1)$ brain regions, and the SKM data represent their joint neural coactivity signal. (b) A visualization of the SKM data attributed to nodes, edges, triangles, and tetrahedra, which are simplices of dimension/order $k$, for $k=0,1,2,3$. (c) The high-level pipeline of the quantum algorithms, which we call the \textit{quantum SKM algorithms}: (left) from SKM data, \Cref{alg:T1} takes $\{\theta_{i}^{k},\sigma_{k}^{i}\}_{i\in[n_{k}]}$ as input and \Cref{alg:T2} takes $\{\omega_{i}^{k},\sigma_{k}^{i}\}_{i\in[n_{k}]}$ as input; (right) given quantum states $\ket{\brm{\theta^{k}}}$ and $\ket{\brm{\omega^{k}}}$ which encode $\{\theta_{i}^{k}\}_{i\in[n_{k}]}$ and $\{\omega_{i}^{k}\}_{i\in[n_{k}]}$ in the amplitudes, respectively, and $\{\sigma_{k}^{i}\}_{i\in[n_{k}]}$ in the basis state, the algorithms compute  simplicial order parameter $\mathrm{R}$ estimation that measures the synchronization of the SKM at a given time (Task 1), and no-phase-locking regime certification that identify whether the coupling constant $K_{\mathfrak{q}}$ is smaller or greater than the critical coupling determined by the projected natural frequencies (Task 2). The higher-order synchronization is argued to effectively capture the information spread in neuronal activity~\cite{Haruna2016HodgeNetworks,Parastesh2022SynchronizationIH,Majhi2025PatternsON}, and localization in the brain network cavities~\cite{Reimann2017Cliques,Sizemore2019TheNeuroscientist,millan2025topology}. On the other hand, the detection of no phase locking indicates that the neural signal cannot achieve stability. The SKM no-phase-locking could be related to pathological oscillations in neurological disease~\cite{FranciChailletPanteleyLamnabhi-Lagarrigue2012, WeerasingheDCBB19}, as well as help to model effective neuronal communication~\cite{LowetRobertsBonizziKarelDeWeerd2016,ReynerParraHuguet2022} beyond the pairwise interaction model. 
    }
    \label{fig:SKM_plus_quantum_tasks}
\end{figure*}

Building on tools from quantum topological signal processing~\cite{Leditto2023topological,leditto2024quantumhodgeranktopologybasedrank}, we derive end-to-end quantum algorithms and gate complexities for these two SKM tasks. Under input models in which the simplicial data admit efficient quantum-state preparation, and for a family of simplicial complexes with suitable topological structure, we benchmark the resulting quantum costs against classical arithmetic costs such as sparse matrix--vector multiplications. Under these assumptions, our synchronization algorithm achieves a polynomial improvement in the number of nodes, while our NPL-certification algorithm yields a super-polynomial advantage. More broadly, these results extend quantum algorithms for simplicial data from structural descriptors to nonlinear dynamical diagnostics, and show that higher-order network dynamics can provide a natural setting for quantum advantage. These algorithms not only remove a major computational bottleneck for tasks important in network science, but they also pave the way to developing quantum methods that could help uncover higher-order phenomena that are otherwise difficult to access.
\section{Higher-order interactions in coupled phase oscillators}\label{sec:review_SKM}

\subsection{Kuramoto model (KM) of coupled phase oscillators}

Many dynamical systems can be modelled as populations of oscillators that interact through their phase differences~\cite{Acebron2005KuramotoReview,Rodrigues2016KuramotoComplexNetworks,Guo2021KuramotoPowerSys}.  When interactions are pairwise, the canonical description is the Kuramoto model~\cite{Kuramoto1975}.  {Let $\mathcal{V}$ be a vertex set with $n=|\mathcal{V}|$, $\mathcal{E}$ be the edge set constructed from any two elements $i,j\in\mathcal{V}$, and $G=(\mathcal{V},\mathcal{E})$ be an undirected graph. Let $\mathbf{A}$ be the adjacency matrix of the graph $G$.} One assigns a phase $\theta_{i}$ and a natural frequency $\omega_{i}$ to each vertex $i$.  The standard Kuramoto dynamics is given by
\begin{equation}
\dot{\theta}_{i} =\omega_{i}-K\sum_{j=1}^{n}A_{ij}\,\sin(\theta_{j} - \theta_{i})\,,
\label{eq:node_kuramoto}
\end{equation}
where $K\ge 0$ is the coupling strength and $A_{ij}$ is the $ij$-th element of $\brm{A}$, i.e., $A_{ij}=1$ if $\{i,j\}\in\mathcal{E}$, and zero otherwise. Despite its simplicity, this model captures a variety of synchronization phenomena, e.g., the phase configuration when 
\begin{align}
    |\theta_{i}-\theta_{j}|=0,\quad \text{for all $\{i,j\}\in\mathcal{E}$,}
\end{align}
and has been extensively studied on complex networks~\cite{Rodrigues2016KuramotoComplexNetworks,Gupta2014KuramotoSynchronization}. In this pairwise setting, the degree of synchrony is often quantified by the order parameter
\begin{align}\label{eq:KM_order_parameter}
    \overline{\mathrm{R}}(\brm{\theta})=\frac{1}{n}\,\left|\sum_{j\in[n]}e^{i\theta_{j}}\right|\in[0,1]
\end{align}
where $\brm{\theta}=\{\theta_{j}\}_{j\in [n]}$ and $\overline{\mathrm{R}}(\brm{\theta})=1$  indicates \textit{synchronization}. 

While synchronization captures a snapshot of the phase alignment at a given time, it does not reveal the stability of the phase dynamics. This stability is observed when {the angular frequency difference of each coupled oscillator is constant.} Such a phenomenon is called \textit{phase locking}~\cite{Jadbabaie2005_OnStabilityKuramoto, Gupta2014KuramotoSynchronization}. Formally speaking, phase locking in KM is defined as when
\begin{gather}
 |\dot{\theta}_i - \dot{\theta}_j| = 0 \qquad \text{for all} \qquad \{i, j\} \in \mathcal{E}.
\end{gather}
Notice that if the synchronization persists over time, such behavior also leads to the phase locking, i.e., as $t\to\infty$, $|\theta_{i}-\theta_{j}|=0\to {\rm d}/{\rm d}t\,(|\theta_{i}-\theta_{j}|)=0$. Ref.~\cite{Jadbabaie2005_OnStabilityKuramoto} shows that the stable configuration of $\{\dot{\theta}_{i}\}_{i\in[n]}$ arises when the coupling strength $K$ exceeds a critical value determined by the connectivity of the graph and the distribution of natural frequencies $\{\omega_{i}\}_{i\in[n]}$. When $K$ falls below the critical value, the system fails to reach stability. Observing such a condition defines the \textit{no-phase-locking} certification task in KM.

\subsection{Simplicial Kuramoto model (SKM)}

Let $n=|\mathcal{V}|$ be the number or vertices of the vertex set $\mathcal{V}$. Extending the concept of a system of oscillators on a graph $G$ to one that is defined on a simplicial complex $\mathcal{K}_{n}$ (see \Cref{fig:SKM_plus_quantum_tasks}(a) for illustration) leads to phase dynamics and properties beyond KM. In this model, each oscillator is assigned to a simplex of order (or dimension) $k$ (i.e., a group of $(k+1)$ vertices), and the coupling among phases of oscillators are determined by the connectivity of the simplices of the same order in $\mathcal{K}_{n}$. In the literature, such a model is known as a simplicial Kuramoto model (SKM)~\cite{Millan2020ExplosiveComplexes,Ghorbanchian2021Higher-orderSignals,Arnaudon2022ConnectingHodgeSakaguchi,Nurisso2024UnifiedSimplicialKuramoto}.

To give a high-level picture of SKM that generalizes KM in a higher-order network model, we briefly review the concepts of simplicial complexes and the corresponding algebraic structures (for more details, see e.g., \cite{hatcher_algebraic_2001}). An \textit{oriented simplicial complex} $\mathcal{K}_{n}$ on $\mathcal{V}$ is a collection of ordered subsets of $\mathcal{V}$ satisfying the following two properties:
\begin{itemize}
    \item [(i)] for all $\sigma\in\mathcal{K}_{n}$, if $\tau\subset\sigma$ then $\tau\in\mathcal{K}_{n}$, and 
    \item [(ii)] for any $\tau,\sigma\in\mathcal{K}_{n}$, the intersection $\tau\cap\sigma\in\mathcal{K}_{n}$.
\end{itemize}
For each simplex $\sigma=\{i_{0},\cdots,i_{k}\}$, the orientation of $\sigma$ is an equivalence class of orderings $(i_{0},\cdots,i_{k})$ of the vertices, such that two orderings are called coherent if they are the same (either even or odd) permutation parity, otherwise they are anti-coherent.

The connections between $k$-simplices, i.e., $\{\sigma_{k}\in\mathcal{K}_{n}\mid\,|\sigma_{k}|=k+1\}$, and $(k-1)$-simplices are given by the \textit{boundary matrix} $\bk\in\{-1,0,1\}^{n_{k-1}\times n_{k}}$. Here, $n_{k}$ denotes the number of $k$-simplices in $\mathcal{K}_{n}$, and $1$ or $-1$ entry of $\bk$ indicates the coherence or anti-coherence between the orientation of $k$-simplex $\sigma$ and the orientation of any $(k-1)$-simplex $\tau\subset\sigma$. Note that for $k=1$, $\bk$ coincides with the signed incidence matrix $\brm{B}$ of a graph, so this generalizes the connection to the adjacency matrix of the graph in the KM. The boundary $\bk$ act on the vector space $\mathbb{R}^{n_{k}}$ spanned by the bases $\{\sigma_{k}^{i}\}_{i\in[n_{k}]}$ indexed by $k$-simplices.

In SKM, for each $i\in[n_k]$, each oriented $k$-simplex $\sigma_{k}^{i}$ plays the role of an oscillator with phase and natural frequency, called a \textit{simplicial phase}, $\theta_i^{k}\in[-\pi,\pi)$ and \textit{simplicial natural frequency} $\omega_i^{k}\in \mathbb{R}$, respectively. Equivalently, they are the $i$-th elements of the vectors $\brm{\theta}^{k},\brm{\omega}^{k}\in\mathbb{R}^{n_{k}}$, associated to the basis element $\sigma_i^k \in\mathcal{K}_{n}$. This generalization allows us to couple phases and natural frequencies through \emph{lower or upper adjacent simplices}. Here, the lower adjacent $k$-simplices are $k$-simplices that share the common $(k-1)$-simplex, while the upper adjacent $k$-simplices are $k$-simplices that are subsets of the common $(k+1)$-simplex. While the adjacency $\brm{A}$ dictates the interaction in KM dynamics in Eq.~\eqref{eq:node_kuramoto}, the higher-order analogue of this interaction is defined in terms of the adjacency of $k$-simplices through $\bk$ and $\bkk^{\rm T}$. Thus, the SKM dynamics is defined as follows~\cite{Millan2020ExplosiveComplexes}
\begin{gather}
    \begin{split}
    \dot{\brm{\theta}}^{k}
    =& \ \brm{\omega}^{k}
    - K_{k-1}\,\bk^{\rm T}\,\sin\big(\bk\,\brm{\theta}^{k}\big) \\&- K_{k+1}\,\bkk\,\sin\big(\bkk^{\rm T}\,\brm{\theta}^{k}\big).
    \label{eq:simplicial_kuramoto}
\end{split}
\end{gather}
Here, $K_{k-1}$ and $K_{k+1}$ are coupling constants associated with interactions via $(k-1)$‑ and $(k+1)$‑simplices, respectively. As in KM dynamics, the strength of the coupling constants determines the stability of the SKM dynamics, higher values of $K_j$ mean greater correlations between the oscillators. When $k=0$, the second term vanishes and Eq.~\eqref{eq:simplicial_kuramoto} reduces to the pairwise Kuramoto model in Eq.~\eqref{eq:node_kuramoto}, with $\{\theta_{i}^{0},\omega_{i}^{0}\}_{i\in[n]}=\{\theta_{i},\omega_{i}\}_{i\in[n]}$.  

Using the algebraic properties of boundary operator, that is $\bk\bkk=\{0\}^{n_{k-1}\times n_{k+1}}$ (equivalently, $\bkk^{\rm T}\bk^{\rm T}=\{0\}^{n_{k+1}\times n_{k-1}}$), one can decouple the the SKM dynamics in Eq.~\eqref{eq:simplicial_kuramoto} into the dynamics on $(k-1)$‑simplices and the dynamics on $(k +1)$‑simplices. This is done by simply applying $\bk$ or $\bkk^{\rm T}$ to Eq.~\eqref{eq:simplicial_kuramoto}, which eliminates the contribution of one of the last two terms. The upper and lower projected simplicial phase and the upper and lower projected simplicial natural frequency are defined as follows~\cite{Millan2020ExplosiveComplexes,Nurisso2024UnifiedSimplicialKuramoto}
\begin{align}\label{eq:projected_signal}
    \begin{aligned}
        &\brm{\theta}^k_{[-]}=\bk\,\brm{\theta}^k,\quad\brm{\omega}^k_{[-]}=\bk\,\brm{\omega}^k,\\&\brm{\theta}^k_{[+]}=\bkk^{\rm T}\,\brm{\theta}^k,\quad\brm{\omega}^k_{[+]}=\bkk^{\rm T}\,\brm{\omega}^k.
    \end{aligned}
\end{align}
The \textit{projected SKM dynamics} are then given by
\begin{align}\label{eq:projected_SKM_dynamics}
\begin{aligned}
    \dot{\brm{\theta}}^k_{[-]} &= \brm{\omega}^k_{[-]} - K_{[-]}\,\brm{L}_{k-1}^{[+]}\,\sin\bigl(\brm{\theta}^k_{[-]}\bigr),\\
    \dot{\brm{\theta}}^k_{[+]} &= \brm{\omega}^k_{[+]} - K_{[+]}\,\brm{L}_{k+1}^{[-]}\,\sin\bigl(\brm{\theta}^k_{[+]}\bigr),
\end{aligned}
\end{align}
{where $\brm{L}_{k-1}^{[+]}:=\bk^{\rm T}\bk$ and $\brm{L}_{k+1}^{[-]}:=\bkk\bkk^{\rm T}$. Throughout this paper, $[\pm]$ as an index means $[+]$ or $[-]$ index, and $\pm$ indicates that the operations in the equation is either $+$ or $-$ depending on the index we choose.} One can view this projection as connecting the phase and natural frequency configuration and the dynamics on $k$-simplices via the (co-)boundary matrix. We will see in the next section that the projection becomes important when we observe synchronization and NPL phenomena in SKM. 

\subsection{Synchronization Measures and No-Phase-Locking Certification of the SKM}

The synchronization and NPL condition in KM are naturally generalized to the similar phenomena in SKM through the language of discrete exterior calculus~\cite{hiranithesis2003discrete,Millan2020ExplosiveComplexes,Nurisso2024UnifiedSimplicialKuramoto}. Technically speaking, both phenomena in the KM case can be identified from the phase and frequency configuration that belongs to the kernel of $\brm{B}^{\rm T}$. The KM synchronization, i.e., $|\theta_{i}^{0}-\theta_{j}^{0}|=0$ for all $\{i,j\}\in\mathcal{E}$, can be rewritten as $\brm{B}_{1}^{\rm T}\brm{\theta}^{0}=\brm{0}$. As for the NPL condition, it can be formulated $\brm{B}_{1}^{\rm T}\dot{\brm{\theta}}^{0}=\brm{0}$. Since the incidence matrix $\brm{B}_{1}$ naturally extends to higher dimensions as the boundary matrix, one extends the synchronization and NPL definitions to higher dimensions to phases and frequencies associated with simplices in the simplicial complex.

In this language, a topology-aware KM order parameter $\mathrm{R}(\brm{\theta})$ that generalizes \Cref{eq:KM_order_parameter} is defined by~\cite{Jadbabaie2005_OnStabilityKuramoto}
\begin{align}\label{eq:KM_order_parameter_topology_aware}
    \mathrm{R}(\brm{\theta}^{0}):=1-2\left(\frac{|\mathcal{E}|}{n^2}+\frac{1}{n^{2}}\sum_{i\in [|\mathcal{E}|]}\cos{\left(\left[\brm{B}_{1}^{\rm T}\brm{\theta}^{0}\right]_{i}\right)}\right).
\end{align}
Ref~\cite{Arnaudon2022ConnectingHodgeSakaguchi,Nurisso2024UnifiedSimplicialKuramoto} modifies and extends $\mathrm{R}(\brm{\theta}^{0})$ to the SKM case, called a \textit{simplicial order parameter}, that is defined by
\begin{eqnarray}\label{eq:general_order}
    \mathrm{R}(\brm{\theta}^{k})=b_{[-]}\,\mathrm{R}_{[-]}(\brm{\theta}^{k})+b_{[+]}\,\mathrm{R}_{[+]}(\brm{\theta}^{k}),
\end{eqnarray}
where $b_{[\pm]}:=n_{k\pm1}/(n_{k+1}+n_{k-1})$ and 
\begin{gather}\label{eq:general_projected_order}
    \mathrm{R}_{[\pm]}(\brm{\theta}^{k}):=\frac{1}{n_{k\pm1}}\sum_{i\in[n_{k\pm1}]}\cos\left(\left[\brm{\theta}^{k}_{[\pm]} \right]_{i}\right),
\end{gather}
{with $\mathrm{R}_{[-]}(\brm{\theta}^{k})$ and $\mathrm{R}_{[+]}(\brm{\theta}^{k})$} are called the \textit{lower and upper projected order parameters}~\cite{Nurisso2024UnifiedSimplicialKuramoto}, respectively. Note that $\mathrm{R}(\brm{\theta}^{k})$ ranges between $-1$ and $1$, attaining $1$ when the phases lie in the harmonic subspace. 

The condition under which phase-locking can occur in the KM can also be rewritten in terms of discrete exterior calculus. The NPL condition is formulated as a failure of the frequencies to satisfy $|\dot{\theta}_{i}^{0}-\dot{\theta}_{j}^{0}|=0$ for all $\{i,j\}\in\mathcal{E}$ as $\brm{B}_{1}^{\rm T}\dot{\brm{\theta}}^{0}=\brm{0}$. Thus, similar to the extension of synchronization on KM to SKM, one can investigate the phase locking behaviour of SKM based on the projected dynamics in Eq.~\eqref{eq:projected_SKM_dynamics}. A simplicial phase $\brm{\theta}^{k}$ is said to be phase locked from above (below) if $\dot{\brm{\theta}}^{k}_{[+]}=\brm{0}$ ($\dot{\brm{\theta}}^{k}_{[-]}=\brm{0}$)~\cite{Nurisso2024UnifiedSimplicialKuramoto}. The phase-locked solutions are the equilibria of the projected dynamics.

A sufficient condition for the \emph{absence} of equilibria, hence no-phase-locking, is derived by considering the $\ell_2$–projection of $\brm{\omega}^k$ onto the image of $\bk$ and $\bkk^{\rm T}$, which are called gradient and curl subspaces, respectively~\cite{Nurisso2024UnifiedSimplicialKuramoto}. Let $\bullet^+$ be a Moore-Penrose pseudoinverse transformation. Define
\begin{align}\label{eq:no_phase_locking_state}
    \begin{aligned}
        \brm{\omega}^{k-1}_\ast&:=\left(\brm{L}_{k-1}^{[+]}\right)^{+}\,\brm{\omega}^k_{[-]}=\Big(\bk^{\rm T}\Big)^{+}\,\brm{\omega}^k\,,\\
    \brm{\omega}^{k+1}_\ast&:=\left(\brm{L}_{k+1}^{[-]}\right)^{+}\,\brm{\omega}^k_{[+]}=\Big(\bkk\Big)^{+}\,\brm{\omega}^k,
    \end{aligned}
\end{align}
so that $\bk^{\rm T}\brm{\omega}^{k-1}_\ast$ (or $\bkk\,\brm{\omega}^{k+1}_\ast$) is the closest vector to $\brm{\omega}^k$ in the gradient (or curl) subspace. Let $\mathfrak{q}=\{k-1,k+1\}$. If the couplings satisfy 
\begin{eqnarray}\label{eq:no_phase_locking}
    K_{\mathfrak{q}}&<&\frac{1}{\sqrt{n_{\mathfrak{q}}}}\,\left\|\brm{\omega}^{q}_{\ast}\right\|_{2}=:K_{\mathfrak{q}}^{s}
\end{eqnarray}
then the corresponding ~\Cref{eq:projected_SKM_dynamics} admit no equilibria. This regime is called an \textit{NPL regime} \red{or} \textit{NPL condition}. In other words, whether phase‑locked solutions exist depends on the ability of the couplings $K_{\mathfrak{q}}$ to overcome the mismatch between $\brm{\omega}^k$ and its projection onto the gradient or curl subspace. The quantity $K_{\mathfrak{q}}^{(s)}$ is referred to as the \emph{simplicial NPL critical value}. The square of it, $(K_{\mathfrak{q}}^{s})^{2}$, is a higher-order analogue of what is known as the \textit{synchrony alignment function}~\cite{Skardal2014OptimalSynch}, which is important in optimizing synchronization of SKM based on the projected dynamics~\cite{Skardal2021Higher-orderOptimize}. Thus, computing the critical value, a subroutine in our algorithm (see \Cref{sec:quantum_algo_details}), is also beneficial for this optimization task. 
\section{The Problem Definitions and Main Results}\label{sec:quantum_algo_overview}

In this section, we first introduce the general input model of the quantum algorithms we use to solve two computational tasks for SKM. We then formally describe the computational tasks: simplicial order-parameter estimation (Task 1) and NPL certification (Task 2). The results present the resource estimates for the algorithms solving the two tasks in terms of query complexity with respect to the input models defined in \Cref{def:prob-SKM,def:prob-sigma,def:membership-oracle} and the additional quantum gates. The details of the algorithms appear in \Cref{sec:quantum_algo_details}.

\subsection{The Input and Access Models}

We employ a common basis encoding of a $k$-simplex $\sigma_{k}$ in an $n$-qubit register introduced in Ref.~\cite{Lloyd2016QuantumData}. Fix a simplicial complex $\mathcal{K}_{n}$ on a vertex set $\mathcal{V}$. Every $k$-simplex $\sigma_k\subseteq [n]$ is represented on $n$ qubits by a basis state $\ket{\sigma_k}\in\{0,1\}^{\otimes n}$ with Hamming weight $k+1$, where each qubit in state $\ket{1}$ denotes the presence of the corresponding vertex in $\sigma_{k}$. We fix the \emph{canonical orientation} given by lexicographic ordering, i.e., increasing vertex labels; the boundary matrix $\bk$ and coboundary matrix $\bkk^{\rm T}$ are taken with respect to this orientation.

Next, we adopt the framework of \emph{quantum topological signal processing}~\cite{Leditto2023topological,leditto2024quantumhodgeranktopologybasedrank}, in which linear-algebraic transformations of simplicial data are transformed by applying quantum singular value transformation (QSVT)~\cite{Gilyen2019QuantumArithmetics} to \emph{projected unitary encodings} (PUEs) of boundary matrices~\cite{McArdle2022AQubits}. Given simplicial data $\brm{s}^{k}:=\{s^{k}_{i},\sigma_{k}^{i}\}_{i\in[n_{k}]}$, the \emph{simplicial state} $\ket{\brm{s}^{k}}$ is given by
\begin{equation}\label{eq:simplicial-phase-state}
  \ket{\brm{s}^{k}} \;:=\; \frac{1}{\mathcal{N}(\brm{s}^{k})}\sum_{i=1}^{n_k}s_i^{k} \,\ket{\sigma_k^i},
\end{equation}
where we denote $\mathcal{N}(\brm{v}) := \|\brm{v}\|_2 \;=\;\bigl(\sum_{j} v_j^2\bigr)^{1/2}$ for any vector $\brm{v}$. Also, we use $\ket{\brm{\sigma}_k}$  to denote a uniform superposition over $k$-simplices in $\mathcal{K}_n$, that is, a special case of the simplicial phase state with $s_{i}=1$ for all $i\in[n_{k}]$.

In this paper, we assume that the above states are prepared by {what we call} \textit{probabilistic state preparation unitaries} acting on data and ancilla registers, which can prepare the desired state by postselecting the ancilla qubits with some probability of success. We adopt this generic probabilistic state-preparation model to cleanly separate algorithmic complexity from the assumptions of our access model. Explicitly, it covers QRAM/QROM~\cite{Giovannetti2008QRAM,Babbush2018EncodingComplexity} data loading for general datasets deterministically, and our later explicit efficient circuit constructions for synthetic instances, e.g., the explicit construction of \Cref{def:prob-sigma} via Dicke state preparation~\cite{Bartschi2019DeterministicStates} and a simplex membership oracle (described in \Cref{def:membership-oracle}).
\begin{definition}[Probablilistic state preparation unitary]\label{def:prob_state}
Given input data $\{\psi_{i},v_{i}\}_{[i\in N]}$, where $\psi_{i}$ is the $i$-th component of a vector $\brm{\psi}$ in an $N$-dimensional vector space $\mathbb{V}^{N}$ with the corresponding basis $v_{i}$, we define a probabilistic state preparation unitary $\brm{U}_{\rm prep}(\{\psi_{i},v_{i}\}_{[i\in N]})$ such that $\brm{U}_{\rm prep}(\{\psi_{i},v_{i}\}_{[i\in N]})\ket{0}^{\otimes(n+a_{\brm{\psi}})}=\ket{\brm{\Psi}}$, where
\begin{align}\label{eq:prob_state}
    \ket{\brm{\Psi}}:=\frac{1}{\alpha_{\brm{\psi}}}\,\ket{\brm{\psi}}\ket{0}^{\otimes a_{\brm{\psi}}}+\ket{\perp_{\brm{\Psi}}}
\end{align}
and $\ket{\brm{\psi}}:=1/\mathcal{N}(\brm{\psi})\,\big(\sum_{i\in N}\psi_{i}\ket{v_{i}}\big)$ satisfying $(\brm{1}_{n}\otimes\bra{0}^{\otimes a_{\brm{\psi}}})\ket{\perp_{\brm{\Psi}}}=0$ with $\alpha_{\brm{\psi}}> 1$.
\end{definition}
\noindent Postselecting on $\ket{0}^{a_{\brm{\psi}}}$ on the output of this unitary, we obtain $\ket{\brm{\psi}}$ with probability of success $1/\alpha_{\brm{\psi}}^{2}$. The deterministic state preparation unitary is a special case of this unitary with {no ancilla ($a_{\brm{\psi}}=0$), no postselection ($\alpha_{\brm{\psi}}=1$) and no superposition with garbage state $\ket{\perp_{\brm{\Psi}}}$, i.e., $\ket{\brm{\Psi}}=\ket{\brm{\psi}}$.}

Given the simplicial data $\brm{s}^{k}$ of the SKM, either $\brm{s}^{k}=\brm{\theta}^{k}:=\{\theta_{i}^{k},\sigma_k^i\}_{i\in[n_k]}$ or $\brm{s}^{k}=\brm{\omega}^{k}:=\{\omega_i^k,\sigma_k^i\}_{i\in[n_k]}$, and access to $\brm{U}_{\rm prep}$ in \Cref{def:prob_state}, we define a unitary $\brm{U}_{\brm{S}^{k}}=\brm{U}_{\rm prep}(\brm{s}^{k})$ as a subroutine to prepare
\begin{align}
    \ket{\brm{S}^{k}}=\frac{1}{\alpha_{\brm{s}^{k}}}\,\ket{\brm{s}^{k}}\ket{0}^{\otimes a_{\brm{s}^{k}}}+\ket{\perp_{\brm{S}^{k}}}
\end{align}
as in \Cref{eq:prob_state}, where $\ket{\brm{s}^{k}}$ is defined in \Cref{eq:simplicial-phase-state}. To differentiate the use of this unitary and the state it prepares as an input for our algorithms, we denote 
\begin{align}\label{def:prob-SKM} 
    \brm{U}_{\brm{S}^{k}}=
    \begin{cases}
        \UTheta\quad&\text{preparing $\ket{\brm{\Theta}^{k}}$ for $\brm{s}^{k}=\brm{\theta}^{k}$}\\ \UOmega\quad&\text{preparing $\ket{\brm{\Omega}^{k}}$ for $\brm{s}^{k}=\brm{\omega}^{k}$}.
    \end{cases}
\end{align}
We also denote the corresponding rescaling factors and the number of ancilla qubits used as
\begin{align}\label{def:prob_SKM_rescaling_factors}
    (\alpha_{\brm{s}^{k}},a_{\brm{s}^{k}})=
    \begin{cases}
        (\alpha,a)\quad&\text{for $\brm{s}^{k}=\brm{\theta}^{k}$},\\
        (\beta,b)\quad&\text{for $\brm{s}^{k}=\brm{\omega}^{k}$}.
    \end{cases}
\end{align}

In addition to the previous unitary, given only $\{\sigma_{k}^{i}\}_{i\in[n_{k}]}$, we denote by $\brm{U}_{\brm{\Sigma}_k}$ a unitary which acts on $n+a_k$ qubits such that 
\begin{equation}\label{def:prob-sigma}
\brm{U}_{\brm{\Sigma}_k}\ket{0}^{\otimes(n+a_k)}=\frac{1}{\mu_k}\,\ket{\brm{\sigma}_k}\ket{0}^{\otimes a_k} +\ket{\perp_{\Sigma_{k}}},
\end{equation}
where $\ket{\brm{\sigma}_k}:=1/\sqrt{n_{k}}\sum_{i\in[n_{k}]}\ket{\sigma_{k}^{i}}$, {$\mu_{k}>1$,} and $(\brm{1}_{n}\otimes\bra{0}^{\otimes a_k})\ket{\perp_{\Sigma_{k}}}=0$. Although $\USigma$ can be defined as a special case of $\brm{U}_{\rm prep}(\{1,\sigma_{k}^{i}\}_{i\in[n_{k}]})$, we differentiate access to $\brm{U}_{\rm prep}$ and $\USigma$ in our complexity analysis because the preparation of a uniform superposition can be much simpler than general data encoding. This difference becomes evident when we start explicitly constructing such unitaries for synthetic datasets in \Cref{sec:quantum_advantage_regimes}.

We now define another important unitary that plays a key role in constructing the boundary matrix in a quantum circuit~\cite{Lloyd2016QuantumData}. A \textit{simplex membership oracle} {flagging simplices that are in $\mathcal{K}_{n}$} is given by 
\begin{equation}\label{def:membership-oracle}
  \brm{O}_{m_k}\ket{\sigma_k}\ket{a}=\ket{\sigma_k}\ket{a\oplus \mathbb{1}_{\mathcal{K}_n}(\sigma_k)},
\end{equation}
where \(\mathbb{1}_{\mathcal{K}_n}(\sigma_k)=1\) if \(\sigma_k\) is a $k$-simplex in ${\cal K}_{n}$ and is $0$ otherwise. We denote $\Pi_k := \sum_{i\in[n_{k}]}\, \ket{\sigma_k^i}\bra{\sigma_k^i}$ to be the projector onto the subspace of valid $k$-simplices. 

Using $\brm{O}_{m_k}$, one can implement $\bk$ within a unitary with larger size~\cite{berry2023analyzing,McArdle2022AQubits,Kerenidis2022QuantumStates}. Define
\begin{align}\label{eq:controlled_k_projection}
    \mathrm{C}_{\Pi_k}^{\perp}\mathrm{NOT}=\Pi_{k}\otimes\brm{1}+\big(\brm{1}_{n}-\Pi_{k}\big)\otimes\brm{X},  
\end{align}
which can be applied as $\big(\brm{1}_{n}\otimes\brm{X}\big)\brm{O}_{m_{k}}$. This $\mathrm{C}_{\Pi_k}^{\perp}\mathrm{NOT}$ operator along with another unitary called the Dirac operator~\cite{Akhalwaya2022RepresentationOperator,Kerenidis2022QuantumStates} (described explicitly in \Cref{app:embedding_boundary_matrix}), are used to construct a bigger unitary $\brm{U}_{\bk}$ that encodes the boundary matrix $\bk$. Thus, the unitary $\brm{U}_{\bk}$ is a projected unitary encoding (PUE) of $\bk$ defined so that
\begin{align}\label{eq:boundary-pue}
  \left(\bra{0}^{\otimes (a_{\bk}+2)}\otimes\brm{1}_{n}\right)  \, \brm{U}&_{\bk} \, \left(\ket{0}^{\otimes(a_{\bk}+2)}\otimes\brm{1}_{n}\right)= \frac{\bk}{\sqrt{\nu_k}}
\end{align}
for a known positive rescaling factor \(\nu_k\in\mathbb{N}\). The PUE construction with different constructions of $\brm{U}_{\bk}$ and scaling factor $\nu_{k}$ has been explored in the literature. Ref~\cite{berry2023analyzing} constructed $\brm{U}_{\bk}$ with $\nu_k=n^2$ for all $k$, whereas the recent methods given in~\cite{Kerenidis2022QuantumStates,McArdle2022AQubits} improved the scaling factor quadratically, i.e., $\nu_k=n$. Our algorithms and results work equally well with all constructions. 

\subsection{Main Results}

We study two computational problems arising from SKM dynamics and construct quantum algorithms to address them. We present the complexity of the algorithms as our main results. {For concreteness and simplicity, we present the complexities of the algorithms in terms of the query complexities of the probabilistic state preparation unitaries in \Cref{def:prob-SKM,def:prob-sigma}, and we use the PUE method from Ref.~\cite{Kerenidis2022QuantumStates}, which utilizes a membership oracle from Ref.~\cite{Metwalli2021FindingComputer}.} Details of the algorithms and their complexity analyses are provided in \Cref{sec:quantum_algo_details}. To begin with, we simplify the notation in the query complexity statements by denoting oracle access to membership function $\Om$ for $p\in\{k-1,k,k+1\}$ since we query the oracle for different $p$.  

The formal statements of the first task is given as follows. 
\begin{problem}[\textbf{Simplicial order parameter estimation}]\label{prob:order_parameter_estimation}
    Given simplicial data $\{\theta_{i}^{k},\sigma_{k}^{i}\}_{i\in[n_{k}]}$ and $\delta,\varepsilon\in(0,1/2)$, output an estimate $\widehat{\mathrm{R}}(\brm{\theta}^k)$ of the order parameter $\mathrm{R}(\brm{\theta}^k)$ defined in~\eqref{eq:general_order}, such that
  \begin{align}
    \bigl|\widehat{\mathrm{R}}(\brm{\theta}^k)-\mathrm{R}(\brm{\theta}^k)\bigr|\le \varepsilon,
  \end{align}
  with {failure probability $\delta$}. 
\end{problem}

Observe that the simplicial order parameter $\mathrm{R}(\brm{\theta}^{k})$ is a summation of the projected order parameters $\mathrm{R}_{[+]}(\brm{\theta}^{k})$ and $\mathrm{R}_{[-]}(\brm{\theta}^{k})$ defined in \Cref{eq:general_projected_order}. Define the weights $b_{[\pm]}=n_{k\pm1}/(n_{k+1}+n_{k-1})$ and let $\ALU=\mathrm{diag}\big(\big(\brm{\theta}^{k}_{[\pm]}\big)_{0},\cdots,\big(\brm{\theta}^{k}_{[\pm]}\big)_{n_k\pm1}\big)$, where $\brm{\theta}^{k}_{[+]}=\bkk^{\rm T}\brm{\theta}^{k}$ and $\brm{\theta}^{k}_{[-]}=\bk\brm{\theta}^{k}$. To solve \Cref{prob:order_parameter_estimation}, we independently estimate $\mathrm{R}_{[+]}(\brm{\theta}^{k})$ and $\mathrm{R}_{[-]}(\brm{\theta}^{k})$ via trace estimation of matrix $\ALU$ and combine both estimations to obtain an estimator for $\mathrm{R}(\brm{\theta}^{k})$. Here, $\ALU$ is a diagonal matrix where its $i$-th eigenvalue is the $i$-th component of the projected simplicial phase $(\brm{\theta}^{k}_{[\pm]})_{i}$ defined in \Cref{eq:projected_signal} with eigenvector $\ket{\sigma_{k\pm1}^{i}}$. Given a quantum state $\ket{\brm{\sigma}_{k\pm1}}=1/n_{k\pm1}\sum_{i\in[n_{k\pm1}]}\ket{\sigma_{k\pm1}^{i}}$, the trace estimation of $\ALU$ with a reference state $\ket{\brm{\sigma}_{k\pm1}}$ yields an estimator for $\langle\ALU\rangle_{\brm{\sigma}_{k\pm1}}=\mathrm{R}_{[\pm]}(\brm{\theta}^{k})$. Then, the estimations for $\langle\ALU\rangle_{\brm{\sigma}_{k+1}}$ and $\langle\ALU\rangle_{\brm{\sigma}_{k-1}}$ are employed to estimate $b_{[+]}\langle\ALU\rangle_{\brm{\sigma}_{k+1}}+b_{[-]}\langle\ALU\rangle_{\brm{\sigma}_{k-1}}=\mathrm{R}(\brm{\theta}^{k})$. 

At a high level, to implement the trace estimation of $\ALU$, we first encode the simplicial phase $\{\theta_{i}^{k}\}_{i\in[n_{k}]}$ in the amplitudes of the simplicial state defined in \Cref{eq:simplicial-phase-state}. We denote this state $\ket{\brm{\theta}^{k}}$. We use $\UTheta$ defined in \Cref{def:prob-SKM} to create $\ket{\brm{\Theta}^{k}}$ that is a superposition between $\ket{\brm{\theta}^{k}}$ and a garbage state $\ket{\perp_{\brm{\Theta}^{k}}}$. We then apply the PUE of $\bk$ or $\bkk^{\rm T}$ to $\ket{\brm{\Theta}^{k}}$ to transform the amplitudes of $\ket{\brm{\theta}^{k}}$ from $\theta_{i}^{k}$ to $(\brm{\theta}^{k}_{[\pm]})_{i}$. We construct the block encoding of a diagonal matrix $\ALU$, up to a rescaling factor that we detail later. Then, we estimate $\langle\ALU\rangle_{\brm{\sigma}_{k\pm1}}$ with a reference state that is given by $\ket{\brm{\sigma}_{k\pm1}}=1/n_{k\pm1}\sum_{i\in[n_{k\pm1}]}\ket{\sigma_{k\pm1}^{i}}$. The reference state plus the garbage state $\ket{\perp_{\Sigma_{k\pm1}}}$ is prepared by using $\USigmaLU$ defined in \Cref{def:prob-sigma}, and the trace estimation is done via the Hadamard test and amplitude estimation. Lastly, combining both estimations for $\langle\ALU\rangle_{\brm{\sigma}_{k+1}}$ and $\langle\ALU\rangle_{\brm{\sigma}_{k-1}}$ gives the estimation of the simplicial order parameter. The details of the algorithm appears in \Cref{alg:T1}.

The complexity of the algorithm is stated below. We present the proof in \Cref{app:tak1_details}.
\begin{theorem}[Query complexity of the simplicial order parameter estimation]\label{thm:main-T1}
\Cref{alg:T1} solves \Cref{prob:order_parameter_estimation} using (controlled versions of) unitaries $\UTheta$, $\USigmaLU$, and $\Om$. The algorithm calls $\widetilde{O}(r_{1}\,\gamma_{1})$ times to $\UTheta$ and $\Om$, $O(r_{1})$ times to $\USigmaLU$, and uses $\widetilde{O}((n+a)\,r_{1}\,\gamma_{1})$ additional gates, where $\gamma_{1}=\alpha\,\sqrt{n}\,{\cal N}(\brm{\theta}^{k})$ and
\begin{gather*}
r_{1}=O\left(\frac{b_{[-]}\,\mu_{k-1}^{2}+b_{[+]}\,\mu_{k+1}^{2}}{\left(n_{k+1}+n_{k-1}\right)\,\varepsilon}\,\log{\left(\frac{1}{\delta}\right)}\right).
\end{gather*} 
\end{theorem}
\noindent Here, big-$\widetilde{O}$ notation hides additional (poly)logarithmic factors in $\varepsilon,\,\alpha,\,a,\,{\cal N}(\brm{\theta}^{k})$. One can see that this complexity depends on several parameters that come from the input data: $n_{k\pm1},\,\mathcal{N}(\brm{\theta}^{k})$ and the rescaling factors $\alpha,\,\mu_{k\pm1}$. To assess the performance of \Cref{alg:T1}, we specialize the algorithm for specific instances of Task 1 in the next section.

For the second task, the formal statement is given as follows.
\begin{problem}[\textbf{No–phase–locking (NPL) certification}]\label{prob:no-phase-locking_certification}  
    Fix $\mathfrak{q}\in\{k-1,k+1\}$, a coupling constant $K_{\mathfrak{q}}>0$, and a promise gap $\Delta>0$. Given simplicial data $\{\omega_{i}^{k},\sigma_{k}^{i}\}_{i\in[n_{k}]}$ and the promise that $\bigl|K_{\mathfrak{q}}-K_{\mathfrak{q}}^{s}\bigr|\ge\Delta$, decide, with failure probability $\delta\in(0,1/2)$, whether the NPL condition holds, i.e., output
  \begin{align}\label{eq:Task2_output}
      z=\begin{cases}
        1 & \text{if }K_{\mathfrak{q}}<K_{\mathfrak{q}}^{s}\quad(\text{NPL certified}),\\[2pt]
        0 & \text{if }K_{\mathfrak{q}}\ge K_{\mathfrak{q}}^{s}\,\quad(\text{inconclusive}),
     \end{cases}
  \end{align}
  under the stated promise.
\end{problem}

We solve \Cref{prob:no-phase-locking_certification} by implementing norm estimation of $\|\brm{\omega}_{\ast}^{\mathfrak{q}}\|_{2}=K_{\mathfrak{q}}^{s}$ via amplitude estimation and inequality testing between this norm and $K_{\mathfrak{q}}$ via quantum arithmetics. At a high level, we first prepare the normalized version of $\brm{\omega}_{\ast}^{\mathfrak{q}}$ defined in \Cref{eq:no_phase_locking_state} as a quantum state $\ket{\brm{\omega}_{\ast}^{k\pm1}}$. This is done by constructing an approximate block encoding of $\big(\brm{L}_{k+1}^{[-]}\big)\bkk^{\rm T}$ or $\big(\brm{L}_{k-1}^{[+]}\big)\bk$ via QSVT, and apply the constructed block encoding to $\ket{\brm{\omega}^{k}}$ plus some garbage state $\ket{\perp_{{\Omega}_{k}}}$ as defined in \Cref{def:prob-sigma}. After the application of the block encoding, the prefactor of the $\ket{\brm{\omega}^{k}}$ encodes the norm value $K_{\mathfrak{q}}^{s}$. We then use amplitude estimation to extract such information and encode it into a quantum register. The estimated value of $K_{\mathfrak{q}}^{s}$ is then compared against $K_{\mathfrak{q}}$, encoded in another quantum register via quantum arithmetics, for inequality testing. The inequality testing circuit encodes a single bit value, either $1$ or $0$, in a single qubit register. Such a value tells us whether $K_{\mathfrak{q}}$ is less or greater than $K_{\mathfrak{q}}$ with a gap assumption $\Delta$, i.e., $|K_{\mathfrak{q}}^{s}-K_{\mathfrak{q}}|\ge\Delta>0$. Thus, we achieve the NPL certification task by revealing the bit value through a measurement of the single qubit register. The detailed procedure appears in \Cref{alg:T2}. 

We now present the complexity of the algorithm as follows. We provide the proof in \Cref{app:task2_details}. Let $1/\kappa_{k}\in(0,\sqrt{n}/\zeta^{(k)}_{\min})$ where $\zeta^{(k)}_{\min}$ is the smallest non-zero singular value of the boundary matrix $\bk$.
\begin{theorem}[Query complexity of the NPL certification]\label{thm:no-phase-locking}
Fix $\Delta\in O(1)$, \(\delta\in(0,1/2)\), and ${\mathfrak{q}}\in\{k-1,k+1\}$. Define ${\mathfrak{m}}=k$ if ${\mathfrak{q}}=k-1$ and ${\mathfrak{m}}=k+1$ if ${\mathfrak{q}}=k+1$. For each $\mathfrak{q}$, \Cref{alg:T2} solves \Cref{prob:no-phase-locking_certification} using (controlled applications of) unitaries $\UOmega$ and $\Om$. The algorithm calls $\widetilde{O}(r_{\mathfrak{m}})$ times to $\UOmega$, $\widetilde{O}(r_{\mathfrak{m}}\,d_{\mathfrak{m}})$ times to $\Om$, and uses $\widetilde{O}(n\,r_{\mathfrak{m}}\,d_{\mathfrak{m}})$ additional gates, where
\begin{gather*}
r_{\mathfrak{m}}=O\left(
\frac{\kappa_{\mathfrak{m}}^{2}\,\beta^{2}\,\mathcal{N}(\brm{\omega}^{k})^{2}}{n\,n_{\mathfrak{q}}\,\Delta^{2}}\,
\log\left(\frac{1}{\delta}\right)\right),\,\text{and}\\
d_{\mathfrak{m}}=O\left(\kappa_{\mathfrak{m}}^{2}\log{\left(\frac{\mathcal{N}(\brm{\omega}^{k})}{\min\{\Delta^{2}/(6\,K_{\mathfrak{q}}^{s}),\Delta/12\}}\left(\frac{n}{n_{\mathfrak{q}}}\right)^{1/2}\right)}\right).
\end{gather*}
\end{theorem}
\noindent In the above theorem, the big-$\widetilde{O}$ notation hides additional (poly)logarithmic factors in $r_{\mathfrak{m}}$. The details of the algorithm and proof are also given in \Cref{sec:quantum_algo_details}. Similar to the previous result, the complexity of \Cref{alg:T2} depends on the input data $n_{\mathfrak{q}},\,\mathcal{N}(\brm{\omega}^{k})$ and the rescaling factors $\beta$, with additional dependency on the boundary matrix condition number $\kappa_{\mathfrak{m}}$. Thus, we evaluate the performance of this algorithm by considering specific instances of Task 2 in the next section. 
\section{Quantum Advantage Regimes}\label{sec:quantum_advantage_regimes}

To identify the regime of quantum advantage, we benchmark our algorithm by comparing the number of quantum gates, including the gates used in constructing the oracles, with the number of classical arithmetic operations required to (i) compute the simplicial order parameter $\mathrm{R}(\boldsymbol{\theta}^{(k)})$ and (ii) certify NPL, i.e., checking whether the coupling strength satisfies $K_{k\pm1} < K_{k\pm1}^{s}$. We define those numbers as the \textit{quantum cost} and \textit{classical cost}, respectively. For Task 1, the classical algorithm we use as a benchmark is exact sparse matrix-vector multiplication, while for Task 2, we use an iterative linear equation solver~\cite{Sachdeva2014FasterAlgorithm,Musco2024StableLanczos}. While comparing classical and quantum operations is a somewhat crude comparison, it does allow us to identify a promising avenue for speeding up classically difficult computations as quantum computers develop. Some literature in QTDA takes this framework to investigate the resource separation between quantum and the best classical algorithm~\cite{McArdle2022AQubits,berry2023analyzing,Hayakawa2022QuantumAnalysis}.

\begin{figure*}
    \centering
    \includegraphics[width=0.85\linewidth]{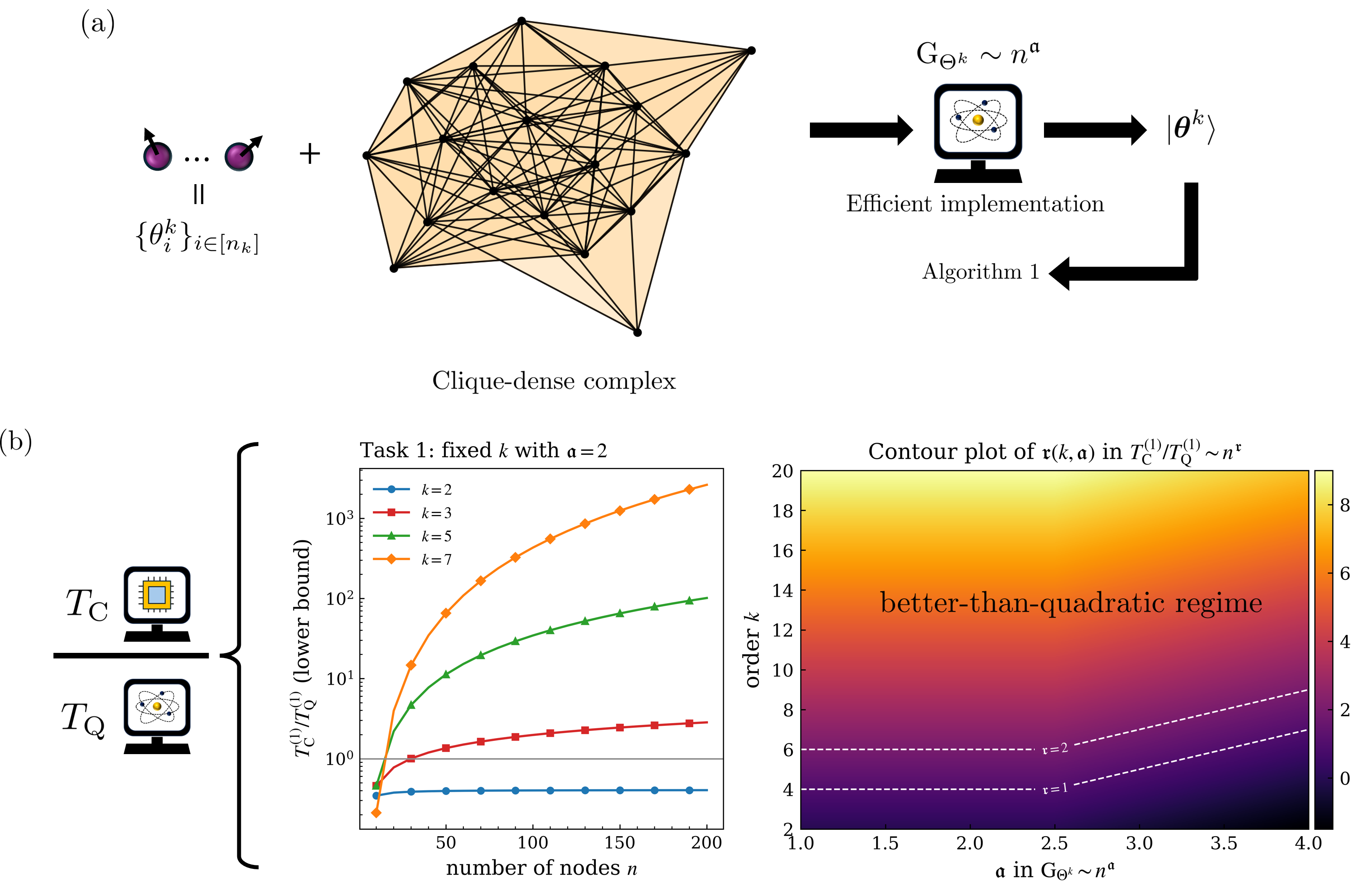}
    \caption{The quantum advantage regime for \Cref{prob:order_parameter_estimation}. (a) The instance of \Cref{alg:T1} is $\{\theta^{k}_{i}\}_{i\in[n_{k}]}$ (at a given time) defined on $k$-simplices in a clique-dense complex $\mathcal{K}_{n}(G_{p_{k}})$. We assume an efficient probabilistic state preparation unitary $\UTheta$ with gate complexity $\mathrm{G}_{\brm{\Theta}^{k}}\sim n^{\mathfrak{a}}$ to prepare the input state $\ket{\brm{\theta}^{k}}$ of the algorithm. (b) The classical-quantum cost ratio for Task 1, $T_{\mathrm{C}}^{(1)}/T_{\mathrm{Q}}^{(1)}$. The advantage regime is the regime in which the ratio is above $1$. The first two plots depict the behaviour of the ratio given $\mathrm{G}_{\brm{\theta}^{k}}\sim n^{2}$ for (i) fixed $k$ and (ii) fixed scaling $\rho=k/n$. The heat map shows the advantage gained in terms of the exponent $\mathfrak{r}(k,\mathfrak{a})$ of the ratio $T_{\mathrm{C}}^{(1)}/T_{\mathrm{Q}}^{(1)}\sim n^{\mathfrak{r}}$, given the gate complexity $n^{\mathfrak{a}}$. For Task 1, the quadratic advantage regime is achieved only for $k\ge6$.}
    \label{fig:Task1_advantage_instance_analysis}
\end{figure*}

For Task 1, we denote the quantum and classical costs by $T_{\rm Q}^{(1)}$ and $T_{\rm C}^{(1)}$, and for Task 2, by $T_{\rm Q}^{(2)}$ and $T_{\rm C}^{(2)}$. The quantum advantage regime is defined as the region where the corresponding classical–quantum cost ratios $T_{\rm C}^{(1)}/T_{\rm Q}^{(1)}$ and $T_{\rm C}^{(2)}/T_{\rm Q}^{(2)}$ greater than one for sufficiently large $n$. It should be emphasized that, while an approximate classical linear solver is used as a benchmark for \Cref{alg:T2}, we are not aware of any approximate classical method that exists for solving \Cref{prob:order_parameter_estimation}. Consequently, our quantum advantage analysis for Task 1 is restricted to the exact computation of sparse matrix-vector multiplication.

We now summarize the parameter dependencies in the quantum costs for the two tasks. As shown in \Cref{thm:main-T1,thm:no-phase-locking}, the overall costs are dominated by the query complexities of the oracles $\UTheta$, $\UOmega$, $\USigma$, and $\brm{O}_{m_{k}}$, defined in \Cref{def:prob-SKM,def:prob-sigma,def:membership-oracle}. In addition to the size parameters $(n, n_{k\pm1})$ and the norms of the input data $(\mathcal{N}(\brm{\theta}^{k}),\mathcal{N}(\brm{\omega}^{k}))$, these query complexities depend on the rescaling factors $(\alpha, \beta, \mu_{k\pm1})$ introduced in \Cref{def:prob_SKM_rescaling_factors,def:prob-sigma}, as well as on the condition number $\kappa_m$ of the boundary operator $\brm{B}_{\mathfrak{m}}$. These quantities, in turn, are determined by the data characteristics and the topology of the underlying simplicial complex. 

For both tasks, we consider the SKM dynamics on a clique complex. A clique complex $\mathcal{K}_{n}(G)$ is defined as a simplicial complex where each $k$-simplex corresponds to the $(k+1)$-clique of the underlying graph $G$. In the QTDA literature~\cite{Gyurik2022TowardsAnalysis,Hayakawa2022QuantumAnalysis,berry2023analyzing,McArdle2022AQubits}, this type of simplicial complex allows an efficient construction of the PUE of $\bk$. In particular, we can construct $\Om$ for any $p$ with only $O(|\mathcal{E}|)$ or $O(|\mathcal{E}^{c}|)$ (where $|\mathcal{E}^{c}|$ is the number of the missing edges in the graph) gates. Since \Cref{alg:T1,alg:T2} rely heavily on the query complexity of $\Om$, this implementation is crucial to showcase the advantage regime. See the details in \Cref{app:proof_quantum_advantage_regimes}.

To optimize the classical–quantum cost ratio for both tasks, we provide end-to-end gate counts for oracle construction by examining problem instances that minimize parameter dependencies. Here, we specify the instances by restricting the simplicial phase/frequency data and the structure of the simplicial complex for the inputs to Tasks 1 and 2. Recall that the simplicial data for Task 1 is the simplicial phases, and for Task 2, it is the natural frequencies. We impose specific conditions on the phase and natural frequency values, as well as on the clique complex to which they are assigned. We explain such instance specifications below.

For Task 1, we consider a simplicial phase $\{\theta^{k}_{i}\}_{i\in[n_{k}]}$ that enables a construction of a probabilistic state preparation unitary $\UTheta$ with the total number of gates $\mathrm{G}_{\brm{\Theta}^{k}}$ and $\alpha,a= O(1)$. Let $p_{k}=n_{k}/\binom{n}{k+1}\in[0,1]$ be the \textit{clique density} and the $G_{p_{k}}$ is the underlying graph of $\mathcal{K}_{n}(G_{p_{k}})$. Also, we focus on a clique-dense complexes $\mathcal{K}_{n}(G_{c})$, where the underlying graph $G_{c}$ has a clique density $p_{k}:=n_{k}/\binom{n}{k+1}=c\in\Theta(1)$ (see \Cref{fig:Task1_advantage_instance_analysis}(a) for an example of $\mathcal{K}_{18}(G_{0.75})$). One can also relax the density assumption to any slowly decreasing clique density. For simplicity in our calculations, we assume a constant density. Thus, focusing on clique-dense complexes yields that $(b_{[+]}\,\mu_{k+1}^{2}+b_{[-]}\,\mu_{k-1}^{2})/(n_{k+1}+n_{k-1})= O(1)$ for large $n$ small $k$. As a result, we arrive at this classical-quantum cost ratio under these assumptions. We present the detailed proof in \Cref{app:proof_quantum_advantage_regimes}.
\begin{corollary}[Classical-quantum cost ratio for order parameter estimation]\label{cor:order_parameter_est_clique}
    Consider a clique-dense complex $\mathcal{K}_{n}(G_{c})$ and a probabilistic state preparation unitary $\UTheta$ constructed with the total number of gates $\mathrm{G}_{\brm{\Theta}^{k}}$ and parameters $\alpha,a= O(1)$. For all $n$ and $k$, the end-to-end gate complexity of \Cref{alg:T1} yields a classical-quantum cost ratio (expressed in floating-point-operations unit/gates) that is given by
    \begin{align*}
        \frac{T_{\rm C}^{(1)}}{T_{\rm Q}^{(1)}}=\Omega\left(\frac{\left(n\,\binom{n}{k+1}\right)^{1/2}}{\mathrm{G}_{\brm{\Theta}^{k}}+n^{2}}\right).
    \end{align*}
\end{corollary}
\noindent {Since the $T_{\rm C}^{(1)}=\Omega\big(n\binom{n}{k+1}\big)$ based on the complexity of sparse matrix vector multiplication (see \Cref{app:proof_quantum_advantage_regimes}), the quantum cost $T_{\rm Q}^{(1)}$ only reduces the classical cost sub-quadratically.} We also note that since the gate complexity of $\brm{U}_{\brm{\Theta}^{k}}$ may dominate the denominator of the ratio,  more expensive implementations of $\UTheta$ erode advantage more aggressively. Practically, to optimize $\mathrm{G}_{\brm{\Theta}^{k}}$, one can trade gate complexity for ancilla qubits used in state preparation~\cite{Low2024tradingtgatesdirty}. 

Here, we estimate how efficient $\UTheta$ is required to be to achieve the quantum advantage, i.e., the lower bound of the ratio greater than one. While in terms of classical cost, the advantage is only sub-quadratic, we are particularly interested in the regime where the advantage is greater than polynomial in the number of nodes $n$. This consideration is a common practice for assessing the complexity of quantum algorithms in higher-order networks, including Betti number estimation~\cite{Lloyd2016QuantumData,Gyurik2022TowardsAnalysis,berry2023analyzing,McArdle2022AQubits,song2024quantum,Hayakawa2022QuantumAnalysis,hayakawa2024quantumwalkssimplicialcomplexes,Leditto2023topological,leditto2024quantumhodgeranktopologybasedrank}. 

{Let $L_1$ denote the asymptotic lower bound of $T_{\rm C}^{(1)}/T_{\rm Q}^{(1)}$, and $L_1>1$ certify an advantage
regime. Suppose $\mathrm{G}_{\Theta^k}=O(n^{\mathfrak{a}})$ and $L_1=\Omega\big(n^{\mathfrak{r}(k,\mathfrak{a})}\big)$ for fixed $k$. Then, in the large $n$ and small $k$ regime, we obtain
\begin{equation}\label{eq:r_corrected}
\mathfrak{r}(k,\mathfrak{a})=\frac{k+2}{2}-\max(\mathfrak{a},2).
\end{equation}
The exponent $\mathfrak{r}$ characterizes the minimum scaling advantage in $n$ attainable for a given state preparation cost.} As an example, one can see, from the left panel of \Cref{fig:Task1_advantage_instance_analysis}(b), the asymptotic behaviour of $L_1$ with $\mathfrak{a}=2$ in \Cref{eq:r_corrected} for various $k$. In addition, the right panel in \Cref{fig:Task1_advantage_instance_analysis}(b) shows the regime where quadratic advantage in $n$ can be achieved for $k\ge6$. 

\begin{figure*}
    \centering
    \includegraphics[width=1\linewidth]{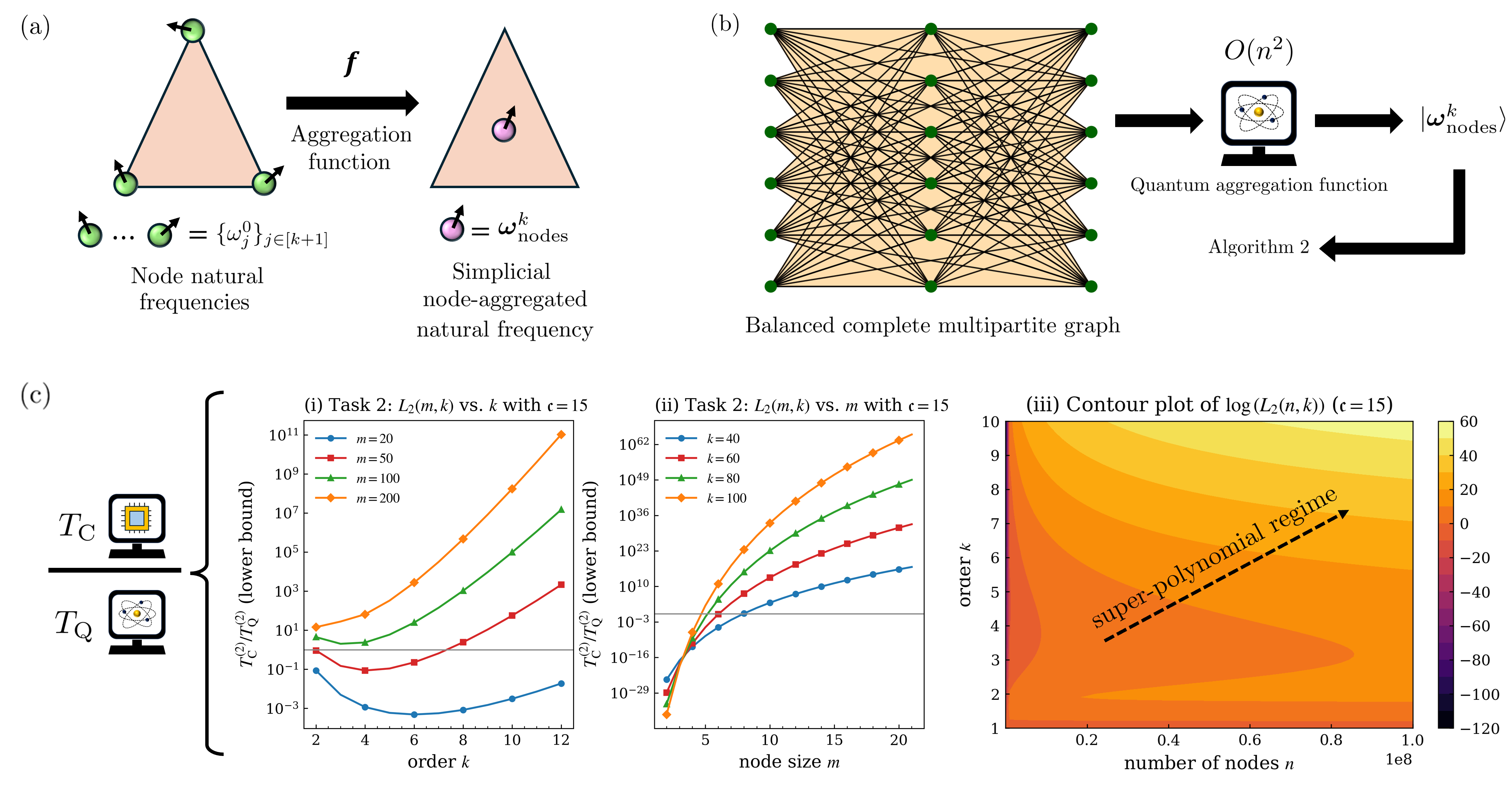}
    \caption{The quantum advantage regime for \Cref{prob:no-phase-locking_certification}. The instance of \Cref{alg:T2} is a simplicial node aggregated frequency $\brm{\omega}^{k}_{\rm nodes}$ (purple dot) that aggregates data from node frequencies $\{\omega^{0}_{j}\}_{j\in[k+1]}$ (green dots)  and balanced complete multipartite graph $\mathcal{K}_{n}(G_{m,k})$. (a) The depiction of an aggregation function that aggregates node natural frequencies to simplicial natural frequency. (b) Given the natural frequencies defined on the nodes of $\mathcal{K}_{n}(G_{m,k})$, there exists an efficient quantum circuit implementation of the aggregation function with the QFT subroutines that costs $O(k^{\mathfrak{c}})$ gates (see details in \Cref{alg:node-aggregated-state}). Here, $\mathfrak{c}= O(1)$. The output of this algorithm is an input of \Cref{alg:T2}. (c) We set $\mathfrak{c}=15$ as a representative constant for visualization. The smaller $\mathfrak{c}$ strengthens the advantage, and the qualitative regime boundaries are not sensitive to the exact constant value. In this example, the algorithm achieves advantage regimes (i.e., regimes where the plots lie above $L_2=1$, that is given as horizontal grey lines) (i) for $k>2$ as $m$ increases, or (ii) $m>2$ as $k$ increases. In other words, the algorithm delivers an advantage for large $(m,k)$. Importantly, it yields a superpolynomial advantage in $n$ for $k\in\Theta(\log{n})$.}
    \label{fig:Task2_advantage_instance_analysis}
\end{figure*}

To showcase the advantage in solving Task 2, we introduce the specific SKM input data types used in our instances. This is a practical data model that captures how simplicial data emerge from local node data. Such an approach is commonly employed to infer higher-order data in sensor networks~\cite{Rajagopalan2006DataAggregation} and define variants of the SKM~\cite{Lucas2020MultiorderNetworks,Skardal2020HOI}. Define $\brm{\omega}^{k}_{\mathrm{nodes}}$ be the \textit{simplicial node-aggregated natural frequency} where the $i$-th component $(\brm{\omega}^{k}_{\mathrm{nodes}})_{i}$ is a function of its constituent node natural frequencies $\{\omega^{0}_{i,j}\}_{j\in[k+1]}$, for all $i\in[n_{k}]$. Explicitly, such a component is given by
\begin{align}\label{def:simplicial node data}
    (\brm{\omega}^{k}_{\mathrm{nodes}})_{i}=f(\omega^{0}_{i,0},\cdots,\omega^{0}_{i,k}),
\end{align}
where $f:\mathbb{R}^{k+1}\to\mathbb{R}$ is a multivariable function called the \textit{aggregation function} (see \Cref{fig:Task2_advantage_instance_analysis}(a)). {We then require that this aggregation function: (i) implementable via quantum Fourier transform (QFT)-based arithmetics~\cite{Cuccaro2004ACircuit} with $O(k^{\mathfrak{c}})$ gates, where $\mathfrak{c}= O(1)$, (ii) has a bounded value, and (iii) satisfies the Lipschitz condition (see \Cref{app:node-aggregated_KM}).} This data type admits efficient construction of the unitaries $\UTheta$ and $\UOmega$ using only $\mathrm{poly}(n)$ gates. Full details appear in \Cref{alg:node-aggregated-state}.

Next, we consider the case of complete multipartite clique complexes $\mathcal{K}_{n}(G_{m,k})$ given in Ref.~\cite{berry2023analyzing}, where a clique complex is built from a complete $(k+1)$-partite graph $G_{m,k}$ with $m$ isolated nodes in each part  (see a simple example for $\mathcal{K}_{18}(G_{6,2})$ in \Cref{fig:Task2_advantage_instance_analysis}(b)). Thus, we have $n(m,k)=m(k+1)$ and $n_{k}(m,k)=m^{k+1}$. In this case, the classical-quantum ratio cost of \Cref{alg:T2} is given as follows. The proof is given in \Cref{app:proof_quantum_advantage_regimes}.
\begin{corollary}[Classical-quantum cost ratio for NPL certification]\label{cor:no-phase-locking_complete_multipartite_graphs}
    Consider balanced complete multipartite graphs $\mathcal{K}_{n}(G_{m,k})$, simplicial node-aggregated natural frequency input $\brm{\omega}^{k}_{\rm nodes}$ as in \Cref{def:simplicial node data}, and a constant $\mathfrak{c}= O(1)$. The end-to-end gate complexity of \Cref{alg:T2} yields a classical-quantum cost ratio (expressed in floating-point-operations unit/gates) that is given by
    \begin{align*}
        \frac{T_{\rm C}^{(2)}}{T_{\rm Q}^{(2)}}=\Omega\left(\frac{k^{2}\,(m/e)^{k}}{m\log{m}}\,\left(1+\frac{k^{(\mathfrak{c}-1)}}{m^{2}}\right)^{-1}\right).
    \end{align*}
\end{corollary}
\noindent As mentioned before, the value of the constant $\mathfrak{c}$ depends on the aggregated function we use. For example, this value may be big if the QFT arithmetics require many multiplications to approximate the function. 

Here, the asymptotic lower bound of the classical-quantum cost ratio is denoted by $L_{2}(m,k)$ and the quantum regime is justified when its value is greater than $1$. Unlike $(n,k)$, which are often presented as dependent parameters to demonstrate the advantage regime in the QTDA literature~\cite{berry2023analyzing}, both $m$ and $k$ are independent parameters. From \Cref{cor:no-phase-locking_complete_multipartite_graphs}, we have
\begin{align}
    L_{2}(m,k)=\underbrace{\left(\frac{k^{2}}{e^{k}}\;\frac{m^{k-1}}{\log m}\right)}_{:=\,g(m,k)}\;\underbrace{\left(1+\frac{k^{\mathfrak c-1}}{m^{2}}\right)^{-1}}_{:=\,p(m,k)}.
\end{align}
We denote by $g(m,k)$ the growth term and by $p(m,k)$ the penalty term. We analyze the asymptotic behaviours of $g(m,k)$ and $p(m,k)$ to better understand where the advantage threshold lies.

We study $L_2(m,k) = g(m,k)\,p(m,k)$ by separating the growth and penalty terms. For fixed $m \ge 3$, the function $g(m,k)$ grows rapidly with $k$ (dominated by the factor $m^{k-1}/e^{k}$), while for fixed $k>1$, $g(m,k)$ increases with $m$.  Conversely, $p(m,k) = \left(1 + k^{c-1}/m^{2}\right)^{-1} \in (0,1]$, which approaches $1$ when $m \gg k^{(c-1)/2}$ and decays when $m \ll k^{(c-1)/2}$.  The advantage threshold $L_2(m,k) \ge 1$ follows from the trade-off between these competing effects.

To see where the superpolynomial growth appears in $L_2$, it is more informative if we re-express the lower bound in terms of $(n,k)$ and consider the $k=\Theta(\log{n})$ case. In this case,  $g(m,k)=g(n)=n^{\Theta(\log{n})}\,(\mathrm{poly}(n))^{-1}$ and the penalty factor $p(m,k)=p(n)=1-o(1)$ as $n\to\infty$. Thus, the ratio grows as
\begin{align}
    L_2(m,k)=L_2(n)=O(n^{\Theta(\log{n})}\,(\mathrm{poly}(n))^{-1}),
\end{align}
i.e., super-polynomial in $n$. This is the precise sense in which the $k=\Theta(\log n)$ regime aligns with the super-polynomial advantage landscape highlighted in QTDA analyses~\cite{berry2023analyzing}. In \Cref{fig:Task2_advantage_instance_analysis}(c), we plot $\log{L_{2}(n,k)}$, where the region with the positive value of $\log{L_{2}(n,k)}$ represents the advantage regime. 
\section{Discussion}\label{sec:discussion}

We have presented quantum algorithms to solve synchronization estimation (\Cref{prob:order_parameter_estimation}) and no-phase-locking certification (\Cref{prob:no-phase-locking_certification}) tasks that are essential to analyzing the instantaneous and long-time behaviour of simplicial Kuramoto model dynamics.  The algorithms exhibit an advantage in higher-order interaction regimes of SKM, i.e., $k > 2$, provided that the simplicial data encoding can be implemented efficiently with a polynomial number of gates, for example, by restricting to simplicial natural frequency data aggregated from the frequencies of their nodes. To make the advantage regimes concrete, we focus on SKM instances on clique-dense complexes (Task 1) and balanced complete multipartite graphs (Task 2), since these families admit efficient implementations of the simplex membership oracle at the gate level. Under our input model, Task 1 yields at most polynomial improvements in $n$ for fixed interaction order $k$. In contrast, for Task 2 on balanced multipartite instances with $k=\Theta(\log{n})$, we obtain a super-polynomial separation in $n$ between the classical and quantum resources required. As for Task 2, this advantage regime is particularly important in higher-order networks where the lower-order interactions (compared to the number of nodes) govern the dynamics of the SKM, such as those found in brain network activity~\cite{Haruna2016HodgeNetworks,Parastesh2022SynchronizationIH,Majhi2025PatternsON,Reimann2017Cliques,Sizemore2019TheNeuroscientist}.

We anticipate that the algorithms could be further utilized to tackle other computational tasks. For example, one can use \Cref{alg:T1} as a subroutine for a larger algorithm. For example, a recent investigation~\cite{chiba2025reservoircomputingkuramotomodel} formulates the output of reservoir computing as a linear combination of Kuramoto order parameters and uses these order parameters to characterize the reservoir’s approximation performance. An extension to SKM-based reservoirs, in principle, requires evaluation of simplicial order parameters. In that case, one can use \Cref{alg:T1} to accelerate the evaluation step and do a similar analysis. As for \Cref{alg:T2}, one can use this algorithm to assess the stability of dynamics in SKM under external driving factors. A canonical example of this class is the simplicial Sakaguchi–Kuramoto model~\cite{Arnaudon2022ConnectingHodgeSakaguchi}, which is a higher-order generalization of the Sakaguchi–Kuramoto model~\cite{SakaguchiKuramoto1986}, that introduces both linear and nonlinear frustration parameters in the SKM dynamics across simplices of various orders. This model has been shown to be useful for controlling the equilibrium of projected dynamics~\cite{Nurisso2024UnifiedSimplicialKuramoto}. In addition, since one can linearize the SKM dynamics and identify it as a diffusion process on simplicial complexes~\cite{Torres2020SimplicialDynamics}, the algorithms can be extended to quantum algorithms that simulate the same process on quantum computers by solving the governing dynamical equation. This line of research aligns closely with the quantum PDE solvers in Refs.~\cite{Liu2021QAlgoDissipative,tanaka2025polynomialtimequantumalgorithm}.

Lastly, another interesting follow-up direction is to place the NPL certification task in the exponential query separation framework of Ref.~\cite{BabbushEtAlPRX23}, which establishes an exponential quantum advantage for coupled oscillator networks via a graph-Laplacian-based dynamical construction, along with oracle lower bounds for the kinetic energy estimation task. Since graph Laplacians are the lowest-order instance of the Hodge-Laplacian, and since solving linear equations in $2$-Hodge Laplacians already captures the hardness of general linear systems~\cite{Ding2022HardnessGadgets}, a natural next step is to ask whether the hard oscillator network instances of Ref.~\cite{BabbushEtAlPRX23} can be embedded into the Hodge-Laplacian linear systems underlying NPL certification. In contrast to the Hamiltonian simulation route of Ref.~\cite{BabbushEtAlPRX23}, the objective here would be to analyze the problem in the matrix inversion task of the quantum linear systems problem~\cite{Harrow2009QuantumEquations,dervovic2018quantumlinearsystemsalgorithms}. Establishing such a reduction would provide the complexity-theoretic framework for \Cref{alg:T2}, suggesting that the NPL certification task may also inherit the same kind of provable quantum advantage previously identified for graph-based oscillator networks.

\section*{Acknowledgements}
C.M.G.L. is supported by the Monash Graduate Scholarship and CSIRO Data 61 Top-Up Scholarship. K.M. acknowledges MOE Kick-Starter grant SKI 2021\_07\_02. 

\bibliography{main}
\clearpage
\appendix
\onecolumngrid
\section{The Quantum Algorithm Details}\label{sec:quantum_algo_details}
\begin{figure*}[ht]
  \centering
  \fbox{
    \parbox{0.49\textwidth}{
\begin{algorithm}[H]
\caption{\textbf{for \Cref{prob:order_parameter_estimation}} \\ Quantum Order Parameter Estimation}
\label{alg:T1}
\renewcommand{\algorithmicrequire}{\textbf{Input:}} 

\renewcommand{\algorithmicensure}{\textbf{Output:}} 

\begin{algorithmic}[1]
\Require\Statex\begin{itemize}[leftmargin=*, nosep]
    \item Simplicial data $\{\theta^{k}_{i},\sigma_{k}^{i}\}_{i\in[n_{k}]}$ 
    \item Accuracy parameters \(\varepsilon\in(0,1/2)\)
    \item Failure probability \(\delta\in(0,1/2)\)
    \item Oracle access to controlled application of \(\UTheta\) preparing \(\ket{\brm{\Theta}^k}\) in \Cref{def:prob-SKM}
    \item Oracle access to controlled application of \(\USigmaLU\) preparing $\ket{\brm{\sigma}_{k\pm1}}$ in \Cref{def:prob-sigma}
    \item Oracle access to controlled application of \(\Om\) preparing PUE of $\bk$ and $\bkk$ \Cref{eq:boundary-pue}
    \item Parameters: \(n\), \(n_{k\pm1}\), \(\alpha\), \(\mu_{k\pm1}\), \(\mathcal{N}(\brm{\theta}^k)\)
\end{itemize}
\Ensure
An estimate \(\widehat{\mathrm{R}}(\brm{\theta}^{k})\) in $\widehat{\mathrm{R}}(\brm{\theta}^{k})$, satisfying \(|\widehat{\mathrm{R}}(\brm{\theta}^{k}) - \mathrm{R}(\brm{\theta}^{k})| \le \varepsilon\) with probability of success \(\ge 1-\delta\)

\Statex
\Statex\textbf{Subroutine:} \textsc{QuantumOPE}
  \State Prepare state \(\ket{\brm{\Theta}^k}\) \label{line:prep-theta} 
  \State Apply PUE of \(\bk\) and \(\bkk\) to obtain $\ket{\brm{\theta}^{k}_{[\pm]}}$ in \Cref{eq:projected_simplicial_phase_state} with normalization factor \(\gamma_{1}=\alpha\,\sqrt{n}\,{\cal N}(\brm{\theta}^{k})\) \label{line:pues}
  \State Construct block encoding \(\UALU\) of a diagonal matrix \(\ALU/\gamma_{1}\) (described in Eq.~\eqref{eq:diag_matrix_simplicial_signal})\label{line:block-alu}
  \State Construct QSVT $\brm{U}_{q_{[\pm]}}$ of \(\UALU\) 
  with polynomial \(q(x)\) of degree $d$ (given in Eq.~\eqref{eq:degree_cos_poly}) that $\varepsilon_{\cos}$-approximates \(\cos(\ALU)\) with $\varepsilon_{\cos}\in(0,1/2)$  \label{line:qsvt-cos}
  \State Prepare reference state \(\ket{\brm{\Sigma}_{k\pm1}}\) \label{line:prep-sigma}
  \State Execute Hadamard test between \(\ket{\brm{\Sigma_{k\pm1}}}\) and \(\brm{U}_{q_{[\pm]}}\ket{\brm{\Sigma_{k\pm1}}}\) to encode success probability $p^{\star}_{[\pm]}$ in \Cref{eq:probability_Hadamard_test} in a single ancilla measurement outcome $z$ \label{line:hadamard-test}
  \State Measure ancilla qubit
  \State \Return measurement bit \(z\in\{0,1\}\) 

\Statex
\Statex \textbf{Main Algorithm:}
\State Set error budget for QSVT polynomial $q(x)$ {used in \textsc{QuantumOPE}}: $\varepsilon_{\cos}\gets\varepsilon/4$
\State Set error budget for Hadamard test: $\varepsilon_{\rm Had}\gets1/(b_{[\pm]}\mu_{k\pm1}^{2})\,\varepsilon/4$
\State Set repetitions: \(r_{1} \gets O(b_{[+]}\,\mu_{k+1}^{2}+b_{[-]}\,\mu_{k-1}^{2})\log(\delta^{-1})\,\varepsilon^{-1}/(n_{k+1}+n_{k-1}))\)
\State Run \textsc{QuantumOPE} for both $[+]$ and $[-]$ independently 
\State Run amplitude estimation with \textsc{QuantumOPE} using \(r_{[\pm]}\) queries to obtain estimate \(\widehat{p}^\star_{[\pm]}\) satisfying \(|\widehat{p}^\star_{[\pm]} - p^\star_{[\pm]}| \le \varepsilon_{\rm Had}\) with probability \(\ge 1-\delta/2\) \label{line:qae}
\State \Return $\widehat{\mathrm{R}}(\brm{\theta}^{k}) \gets b_{[+]}\,\mu_{k+1}^{2}(2\,\widehat{p}^\star_{[+]} - 1)+b_{[-]}\,\mu_{k-1}^{2}(2\,\widehat{p}^\star_{[-]} - 1).$ 
\end{algorithmic}
\end{algorithm}}

    }
\end{figure*}

In this appendix, we provide an implementation-level justification and {proofs (correctness, error analysis, and query and gate counts) for} \Cref{alg:T1,alg:T2}, which solve \Cref{prob:order_parameter_estimation,prob:no-phase-locking_certification}, respectively. The justification includes (i) proof of correctness, (ii) error analysis, and (iii) oracle {query counts} that support the complexity claims in \Cref{thm:main-T1,thm:no-phase-locking}. The main primitives are projected unitary encoding (PUE) of $\bk$ and $\bkk$ described in \Cref{eq:boundary-pue}, block encoding of amplitudes~\cite{guo2021nonlinear,rattew2023nonlinear} the quantum singular value transformation (QSVT)~\cite{Gilyen2019QuantumArithmetics,Martyn2021GrandAlgorithms}, a Hadamard test~\cite{Aharonov2006Jones}, and amplitude estimation~\cite{Brassard2000AmplitudeAmpli}. Throughout, we assume controlled access to each oracle and its inverse, as required by QSVT and amplitude estimation.

\subsection{Embedding the boundary matrix in a quantum circuit}\label{app:embedding_boundary_matrix}

First, to implement \Cref{eq:boundary-pue}, we use the PUE construction of $\brm{U}_{\bk}$ from Ref.~\cite{McArdle2022AQubits}, which has the smallest scaling factor compared to another constructions~\cite{berry2023analyzing,Hayakawa2022QuantumAnalysis}. As such, the rescaling factor of the PUE of $\bk$ is $\nu_{k}=n$ and the number of ancilla used in the PUE is $a_{\bk}=2$. Explicitly, the PUE calls once to a unitary operator $\brm{V}_{\bk}$ that is defined by
\begin{align}\label{eq:Dirac_operator}
    \brm{V}_{\bk}=\frac{1}{\sqrt{n}}\,\sum_{i=1}^{n}\,\brm{Z}^{\otimes(i-1)}\otimes\brm{X}\otimes\brm{1}_{n-i},
\end{align}
where $\brm{1}_{n-i}$ is identity operator acting on the last $(n-i)$-qubit register. 
\begin{lemma}[Correctness of the PUE for the boundary operator~\cite{McArdle2022AQubits}]\label{lem:pue_boundary_correctness} 
{Let $\bk$ be the $k$-th boundary operator. Consider an $(n+2)$-qubit register consisting of two ancilla qubits and an $n$-qubit register where $k$-simplices are represented by states with Hamming weight $k+1$}. Conditioned on postselection of two ancilla qubits being $\ket{00}$, there exists a unitary $\brm{U}_{\bk}$ which is a PUE of $\bk$ with rescaling factor $\nu_{k}^2=\sqrt{n}$, i.e.,
\begin{align}\label{eq:pue_boundary_operator_Dirac}
\left(\brm{1}_{n}\otimes\bra{00}\right)\,\brm{U}_{\bk}\,\left(\brm{1}_{n}\otimes\ket{00}\right)=\bk/\sqrt{n}.
\end{align}
Moreover, $\brm{U}_{\bk}$ uses one call each to the controlled projector{s} $\mathrm{C}_{\Pi_{k}}^{\perp}\mathrm{NOT}$ and $\mathrm{C}_{\Pi_{k-1}}^{\perp}\mathrm{NOT}$, as defined in \Cref{eq:controlled_k_projection} (or, equivalently simplex membership oracles $\brm{O}_{m_k}$ and $\brm{O}_{m_{k-1}}$ as in \Cref{def:membership-oracle}),  and $O(n)$ additional non-Clifford gates in $\brm{V}_{\bk}$.
\end{lemma}
{The circuit diagram on the right side of \Cref{fig:PUE_boundary} describes the implementation of $\brm{U}_{\bk}$ in terms of $\brm{V}_{\bk}$ and the controlled projectors.} For completeness, we include the proof of this lemma from \cite{McArdle2022AQubits}.
\begin{proof}
Fix a computational basis state $\ket{\sigma_k}$ encoding an oriented $k$-simplex. Acting on $\ket{\sigma_k}$, each summand in $\brm{V}_{\bk}$ flips one bit and contributes a sign $(-1)^{\#\{\text{ones before the flipped position}\}}$ from the $\brm{Z}$-string. When the flipped bit is a $1$ (which corresponds to removing a vertex), the resulting bitstring has Hamming weight $k$, corresponding to a $(k-1)$-face of $\sigma_k$, with exactly the orientation sign of the simplex face. When the flipped bit is a $0$ (adding a vertex), the Hamming weight is $k+2$, and the output is rejected by $\mathrm{C}_{\Pi_{k-1}}^{\perp}\mathrm{NOT}$ via the simplex membership oracle $\brm{O}_{m_{k-1}}$. Therefore, conditioning on both projectors accepting (i.e., the two postselection bits being $\ket{00}$), the net action on $\ket{\sigma_k}$ is
\begin{align*}
    \ket{\sigma_k}\mapsto\frac{1}{\sqrt{n}}\sum_{j=0}^{k}(-1)^j\,\ket{\sigma_{k-1}^{j}\subset\sigma_k},
\end{align*}
which is exactly the column action of $\bk/\sqrt{n}$ in the simplex basis. Linearity extends the claim to all superpositions supported on simplices that belong to the simplicial complex $\mathcal{K}_{n}$. Lastly, the stated gate complexity follows from (i) one call to each controlled projector (or, equivalently simplex membership oracle) and (ii) an $O(n)$ non-Clifford gates of $\brm{V}_{\bk}$~\cite{Kerenidis2022QuantumStates}.
\end{proof}

{It is important to emphasize that the controlled projector $\mathrm{C}_{\Pi_k}^{\perp}\mathrm{NOT}$ cannot be implemented using only a Hamming-weight counter. In addition to checking that a basis state has Hamming weight $k+1$, one must also verify that the corresponding support defines a valid $k$-simplex $\sigma_k \in \mathcal{K}_n$. This necessity becomes clear in the implementation of $\bk^{\rm T}$ through $\brm{U}_{\bk}^{\rm T}$. If one were to enforce only the Hamming weight condition, the resulting state would be a linear combination of all computational basis states of Hamming weight $k+1$, including those that do not correspond to $k$-simplices of $\mathcal{K}_n$. Such a state is not the correct output of $\bk^{\rm T}$. Therefore, both the Hamming weight check and the simplex membership test are required, and together they form the simplex membership oracle $\brm{O}_{m_k}$ used inside $\mathrm{C}_{\Pi_k}^{\perp}\mathrm{NOT}$.}

\subsection{Simplicial order parameter estimation}\label{app:tak1_details}

{Assume $\UTheta$ prepares $\ket{\brm{\Theta}^{k}}=\frac{1}{\alpha}\ket{\brm{\theta}^{k}}\ket{0}^{\otimes a}+\ket{\perp_{\Theta}}$, where $\ket{\brm{\theta}^{k}}$ is the normalized amplitude encoding of $\brm{\theta}^k$ and $(\brm{1}\otimes\bra{0}^{\otimes a})\ket{\perp_\Theta}=0$ as in \Cref{def:prob_state}. Define $\gamma_1:=\alpha\,\sqrt{n}\,\mathcal N(\brm{\theta}^k)$ and recall the projected simplicial phase $\brm{\theta}^k_{[-]}=\bk\,\brm{\theta}^k$ and $\brm{\theta}^k_{[+]}=\bkk^{\rm T}\brm{\theta}^k$. }
\begin{lemma}[Probabilistic state preparation unitary for $\brm{\theta}^{k}_{[\pm]}$]\label{lem:prob_projected_simplicial_phase_state}
There exists an $(n+a+2)$-qubit unitary $\brm{U}_{\brm{\Theta}^{k}_{[\pm]}}$ such that
\begin{align}
    \brm{U}_{\brm{\Theta}^{k}_{[\pm]}}\ket{0}^{\otimes(n+a+2)}=\ket{\brm{\Theta}^{k}_{[\pm]}}:=\frac{1}{\alpha_{[\pm]}}\ket{\brm{\theta}^{k}_{[\pm]}}\ket{0}^{\otimes(a+2)}+\ket{\perp_{\Theta_{[\pm]}}},
\end{align}
where 
\begin{eqnarray}
    \ket{\brm{\theta}^{k}_{[\pm]}}:=\frac{1}{\mathcal{N}(\brm{\theta}^{k}_{[\pm]})}\sum_{i\in[n_{k\pm1}]}\left(\brm{\theta}^{k}_{[\pm]}\right)_{i}\ket{\sigma_{k\pm1}^{i}}.\label{eq:projected_simplicial_phase_state}
\end{eqnarray}
is the normalized amplitude encoding of $\brm{\theta}^k_{[\pm]}$, $\alpha_{[\pm]}:=\gamma_1/\mathcal N(\brm{\theta}^k_{[\pm]})$, and $(\brm{1}\otimes\bra{0}^{\otimes(a+2)})\ket{\perp_{\Theta_{[\pm]}}}=0$. Such a unitary uses one call to $\UTheta$, one PUE application of $\bk$ (for $[-]$) or $\bkk^{\rm T}$ (for $[+]$), {and thus} requires {a} single call to $\brm{O}_{k}$ and $\brm{O}_{k-1}$ (for $[-]$) or $\brm{O}_{k}$ and $\brm{O}_{k-1}$ (for $[+]$){, as well as} $O(n)$ additional gates.
\end{lemma}

\begin{proof}
Start from $\ket{0}^{\otimes(n+a+2)}$. Apply $\UTheta$ on the last $n{+}a$ qubits to obtain $\frac{1}{\alpha}\ket{\brm{\theta}^k}\ket{0}^{\otimes a}+\ket{\perp_\Theta}$. Next, apply the PUE $\brm{U}_{\bk}$ (for the $[-]$ branch) or the PUE of $\bkk^{\rm T}$ (for the $[+]$ branch) on the simplex register and the two postselection qubits. By \Cref{lem:pue_boundary_correctness}, conditioning on the postselection qubits being $\ket{00}$, the data register undergoes $\bk/\sqrt{n}$ or $\bkk^{\rm T}/\sqrt{n}$ respectively. Therefore, the “good” component becomes
\begin{align*}
    \frac{1}{\alpha}\left(\frac{1}{\sqrt{n}}\,\bk\ket{\brm{\theta}^k}\right)=\frac{1}{\alpha\sqrt{n}}\left(\frac{1}{\mathcal N(\brm{\theta}^k)}\,\bk\,\brm{\theta}^k\right)=\frac{\mathcal N\left(\brm{\theta}^k_{[-]}\right)}{\gamma_1}\ket{\brm{\theta}^k_{[-]}}
\end{align*}
(and analogously for the $[+]$ branch). Defining $\alpha_{[\pm]}=\gamma_1/\mathcal N(\brm{\theta}^k_{[\pm]})$ gives the stated form. Gate complexity is one call to $\UTheta$ and one PUE call, which contributes one call to $\brm{O}_{m_{k}}$ and $\brm{O}_{m_{k-1}}$ for the $[-]$ branch or $\brm{O}_{m_{k}}$ and $\brm{O}_{m_{k+1}}$ for the $[+]$ branch, and $O(n)$ additional gates by \Cref{lem:pue_boundary_correctness}.
\end{proof}
\noindent See the quantum circuit diagram for $\brm{U}_{\brm{\Theta}^{k}_{[\pm]}}$ in the left panel of \Cref{fig:PUE_boundary}. 

\begin{figure}
    \centering
    \includegraphics[width=1\linewidth]{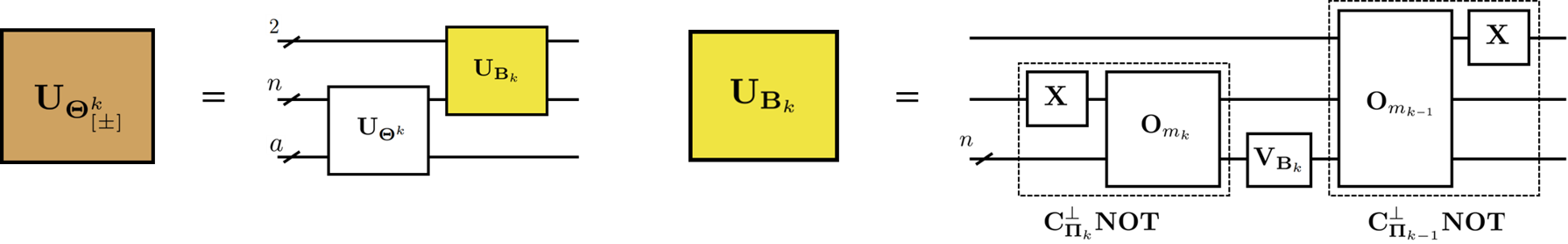}
    \caption{ (Left) Quantum circuit preparing $\ket{\brm{\Theta}^{k}_{[\pm]}}$ that encodes projected simplicial phase $\brm{\theta}^{k}_{[\pm]}$ in \Cref{eq:projected_simplicial_phase_state}. (Right) Projected unitary encoding (PUE) of the boundary operator $\bk$ in \Cref{eq:boundary-pue} with the Dirac operator $\brm{V}_{\bk}$ given in \Cref{eq:Dirac_operator}.}
    \label{fig:PUE_boundary}
\end{figure}

Recall that the simplicial order parameter in Eq.~\eqref{eq:general_projected_order} can be rewritten as
\begin{align}\label{eq:projected_order_trace}
    \mathrm{R}_{[\pm]}(\brm{\theta}^{k})=\frac{\mathrm{Tr}\left(\cos{\ALU}\right)}{n_{k\pm1}},
\end{align}
where 
\begin{gather}\label{eq:diag_matrix_simplicial_signal}
    \ALU=\mathrm{diag}\left(\left(\brm{\theta}^{k}_{[\pm]}\right)_{0},\cdots,\left(\brm{\theta}^{k}_{[\pm]}\right)_{n_{k\pm1}-1}\right),
\end{gather}
are the diagonal matrices whose diagonal entries are the components of $\brm{\theta}^{k}_{[\pm]}$. {Note that $\ALU$ has eigenvalues $\{(\brm{\theta}^{k}_{[\pm]})_{i}\}_{i\in[n_{k\pm1}]}$ with eigenvector $\{\sigma_{k\pm1}^{i}\}_{i\in[n_{k\pm1}]}$.} The simplicial order parameter $\mathrm{R}(\brm{\theta}^{k})$ can be expressed as 
\begin{align}\label{eq:general_order_trace}
    \mathrm{R}(\brm{\theta}^{k})=b_{[+]}\frac{\mathrm{Tr}\left(\cos{\ALU}\right)}{n_{k+1}}+b_{[-]}\frac{\mathrm{Tr}\left(\cos{\ALU}\right)}{n_{k-1}}.
\end{align}
We construct a block encoding $\UALU$ of this matrix with rescaling factor $\gamma_{1}=\alpha\,\sqrt{n}\,\mathcal{N}(\brm{\theta}^{k})$ with the following complexity. We defer the proof to \Cref{App: modified BE of amplitudes}.
\begin{lemma}[Modified block encoding of matrix $\ALU/\gamma$ with simplex membership oracle]\label{lem:modified_BE_of_amplitudes}
    There exists a quantum circuit constructing a block encoding $\UALU$ of a diagonal matrix $\brm{A}_{[\pm]}/\gamma_{1}=1/\gamma_{1}\,\mathrm{diag}\big(\big(\brm{\theta}_{[\pm]}^{k}\big)_{0},\cdots,\big(\brm{\theta}_{[\pm]}^{k}\big)_{n_{k\pm1}-1}\big)${, where the eigenvector of the eigenvalue $(\mathbf{\theta}^k_{[\pm]})_i$ is $\ket{\sigma_{k\pm 1}^i}$ for each $i \in [n_{k \pm 1}]$. This circuit calls} $O(1)$ times to controlled{-}$\UTheta$ and its inverse, $O(1)$ times to controlled{-}$\Om$, $O(n)$ additional gates, and $a_{[\pm]}=n+a+5$ ancilla qubits.
\end{lemma}

\Cref{prob:order_parameter_estimation} is a quantum computational task to obtain an estimator $\widehat{\mathrm{R}}(\brm{\theta}^{k})$ of the simplicial order parameter defined in Eq.~\eqref{eq:general_order}. To do so, we implement a subtask to obtain two independent estimations for upper and lower projected simplicial order parameters defined in \Cref{eq:general_projected_order}, which we denote $\widehat{\mathrm{R}}_{[+]}(\brm{\theta}^{k})$ and $\widehat{\mathrm{R}}_{[-]}(\brm{\theta}^{k})$. The corresponding error budgets and failure probabilities are set to ensure that the combined estimate is a $\varepsilon$-close approximation to $\mathrm{R}(\brm{\theta}^{k})$. First, we formally state the complexity of estimating projected order parameters $\mathrm{R}_{[\pm]}(\brm{\theta}^{k})$ as follows.
\begin{lemma}[Projected order parameter estimation]\label{lem:Projected_order_parameters_estimation}
    Fix $\varepsilon,\delta\in (0,1/2)$. Let $\gamma_{1}=\alpha\,\sqrt{n}\,\mathcal{N}(\brm{\theta}^{k})$. Given access to controlled application{s} of $\UTheta,\,\USigmaLU,\Om$, there exists a quantum algorithm estimating $\mathrm{R}_{[\pm]}(\brm{\theta}^{k})$ within an additive error $\varepsilon$ with probability of success at least $1-\delta$, where
    \begin{align*}
        \mathrm{R}_{[\pm]}(\brm{\theta}^k)=\frac{1}{n_{k\pm1}}\sum_{i\in[n_{k\pm1}]}\cos\!\big((\brm{\theta}^k_{[\pm]})_i\big)=\bra{\brm{\sigma}_{k\pm1}}\cos(\ALU)\ket{\brm{\sigma}_{k\pm1}},
    \end{align*}
    and $\ket{\brm{\sigma}_{k\pm1}}:=\frac{1}{\sqrt{n_{k\pm1}}}\sum_{i\in[n_{k\pm1}]}\ket{\sigma^i_{k\pm1}}$. The algorithm uses 
    \begin{align*}
        \widetilde{O}\left(\gamma_{1}\,\mu_{k\pm1}^2\log{(\delta^{-1})}\varepsilon^{-1}\right)
    \end{align*}
    calls to $\UTheta$ and $\Om$, 
    \begin{align*}
        O\left(\mu_{k\pm1}^2\log{(\delta^{-1})}\varepsilon^{-1}\right)
    \end{align*}
    calls to $\USigmaLU$, and $\widetilde{O}((n+a)\,\gamma_{1}\,\mu_{k\pm1}^2\log{(\delta^{-1})}\varepsilon^{-1})$ additional gates.
\end{lemma}
\noindent {Before proving the above lemma, we state the following result, which we use to implement the trigonometric polynomial transformation of $\ALU$. 
\begin{lemma}[Polynomial approximation to the trigonometric functions~\cite{Gilyen2019QuantumArithmetics}]\label{lem:poly approx to trig functions}
        For $y\in[-1,1]$, $t\ge 0$, and $\varepsilon_b\in(0,1/2)$, the polynomials $q_c,\,q_s$ of degree $2D$ and $2D+1$, respectively, with 
        \begin{align*}
            D=O\left(t+\frac{\log{(\varepsilon_b^{-1})}}{\log{(\log{(\varepsilon_b^{-1})}/t)}}\right)
        \end{align*}
        satisfy
        \begin{align*} 
            \sup_{y\in[-1,1]}\left|\cos(t y) - q_{c}(y)\right|\leq \varepsilon_b\quad\text{and}\quad\sup_{y\in[-1,1]}\left|\sin(t y) - q_{s}(y)\right|\leq\varepsilon_b.
        \end{align*}
    \end{lemma}
}
\begin{proof}[Proof of \Cref{lem:Projected_order_parameters_estimation}]
    Let $\ALU=\mathrm{diag}((\brm{\theta}^k_{[\pm]})_0,\ldots,(\brm{\theta}^k_{[\pm]})_{n_{k\pm1}-1})$. By \Cref{lem:modified_BE_of_amplitudes}, we can implement a block encoding $\UALU$ to create a block encoding of the diagonal matrix: 
    \begin{align}
        \frac{\ALU}{\gamma_{1}}=\frac{1}{\gamma_1}\,\mathrm{diag}\left(\left(\brm{\theta}^{k}_{[\pm]}\right)_{0},\cdots,\left(\brm{\theta}^{k}_{[\pm]}\right)_{n_{k\pm1}-1}\right).
    \end{align}
    The eigenvector of the eigenvalue of $\big(\ALU\big)_{i}$ is $\ket{\sigma_{k\pm1}^{i}}$ for all $i\in[n_{k\pm1}]$. Thus, the application of this block encoding to $\ket{\brm{\Sigma}_{k\pm1}}=\ket{0}^{\otimes a_{k\pm1}}\big(1/\mu_{k\pm1}\,\sum_{i\in[n_{k\pm1}]}\ket{\sigma}_{k\pm1}\big)+\ket{\perp_{\brm{\Sigma}_{k\pm1}}}$ given in \Cref{def:prob-sigma} yields
    
    \begin{eqnarray}
        \Big(\UALU\otimes\brm{1}_{a_{k}}\Big)\ket{0}^{\otimes a_{[\pm]}}\ket{\brm{\Sigma}_{k\pm1}}&=&\bigg[\frac{1}{\mu_{k\pm1}\,\gamma_{1}}\,\ALU\ket{\brm{\sigma}_{k\pm1}}\ket{0}^{\otimes (a_{k\pm1}+a_{[\pm]})}+\ket{\perp_{\ALU}}\bigg]\nonumber\\
        &=&\left[\frac{1}{\mu_{k\pm1}\,\gamma_{1}}\,\mathlarger{\sum}_{i\in[n_{k\pm1}]}\left(\brm{\theta}^{k}_{[\pm]}\right)_{i}\ket{\sigma^{i}_{k\pm1}}\ket{0}^{\otimes (a_{k\pm1}+a_{[\pm]})}+\ket{\perp_{\ALU}}\right],\label{eq:application_of_BE_amplitudes}
    \end{eqnarray}
    where $\left(\brm{1}_{n}\otimes\bra{0}^{\otimes (a_{k\pm1}+a_{[\pm]})}\right)\ket{\perp_{\ALU}}=0$.

    Using QSVT~\cite{Gilyen2019QuantumArithmetics}, we  transform the block encoding of $\ALU/\gamma_{1}$ to the block encoding of $q\big(\ALU/\gamma_{1}\big)$ such that
    \begin{align}\label{eq:approx_cos_diagonal_matrix}
        \left\|\cos{\big(\ALU\big)}-q\big(\ALU/\gamma_1\big)\right\|_2\leq\varepsilon_{\cos}.
    \end{align}
    To do so, the QSVT calls to $\UALU$ several times based on the degree of the polynomial $q$. The output of QSVT is a block encoding of $q(\ALU/\gamma_1)$, which we denote $\brm{U}_{q_{[\pm]}}$, such that
    \begin{align}
        q\big(\ALU/\gamma_1\big)=(\brm{1}_{n+a_{[\pm]}+1}\otimes\bra{0})\brm{U}_{q_{[\pm]}}(\brm{1}_{n+a_{[\pm]}+1}\otimes&\ket{0})
    \end{align} 
    To determine the polynomial $q$ satisfying \Cref{eq:approx_cos_diagonal_matrix}, we use \Cref{lem:poly approx to trig functions}. We then invoke $q=q_c$ with the degree
    \begin{eqnarray}\label{eq:degree_cos_poly}
        d=\Theta\left(\gamma_{1}+\frac{\log{(1/\varepsilon_{\cos})}}{\log{\left(\log{(1/\varepsilon_{\cos})}/\gamma_{1}\right)}}\right).
    \end{eqnarray}
    The query cost of the QSVT step is given by $O(d)$ applications of $\UALU$ with additional  $O((n+a)d)$ gates. The number of ancillas used in this circuit construction up to this point is $a_{[\pm]}+1$. See the illustration of the QSVT circuit for constructing $\brm{U}_{q_{[\pm]}}$ in \Cref{fig:QSVT_and_projected_Hadamard}. 

    We then construct the Hadamard test unitary between $\brm{U}_{q(\ALU)}\ket{\brm{\Sigma}_{k\pm1}}$ and $\ket{\brm{\Sigma}_{k\pm1}}$ (see \Cref{fig:QSVT_and_projected_Hadamard} for the circuit diagram of the Hadamard test). This Hadamard test unitary uses the controlled application of $\brm{U}_{q_{[\pm]}}$ with a single control qubit $\ket{+}=1/\sqrt{2}\,(\ket{0}+\ket{1})$ that we  measure in the Hadamard test, such that we have
    \begin{align*}
        \Big(\brm{U}_{q_{[\pm]}}\otimes\brm{1}_{a_{k\pm1}}&\otimes\ket{1}\bra{1}\Big)\ket{\brm{\Sigma}_{k\pm1}}\ket{0}^{\otimes(a_{[\pm]+1})}\ket{+}\\&=\left[\frac{1}{\mu_{k\pm1}}\,\ket{0}^{\otimes(a_{k\pm1}+ a_{[\pm]}+1)}\,q\left(\ALU\right)\ket{\brm{\sigma}_{k\pm1}}+\ket{\perp_{q(\brm{A})}}\right]\ket{1}+\left(\ket{0}^{\otimes (a_{[\pm]}+1)}\ket{\brm{\Sigma}_{k\pm1}}\right)\ket{0},
    \end{align*}
    where $\left(\brm{1}_{n}\otimes\bra{0}^{\otimes (a_{l\pm1}+a_{[\pm]}+1)}\right)\ket{\perp_{q(\brm{A})}}=0$. {Define 
    \begin{eqnarray}
        p^\star_{[\pm]}=\frac{1+\bra{\brm{\sigma}_{k\pm1}}\cos{\ALU}\ket{\brm{\sigma}_{k\pm1}}/\mu_{k\pm1}^{2}}{2},
    \end{eqnarray}
    thus $\mathrm{R}_{[\pm]}(\brm{\theta}^{k})=\mu_{k\pm1}^{2}\,(2\,p^{\star}_{[\pm]}-1)$. The probability of measuring the ancilla qubit in $\ket{1}$ is given by
    \begin{align}\label{eq:probability_Hadamard_test}
        \tilde{p}^\star_{[\pm]}=\frac{1+\bra{\brm{\sigma}_{k\pm1}}q(\ALU)\ket{\brm{\sigma}_{k\pm1}}/\mu_{k\pm1}^{2}}{2},
    \end{align}
    with $\widetilde{\mathrm{R}}_{[\pm]}(\brm{\theta}^{k}):=\mu_{k\pm1}^{2}\,(2\,\tilde{p}^{\star}_{[\pm]}-1)$. So, we have that
    \begin{eqnarray}\label{eq:approx_cosine_QSVT}
        \left|\widetilde{\mathrm{R}}_{[\pm]}(\brm{\theta}^{k})-\mathrm{R}_{[\pm]}(\brm{\theta}^{k})\right|\leq \varepsilon_{\cos} .
    \end{eqnarray}}

    Lastly, we run amplitude estimation on the Hadamard test unitary to estimate $p^\star_{[\pm]}$ within additive error $\varepsilon_p$ with failure probability at most $\delta$ using $O(\log(\delta^{-1})/\varepsilon_p)$ controlled applications of the unitary and its inverse. This estimation outputs
    \begin{align*}
        \widehat{\mathrm{R}}_{[\pm]}(\brm{\theta}^k):=\mu_{k\pm1}^2\big(2\widehat{p}^\star_{[\pm]}-1\big).
    \end{align*}
    Then
    \begin{align}\label{eq:approx_amplitude_estimation}
        \big|\widehat{\mathrm{R}}_{[\pm]}(\brm{\theta}^k)-\widetilde{\mathrm{R}}_{[\pm]}(\brm{\theta}^k)\big|\le2\mu_{k\pm1}^2\,\varepsilon_p.
    \end{align}
    Choosing $\varepsilon_{\cos}=\varepsilon/2$ in Eq.~\eqref{eq:approx_cosine_QSVT} and $\varepsilon_p=\varepsilon/(4\mu_{k\pm1}^2)$ in \Cref{eq:approx_amplitude_estimation} gives the total error of estimating $\mathrm{R}(\brm{\theta}^{k})$ at most $\varepsilon$ by the triangle inequality.
    
    Each amplitude estimation query invokes one coherent Hadamard test, which invokes one $\USigmaLU$ call and one application of $\brm{U}_{q_{[\pm]}}$, hence $O(d)$ uses of $\UALU$ with $O((n+a)d)$ additional gates. Each $\UALU$ call uses $O(1)$ {controlled applications of} $\UTheta$ and $\Om$ plus $O(n)$ additional gates (\Cref{lem:modified_BE_of_amplitudes}). Therefore, the total query counts are
    \begin{eqnarray*}
        O\left(d\,\mu_{k\pm1}^2\log{(\delta}^{-1})\,\varepsilon^{-1}\right)= \widetilde{O}\left(\gamma_{1}\,\mu_{k\pm1}^2\log{(\delta}^{-1})\,\varepsilon^{-1}\right)\quad\text{for }\UTheta\text{ and }\Om, 
    \end{eqnarray*}
    and
    \begin{eqnarray*}
        O(\mu_{k\pm1}^2\log{(1/\delta})\,\varepsilon^{-1})\quad\text{for }\USigmaLU,
    \end{eqnarray*}
    with {$O((n+a)\,\gamma_{1}\,\mu_{k\pm1}^2\log{(\delta^{-1}})\,\varepsilon^{-1})$ additional gates used}. {This concludes the proof of the total query and gate complexities, accounting for the estimation of} projected order parameters.
\end{proof}

\begin{figure*}
    \centering
    \includegraphics[scale=0.57]{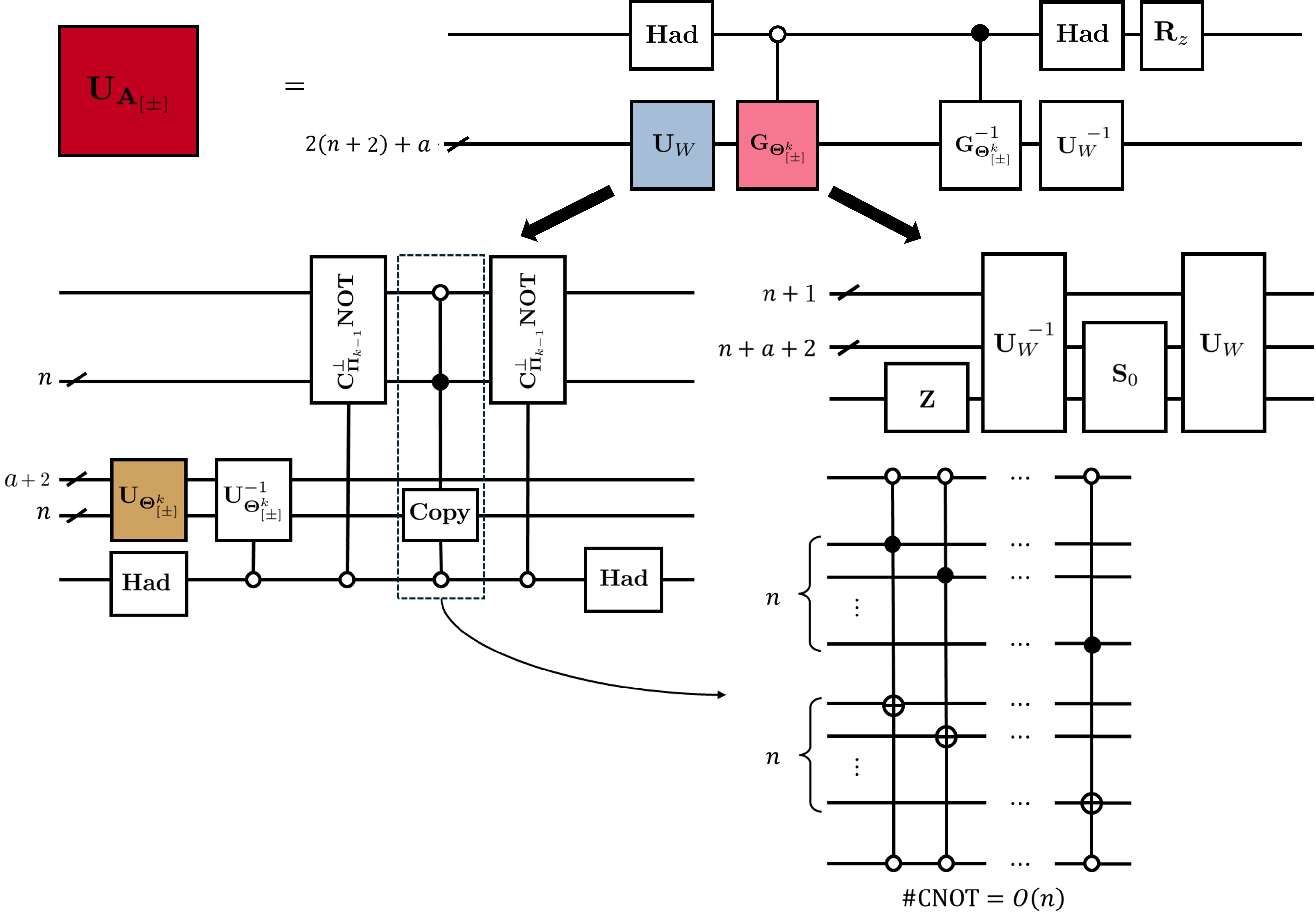}
    \caption{A modified version of a quantum circuit for block encoding the amplitudes with membership oracles that is used to construct $\UALU$. There are two important parts in the block encoding: $\brm{U}_{W}$ and $\brm{G}_{\brm{\Theta}^{k}_{[\pm]}}$. It is important to note that we replace $n$ applications of multi-controlled $\mathrm{NOT}$ gate in the $\brm{Copy}$ subroutine that takes $a$ ancillas as control qubits using the simplex membership oracle $\brm{O}_{m_{k\pm1}}$ in $\mathrm{C}_{\Pi_{k\pm1}}^{\perp}\mathrm{NOT}$, as in \Cref{fig:PUE_boundary}. As a result, the number of $\mathrm{CNOT}$s in the $\brm{Copy}$ subroutine is $O(n)$, instead of $O(na)$. If the gate complexity of $\Om$ plus $O(n)$ gates is less than $O(na)$ gates, the modified version is more efficient than the original one. 
    }
    \label{fig:modified_BE_amplitudes}
\end{figure*}

\begin{figure}
    \centering
    \includegraphics[width=0.9\linewidth]{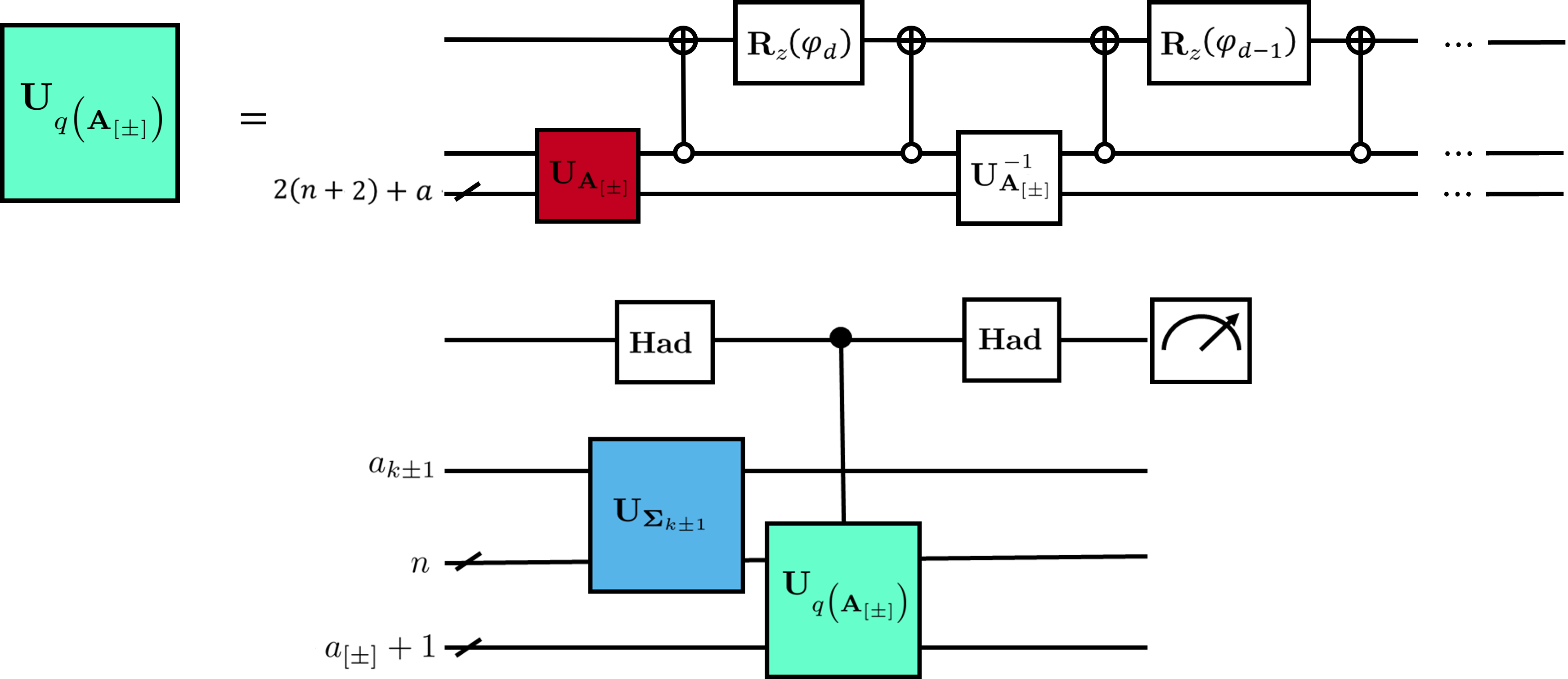}
    \caption{(Top) Implementation of QSVT with the block encoding of amplitudes $\UALU$ constructing a block encoding $\brm{U}_{q_{[\pm]}}$. (Bottom)  Hadamard test subroutine to evaluate the trace $\bra{\brm{\sigma}_{k\pm1}}q\big(\ALU\big)\ket{\brm{\sigma}_{k\pm1}}$ given only probabilistic state preparation unitary $\USigmaLU$. The probability of outputting $z=1$ encodes an estimate of the desired trace value.}
    \label{fig:QSVT_and_projected_Hadamard}
\end{figure}

After establishing a complexity statement for estimating projected order parameters, it is straightforward to prove \Cref{thm:main-T1}. The algorithm procedure is given in \Cref{alg:T1}.
\begin{proof}[Proof of \Cref{thm:main-T1}] 
    By definition given in Eq.~\eqref{eq:general_order},
    \begin{align*}
        \mathrm{R}(\brm{\theta}^k)=b_{[-]}\mathrm{R}_{[-]}(\brm{\theta}^k)+b_{[+]}\mathrm{R}_{[+]}(\brm{\theta}^k),\quad b_{[\pm]}:=\frac{n_{k\pm1}}{n_{k+1}+n_{k-1}}.
    \end{align*}
    To estimate $\mathrm{R}(\brm{\theta}^k)$, we run the projected estimator in \Cref{lem:Projected_order_parameters_estimation} for each branch $[\pm]$ with failure probability $\delta/2$. We then set the additive error budgets so that each weighted contribution $b_{[\pm]}\widehat{\mathrm{R}}_{[\pm]}$ has error at most $\varepsilon/2$, {i.e.,} $|\widehat{\mathrm{R}}_{[\pm]}-\mathrm{R}_{[\pm]}|\le \varepsilon/(2b_{[\pm]})$. Under this choice, a union bound gives success probability at least $1-\delta$, and
    \begin{align*}
        \big|\widehat{\mathrm{R}}(\brm{\theta}^k)-\mathrm{R}(\brm{\theta}^k)\big|\le b_{[-]}\big|\widehat{\mathrm{R}}_{[-]}-\mathrm{R}_{[-]}\big|+b_{[+]}\big|\widehat{\mathrm{R}}_{[+]}-\mathrm{R}_{[+]}\big|\le \varepsilon.
    \end{align*}
    For query complexity, \Cref{lem:Projected_order_parameters_estimation} implies that estimating $\mathrm{R}_{[\pm]}$ to precision $\varepsilon/(2b_{[\pm]})$ {requires} $\widetilde{O}(\gamma_1\mu_{k\pm1}^2 b_{[\pm]}\log(\delta^{-1})\varepsilon^{-1})$ calls to $\UTheta$ and $\Om$, and $O(\mu_{k\pm1}^2 b_{[\pm]}\log(\delta^{-1})\varepsilon^{-1})$ calls to $\USigmaLU$. Summing {the} $[+]$ and $[-]$ contributions yields
    \begin{align}
        \widetilde{O}\!\left(\gamma_1 r_1\right)\ \text{calls to }\UTheta\text{ and }\Om,\quad O(r_1)\ \text{calls to }\USigmaLU,
    \end{align}
    and $\widetilde{O}(n\gamma_1 r_1)$ additional gates, where
    \begin{align}
        r_1=O\left( \big(b_{[-]}\mu_{k-1}^2+b_{[+]}\mu_{k+1}^2\big)\,\log(\delta^{-1})\,\varepsilon^{-1}\right)=O\left(\frac{\mu_{k-1}^2\,n_{k-1}+\mu_{k+1}^2\,n_{k+1}}{n_{k+1}+n_{k-1}}\,\log(\delta^{-1})\,\varepsilon^{-1}\right).
    \end{align}
This completes the proof.
\end{proof}

\subsection{Details of the modified block encoding of amplitudes}\label{App: modified BE of amplitudes}

In this section, we detail the implementation of block encoding of amplitudes $\UALU$ in \Cref{lem:Projected_order_parameters_estimation}. Here, we modify the block encoding of amplitudes described originally in Refs.~\cite{guo2021nonlinear,rattew2023nonlinear}. Our goal is to leverage the simplex membership oracle to reduce the number of gates in the previous construction.

Recall the unitary $\brm{U}_{\brm{\Psi}}$ that prepares $\ket{\brm{\Psi}}=1/\alpha_{\brm{\psi}}\ket{\brm{\psi}}\ket{0}^{\otimes a_{\brm{\psi}}}+\ket{\perp_{\brm{\Psi}}}$ as in \Cref{def:prob-SKM}. We restate the block encoding of amplitudes as follows. 
\begin{theorem}[block encoding of amplitudes~\cite{guo2021nonlinear,rattew2023nonlinear}]
    There exists a quantum circuit constructing a block encoding of a diagonal matrix $\brm{A}_{\brm{\psi}}=1/\alpha_{\brm{\psi}}\,\mathrm{diag}(\psi_{0},\cdots,\psi_{N-1})$ by calling $O(1)$ times to controlled applications of $\brm{U}_{\brm{\Psi}}$ and its inverse, and $O(na_{\brm{\psi}})$ additional gates.
\end{theorem}

The block encoding of amplitudes protocol used in the previous result requires a subroutine that copies {each} basis state of $\ket{\brm{\psi}}$ by applying multi-controlled $\mathrm{NOT}$ gates with controls on $\ket{0}^{a_{\brm{\psi}}}$. This gives $O(na_{\brm{\psi}})$ gate complexity outside the gate complexity of $\brm{U}_{\brm{\Psi}}$. For our case in \Cref{alg:T1}, we introduce a slight modification of this subroutine by using the membership oracles $\brm{O}_{m_{k\pm1}}$ to reduce the number of controls from $a_{\brm{\psi}}$ to two, and thus replace the gate complexity {of} $O(na_{\brm{\psi}})$ {with} $O(\mathrm{G}_{\brm{O}_{m_{k\pm1}}}+n)$, where $\mathrm{G}_{\brm{O}_{m_{k\pm1}}}$ is the the gate complexity of $\brm{O}_{m_{k\pm1}}$. 

\begin{proof}[Proof of \Cref{lem:modified_BE_of_amplitudes}]
    The proof of correctness closely follows the procedure in the original block encoding of amplitudes in Refs.~\cite{guo2021nonlinear,rattew2023nonlinear}. We provide a step-by-step description of our modified version for completeness. The procedure is given as follows. 
\begin{enumerate}
    \item Given access to controlled version of $\brm{U}_{\brm{\Theta}^{k}_{[\pm]}}$ that prepares  $\ket{\brm{\Theta}^{k}_{[\pm]}}$, and its inverse,  construct a $(2(n+2)+a)$-qubit unitary $\brm{U}_{W}$ given in \Cref{fig:modified_BE_amplitudes}. When applied to the state $\ket{0}^{\otimes(n+a+4)}\ket{\brm{\Sigma}_{k\pm1}}$, such a unitary prepares a superposition between $\ket{\brm{\Theta}^{k}_{[\pm]}}$ and $\ket{\sigma_{k\pm1}^{i}}\ket{0}^{\otimes(a+2)}$ such that 
    \begin{eqnarray}\label{eq:W_state}
        \Big(\brm{U}_{W}&\otimes&\brm{1}_{a_{k\pm1}}\Big)\ket{0}^{\otimes(n+a+4)}\ket{\brm{\Sigma}_{k\pm1}}\nonumber\\&=&\frac{1}{2\,\mu_{k\pm1}}\,\mathlarger{\sum}_{i\in[n_{k\pm1}]}\ket{\sigma_{k\pm1}^{i}}\ket{0}^{\otimes a_{k\pm1}}\bigg[\left(\ket{\brm{\Theta}^{k}_{[\pm]}}+\ket{\sigma_{k\pm1}^{i}}\ket{0}^{\otimes (a+2)}\right)\ket{0}+\left(\ket{\brm{\Theta}^{k}_{[\pm]}}-\ket{\sigma_{k\pm1}^{i}}\ket{0}^{\otimes (a+2)}\right)\ket{1}\bigg]\ket{0}\nonumber\\&&+\;\frac{1}{2}\ket{\perp_{\brm{\Sigma}_{k\pm1}}}\bigg[\left(\ket{\brm{\Theta}^{k}_{[\pm]}}+\ket{0}^{\otimes (n+a+2)}\right)\ket{0}+\left(\ket{\brm{\Theta}^{k}_{[\pm]}}-\ket{0}^{\otimes (n+a+2)}\right)\ket{1}\bigg]\ket{1}.
    \end{eqnarray}
    In this step, we use the simplex membership oracle to replace the multi-controlled $\mathrm{NOT}$ gates. The high-level circuit diagram is given in \Cref{fig:modified_BE_amplitudes}, which includes the application of $\brm{U}_{\brm{\Theta}^{k}_{[\pm]}}$ given in \Cref{lem:prob_projected_simplicial_phase_state} and the simplex membership oracle $\brm{O}_{m_{k\pm1}}$ (as defined in \Cref{def:membership-oracle}) that is used in the $\mathrm{C}_{\Pi_{k\pm1}^{\perp}}\mathrm{NOT}$. The application of $\mathrm{C}_{\Pi_{k\pm1}^{\perp}}\mathrm{NOT}$ ensures that only the correct basis states $\{\sigma_{k\pm1}^{i}\}_{i\in[n_{k\pm1}]}$ activate $\brm{Copy}:=\bigotimes_{i\in[n]}\big(\ket{0}\bra{0}\otimes\ket{1}\bra{1}^{\otimes i}\otimes\brm{X}_{i}\big)$ that copies the basis state $\ket{\sigma_{k\pm1}^{i}}$ into another quantum register. Thus, the whole unitary $\brm{U}_{W}$, comprised of controlled version of $\brm{U}_{\brm{\Theta}^{k}_{[\pm]}}$ and its inverse, $\mathrm{C}_{\Pi_{k\pm1}^{\perp}}\mathrm{NOT}$, and $\brm{Copy}$,  requires $O(1)$ applications of controlled applications of $\UTheta$ (and its inverse) and $\Om$ for $p=\{k-1,k,k+1\}$,  with $O(n)$ additional gates that construct $\brm{Copy}$ and the Dirac operator $\brm{V}_{\bk}$ defined in \Cref{eq:Dirac_operator}. The number of ancillas used in this unitary is $n+a+4$.
    \item Construct a Grover-like iterate $\brm{G}_{\brm{\Theta}^{k}_{[\pm]}}$ from $\brm{U}_{W}$ (and its inverse) and a reflection operator around $\ket{0}^{\otimes(n+a+3)}$, that is denoted by  $\brm{S}_{0}=\brm{1}_{n+1}\otimes\Big(\brm{1}_{n+a+3}-2\ket{0}\bra{0}^{\otimes(n+a+3)}\Big)$. Explicitly, the operator is given by
    \begin{align}\label{grover iterate of BE of amplitudes}
        \brm{G}_{\brm{\Theta}^{k}_{[\pm]}}:=\brm{U}_{W}\left(\brm{1}_{n+a+2}\otimes\brm{S}_{0}\right)\brm{U}_{W}\left(\brm{1}_{2n+a+4}\otimes\brm{Z}\right).
    \end{align}
    This operator has eigenvalues $\{-\exp{\Big( i\arccos{\big(2\,\big(\brm{\theta}^{k}_{[\pm]}\big)_{i}}\big)/\gamma_{1}\Big)},-\exp{\Big( -i\arccos{\big(2\,\big(\brm{\theta}^{k}_{[\pm]}\big)_{i}}\big)/\gamma_{1}\Big)}\}$ for all $i\in[n_{k\pm1}]$. Thus, constructing 
    \begin{align}
        -\frac{1}{2} \,\big(\brm{G}_{\brm{\Theta}^{k}_{[\pm]}}+\brm{G}_{\brm{\Theta}^{k}_{[\pm]}}^{-1}\big)
    \end{align}
    yields an operator with eigenvalues $\big\{\big(\brm{\theta}^{k}_{[\pm]}\big)_{i}/\gamma_{1}\big\}_{i\in[n_{k\pm1}]}$.
    \item Adding an ancilla, create a block encoding of a diagonal matrix $\ALU/\gamma_1$ described in \Cref{eq:diag_matrix_simplicial_signal}. This block encoding is given by
    \begin{align}
        \UALU:=\big(\brm{Had}\otimes\brm{1}_{2(n+2)+a}\big)\brm{U}_{\brm{G}}\big(\brm{Had}\,\brm{R}_{z}(\pi)\otimes\brm{1}_{2(n+2)+a}\big),
    \end{align}
    where $\brm{Had}$ is the Hadamard gate and $\brm{U}_{\brm{G}}:=\big(\brm{1}\otimes\brm{U}_{W}^{-1}\big)\left(\ket{1}\bra{1}\otimes\brm{G}_{\brm{\Theta}^{k}_{[\pm]}}^{-1}+\ket{0}\bra{0}\otimes\brm{G}_{\brm{\Theta}^{k}_{[\pm]}}\right)\big(\brm{1}\otimes\brm{U}_{W}\big)$.
\end{enumerate}
Then, when $\big(\UALU\otimes\brm{1}_{a_{k\pm1}}\big)$ is applied to $\ket{\brm{\Sigma}_{k\pm1}}$, it yields
\begin{align}
    \Big(\UALU\otimes\brm{1}_{a_{k\pm1}}\Big)\ket{\brm{\Sigma}_{k\pm1}}=\bigg[\frac{1}{\mu_{k\pm1}\,\gamma_{1}}\,\ALU\ket{\brm{\sigma}_{k\pm1}}\ket{0}^{\otimes a_{k\pm1}}\ket{0}+\ket{\perp_{\brm{\Sigma}_{k\pm1}}}\ket{1}\bigg]\ket{0}^{\otimes(n+a+4)}.
    \end{align}
By denoting $\ket{\perp_{\brm{\Sigma}_{k\pm1}}}\ket{1}\ket{0}^{\otimes(n+a+4)}=\ket{\perp_{\ALU}}$, the above expression is equivalent to Eq.~\eqref{eq:application_of_BE_amplitudes}. The  unitary $\UALU$ uses $O(1)$ controlled applications of $\UTheta$ (and its inverse), and $\Om$ for $p=\{k-1,k+1\}$, with $O(n)$ additional gates and $n+a+5$ ancilla qubits. This proves the statement we give regarding the cost of $\UALU$ in the proof of \Cref{lem:Projected_order_parameters_estimation}.
\end{proof}
\noindent  The high-level circuit diagrams are depicted in \Cref{fig:modified_BE_amplitudes}.

\subsection{No-phase-locking (NPL) certification}\label{app:task2_details}

\begin{figure*}[ht]
  \centering
  \fbox{
    \parbox{0.49\textwidth}{
\begin{algorithm}[H]
\caption{\textbf{for \Cref{prob:no-phase-locking_certification}} \\ No--Phase--Locking Certification}
\label{alg:T2}
\begin{algorithmic}[1]
\renewcommand{\algorithmicrequire}{\textbf{Input:}}
\renewcommand{\algorithmicensure}{\textbf{Output:}}

\Require\Statex\begin{itemize}[leftmargin=*, nosep]
    \item Coupling constant $K_{\mathfrak q}>0$ and index $\mathfrak q\in\{k-1,k+1\}$
    \item Gap promise $|K_{\mathfrak q}-K_{\mathfrak q}^{s}|\ge \Delta>0$
    \item Failure probability $\delta\in(0,1/2)$
    \item Probabilistic state preparation $\UOmega$ preparing $\ket{\brm{\Omega}^{k}}$ in \Cref{def:prob-SKM}
    \item Membership oracles $\Om$ to implement the PUEs of $\bk$ and $\bkk$ in \Cref{eq:boundary-pue}
    \item Parameters: $\beta$, $n$, $n_{\mathfrak q}$, ${\cal N}(\brm{\omega}^{k})$, and $\kappa_{\mathfrak m}$, with
    \[
      \mathfrak m\gets \begin{cases}k & \mathfrak q=k-1\\ k+1 & \mathfrak q=k+1\end{cases}.
    \]
\end{itemize}

\Ensure A decision bit $z\in\{0,1\}$ with $\mathrm{Pr}[z\;\text{is correct}]\leq1-\delta$, where $z=1$ certifies NPL regime.

\Statex
\Statex\textbf{Main Algorithm:}
\State Set error budgets:
\[
\varepsilon_{\rm AE}\gets \frac{n\,n_{\mathfrak q}}{12\,\kappa_{\mathfrak m}^{2}\,\beta^{2}\,{\cal N}(\brm{\omega}^{k})^{2}}\Delta^{2},\qquad
\varepsilon_{\rm enc}\gets \frac{n\,n_{\mathfrak q}}{12\,\kappa_{\mathfrak m}^{2}\,\beta^{2}\,{\cal N}(\brm{\omega}^{k})^{2}}\Delta^{2}.
\]
\State Set the QSVT approximation error budget:
\[
\varepsilon_{\rm QSVT}\gets \frac{\sqrt{n_{\mathfrak q}}}{{\cal N}(\brm{\omega}^{k})}\min\Big\{\frac{\Delta^{2}}{6\,K_{\mathfrak q}^{s}},\frac{\Delta}{12}\Big\}.
\]
\State Construct a QSVT circuit $\brm{U}_{\Pi_{\mathfrak q}^{s}}$ with error $\varepsilon_{\rm QSVT}$  \label{line:qsvt-synth}
\State Construct a unitary $\brm{U}_{\brm{\Omega}_{\ast}^{k}}:=(\brm{1}_{n}\otimes\ket{0}\bra{0}^{\otimes b}\otimes\brm{U}_{\Pi_{\mathfrak{q}}^{s}})(\UOmega\otimes\brm{1}_{3})$
\State Apply $\brm{U}_{\brm{\Omega}_{\ast}^{k}}$ to the data register of $\ket{0}^{\otimes(n+b+3)}$  preparing $\ket{\brm{\Omega}^{q}_{\ast}}=\frac{1}{\beta}\,\widetilde{\Pi}_{\mathfrak q}^{s}\ket{\brm{\omega}^{k}}\ket{0}^{\otimes(b+3)}+\ket{\perp_{\brm{\Omega}^{q}_{\ast}}}$
\State Run amplitude estimation with error $\varepsilon_{\rm enc}$ on $\ket{\brm{\Omega}^{q}_{\ast}}$ with marked state subspace on $\ket{0}^{\otimes(b+3)}$ with
\begin{align*}
    r= O\left(\varepsilon_{\rm AE}^{-1}\log\!\left(\frac{1}{\delta}\right)\right)
\end{align*}
queries to obtain an estimate $\widehat{p}$, that is $\varepsilon_{\rm AE}$-close to $p=\|\widetilde{\Pi}_{\mathfrak q}^{s}\ket{\brm{\omega}^{k}}\|^2_2/\beta^2$
\State Compute a fixed-point encoding $\widetilde{K}_{\mathfrak q}$ of $K_{\mathfrak q}$ with additive error $\varepsilon_{\rm enc}$.
\State {Run inequality testing} to compare $\widehat{p}$ against
\[
\tau \;:=\;\frac{n\,n_{\mathfrak q}}{4\,\kappa_{\mathfrak m}^{2}\,\beta^{2}\,{\cal N}(\brm{\omega}^{k})^{2}}\;\widetilde{K}_{\mathfrak q}^{2}.
\]
\State\Return $z=1$ iff $\widehat{p}>\tau$ {(equivalently $K_{\mathfrak{q}}^{s}>K_{\mathfrak{q}}$)}; otherwise output $z=0$.
\end{algorithmic}
\end{algorithm}}
  }
\end{figure*}

Task 2 is a binary decision problem: given a coupling constant $K_{\mathfrak q}>0$ and a promise gap $|K_{\mathfrak q}-K_{\mathfrak q}^{s}|\ge \Delta>0$, decide whether the no-phase-locking (NPL) condition $K_{\mathfrak q}<K_{\mathfrak q}^{s}$ holds. Recall that for $\mathfrak q\in\{k-1,k+1\}$, the critical value is
\begin{equation}\label{eq:Ks_def_app}
K_{\mathfrak q}^{s}:=\frac{\|\brm{\omega}_{\ast}^{\mathfrak q}\|_{2}}{\sqrt{n_{\mathfrak q}}},
\end{equation}
where $\brm{\omega}_{\ast}^{\mathfrak q}$ is the (projected) natural-frequency vector defined in \Cref{eq:no_phase_locking_state}. Hence, the core computational primitives are \emph{norm estimation} and \emph{value comparison}. As such, the algorithm estimates $\|\brm{\omega}_{\ast}^{\mathfrak q}\|_{2}^{2}/n_{\mathfrak q}=(K_{\mathfrak q}^{s})^{2}$ to sufficient additive precision and compare it against $K_{\mathfrak q}^{2}$ under the promise{d} gap.

{We obtain the projected frequency vector $\brm{\omega}_{\ast}^{\mathfrak q}$ from $\brm{\omega}^{k}$} by applying a projection operator
\begin{equation}\label{eq:NPL_projection}
    \Pi_{k-1}^{s}:=\left(\brm{L}_{k-1}^{[+]}\right)^{+}\bk
    \quad\text{and}\quad
    \Pi_{k+1}^{s}:=\left(\brm{L}_{k+1}^{[-]}\right)^{+}\bkk^{\rm T},
\end{equation}
so that $\brm{\omega}_{\ast}^{\mathfrak q}=\Pi_{\mathfrak q}^{s}\,\brm{\omega}^{k}$. In what follows, we (i) give a QSVT-based implementation of a \emph{rescaled} block encoding of $\Pi_{\mathfrak q}^{s}$ and (ii) show that amplitude estimation yields a provably correct classifier with the claimed query complexity in \Cref{thm:no-phase-locking}.

Let $\ket{\brm{\omega}^{k}}$ be the normalized state encoded in $\ket{\brm{\Omega}^{k}}$ (see \Cref{def:prob-SKM}), i.e., $\ket{\brm{\omega}^{k}}=\brm{\omega}^{k}/{\cal N}(\brm{\omega}^{k})$ on the {``success''} branch of $\UOmega$. The algorithm constructs a marked state whose success probability is
\begin{equation}\label{eq:p_success_def}
p := \frac{1}{\beta^{2}}\Big\|\widetilde{\Pi}_{\mathfrak q}^{s}\ket{\brm{\omega}^{k}}\Big\|_{2}^{2},
\end{equation}
where $\beta$ is {the} rescaling factor of $\ket{\brm{\omega}^{k}}$ defined in \Cref{def:prob_SKM_rescaling_factors} and $\widetilde{\Pi}_{\mathfrak q}^{s}$ is the rescaled approximation of $\Pi_{\mathfrak q}^{s}$ {constructed} from the QSVT circuit in \Cref{lem:BE_of_NPL_projection}. In the ideal case (i.e., the QSVT error $\varepsilon_{\rm QSVT}=0$), $p\propto(K_{\mathfrak q}^{s})^{2}$. Then, the amplitude estimation~\cite{Brassard2000AmplitudeAmpli} recovers $p$ to additive error $\varepsilon_{\rm AE}$, and a single quantum arithmetic comparison~\cite{Cuccaro2004ACircuit} implements the promised-gap classifier.

\begin{lemma}[Block encoding of the projection operator $\Pi_{\mathfrak{q}}^{s}$]\label{lem:BE_of_NPL_projection}
    Let $(\mathfrak{m},\mathfrak{q})$ be a pair of indices such that $\mathfrak{m}=k$ if $\mathfrak{q}=k-1$ and $\mathfrak{m}=k+1$ if $\mathfrak{q}=k+1$. Assume access to the PUE $\brm{B}_{\mathfrak{m}}$ from \Cref{eq:boundary-pue} and its inverse. For each $\mathfrak{q}$ and $\varepsilon_{\rm QSVT}\in(0,1/2)$, there exists a quantum circuit implementing a block encoding of $\Pi_{\mathfrak{q}}^{s}$ as defined in \Cref{eq:NPL_projection} such that
    \begin{align*}
        \widetilde{\Pi}_{\mathfrak q}^{s}:=\big(\brm{1}\otimes\bra{0}^{\otimes 3}\big)\,\brm{U}_{\Pi_{\mathfrak q}^{s}}\,\big(\brm{1}\otimes\ket{0}^{\otimes 3}\big)
    \end{align*} satisfies 
    \begin{equation}\label{eq:BE_of_NPL_projection_clean}
        \Bigg\|\Pi_{\mathfrak q}^{s}-\frac{2\,\kappa_{\mathfrak m}}{\sqrt{n}}\;\widetilde{\Pi}_{\mathfrak q}^{s}\Bigg\|_{2}\le\varepsilon_{\rm QSVT}.
    \end{equation}
    The algorithm makes
    $O(d_{\mathfrak{m}})$ calls to $\brm{U}_{\bk}$ and its inverse for $\mathfrak{q}=k-1$, or to $\brm{U}_{\bkk^{\rm T}}=\brm{U}_{\bkk}^{\rm T}$ and its inverse for $\mathfrak{q}=k+1$, and $O(d_{\mathfrak{m}})$ additional quantum gates, with
    \begin{align}\label{eq:degree_NPL_poly}
        d_{\mathfrak{m}}= O\left(\kappa_{\mathfrak{m}}^{2}\,\log{(\sqrt{n}\,\kappa_{\mathfrak{m}}/\varepsilon_{\rm QSVT})}\right).
\end{align}
    Here, $1/\kappa_{\mathfrak{m}}\in \big(0,\sqrt{n}/\zeta^{(\mathfrak{m})}_{\min}\big)$ for $\mathfrak{m}\in\{k,k+1\}$ and $\zeta^{(\mathfrak{m})}_{\min}$ is the minimum nonzero singular value of $\brm{B}_{m}$. 
\end{lemma}
\begin{proof}
    For $\mathfrak{m}\in\{k,k+1\}$, note that there exists an odd degree polynomial $p_{\mathfrak{m}}(x)=x\,g_{\mathfrak{m}}(x^2)$, where $g_{\mathfrak{m}}(x)$ is a polynomial of degree given in \Cref{eq:degree_NPL_poly} defined on $\mathcal{D}_{\mathfrak{m}}=[-1,-1/\kappa_{\mathfrak{m}}^{2}]\,\cup\,[1/\kappa_{\mathfrak{m}}^{2},1]$ satisfying $|g_{\mathfrak{m}}(x)|\le 1$ and $|1/(2\kappa_{\mathfrak{m}}^{2}x)-g_{\mathfrak{m}}(x)|\leq \varepsilon_{\rm QSVT}/(2\kappa_{\mathfrak{m}}^{2})$ for all $x\in\mathcal{D}_{\mathfrak{m}}$~\cite{Martyn2021GrandAlgorithms}. Next, define a pair of indices $(\mathfrak{m},\mathfrak{q})$ such that $\mathfrak{m}=k$ if $\mathfrak{q}=k-1$ and $\mathfrak{m}=k+1$ if $\mathfrak{q}=k+1$. Then, denote the approximate block encoding of $\Pi_{\mathfrak{q}}^{s}$ defined in \Cref{eq:NPL_projection} as $\brm{U}_{\Pi_{\mathfrak{q}}^{s}}$. For $\mathfrak{q}=k-1$, Ref.~\cite{leditto2024quantumhodgeranktopologybasedrank} showed that $\brm{U}_{\Pi_{\mathfrak{q}}^{s}}$ can be constructed using the QSVT of the $\brm{U}_{\brm{B}_{\mathfrak{m}}}$ defined in \Cref{eq:boundary-pue} with polynomial $p_{\mathfrak{m}}(x)$. Here, we also use the same polynomial to construct the approximate block encoding of the projector for $\mathfrak{q}=k+1$ via QSVT of $\brm{U}_{\brm{B}_\mathfrak{m}^{\rm T}}=\brm{U}_{\brm{B}_{\mathfrak{m}}}^{\rm T}$. Thus, we have
\begin{align}\label{eq:BE_of_NPL_projection}
    \brm{U}_{\Pi_{\mathfrak{q}}^{s}}=
    \begin{cases}
        \brm{U}_{p_{\mathfrak{m}}(\brm{B}_{\mathfrak{m}})}\quad&\text{for $\mathfrak{q}=k-1$},\\
        \brm{U}_{p_{\mathfrak{m}}(\brm{B}_{\mathfrak{m}}^{\rm T})}\quad&\text{for $\mathfrak{q}=k+1$}
    \end{cases}
\end{align}
where $\widetilde{\Pi}_{\mathfrak{q}}^{s}:=\big(\brm{1}_{n}\otimes\bra{0}^{\otimes 3}\big)\brm{U}_{\Pi_{\mathfrak{q}^{s}}}\big(\brm{1}_{n}\otimes\ket{0}^{\otimes 3}\big)$ satisfying $\big\|\Pi_\mathfrak{q}^{s}-\frac{2\kappa_{\mathfrak{m}}^{2}}{\sqrt{n}}\,\widetilde{\Pi}_\mathfrak{q}^{s}\big\|_{2}\leq\varepsilon_{\rm QSVT}$, for $\varepsilon_{\rm QSVT}\in(0,1/2)$. Given the description of $p_{\mathfrak{m}}(x)$, the QSVT implementation~\cite{Gilyen2019QuantumArithmetics} requires $O(d_{\mathfrak{m}})$ calls to $\brm{U}_{\bk}$ and its inverse for $\mathfrak{q}=k-1$, or to $\brm{U}_{\bkk^{\rm T}}$ and its inverse for $\mathfrak{q}=k+1$, and $O(d_{\mathfrak{m}})$ additional quantum gates.
\end{proof}

\begin{proof}[Proof of \Cref{thm:no-phase-locking} {(correctness and complexity of \Cref{alg:T2})}] 

Here, we validate correctness (promise-gap classification) and derive the stated query and gate complexity in \Cref{thm:no-phase-locking}. We first compose the state preparation unitary $\UOmega$ with the block encoding of $\Pi_{\mathfrak{q}}^{s}$ from \Cref{lem:BE_of_NPL_projection}:
\begin{align}
    \brm{U}_{\brm{\Omega}_{\ast}^{k}}:=(\brm{1}_{n}\otimes\ket{0}\bra{0}^{\otimes b}\otimes\brm{U}_{\Pi_{\mathfrak{q}}^{s}})(\UOmega\otimes\brm{1}_{3}),
\end{align}
such that $\brm{U}_{\brm{\Omega}_{\ast}^{k}}$ probalistically prepares a linear combination of $\ket{\brm{\omega}_\ast^{\mathfrak{q}}}$ and garbage state. Explicitly,
\begin{align}
    \brm{U}_{\brm{\Omega}_{\ast}^{k}}\ket{0}^{\otimes(n+b+3)}=\ket{\brm{\Omega}^{q}_{\ast}}:=\frac{1}{\beta}\,\widetilde{\Pi}_{\mathfrak q}^{s}\ket{\brm{\omega}^{k}}\ket{0}^{\otimes(b+3)}+\ket{\perp_{\brm{\Omega}^{q}_{\ast}}},\qquad\text{ where }\big(\brm{1}_{n}\otimes\bra{0}^{\otimes(b+3)}\big)\ket{\perp_{\brm{\Omega}^{q}_{\ast}}}=0).
\end{align}
On input $\ket{0}^{\otimes(b+3)}$, the success amplitude is proportional to $\widetilde{\Pi}_{\mathfrak q}^{s}\ket{\brm{\omega}^{k}}$ and therefore the success probability equals $p$ from \Cref{eq:p_success_def}.

Recall that $\varepsilon_{\rm QSVT}$ is the QSVT error from \Cref{lem:BE_of_NPL_projection}. In the ideal case where $\varepsilon_{\rm QSVT}=0$, \Cref{eq:BE_of_NPL_projection_clean} implies $\widetilde{\Pi}_{\mathfrak q}^{s}=(\sqrt{n}/(2\kappa_{\mathfrak m}))\Pi_{\mathfrak q}^{s}$, hence
\begin{align}
p_{\rm ideal}&=\frac{1}{\beta^{2}}\left\|\frac{\sqrt{n}}{2\kappa_{\mathfrak m}}\Pi_{\mathfrak q}^{s}\ket{\brm{\omega}^{k}}\right\|_{2}^{2}=\frac{n}{4\kappa_{\mathfrak m}^{2}\beta^{2}}\cdot\frac{\|\Pi_{\mathfrak q}^{s}\brm{\omega}^{k}\|_{2}^{2}}{{\cal N}(\brm{\omega}^{k})^{2}}\nonumber\\
&=\frac{n}{4\kappa_{\mathfrak m}^{2}\beta^{2}}\cdot\frac{\|\brm{\omega}^{\mathfrak q}_{\ast}\|_{2}^{2}}{{\cal N}(\brm{\omega}^{k})^{2}}=\frac{n\,n_{\mathfrak q}}{4\kappa_{\mathfrak m}^{2}\beta^{2}{\cal N}(\brm{\omega}^{k})^{2}}\left(K_{\mathfrak q}^{s}\right)^{2}.
\label{eq:pideal_vs_Ks}
\end{align}
Thus, from QSVT implementation yielding $p$ in \Cref{eq:p_success_def}, we obtain
\begin{align}\label{eq:Delta_QSVT_clean}
\big|p-p_{\rm ideal}\big|&=\frac{1}{\beta^{2}}\left|\left\|\widetilde{\Pi}_{\mathfrak q}^{s}\ket{\brm{\omega}^{k}}\right\|_{2}^{2}-\left\|\frac{\sqrt{n}}{2\kappa_{\mathfrak m}}\Pi_{\mathfrak q}^{s}\ket{\brm{\omega}^{k}}\right\|_{2}^{2}\right|
\nonumber\\
&\le \frac{1}{\beta^{2}}\cdot \frac{n}{4\kappa_{\mathfrak m}^{2}}\left(2\left\|\Pi_{\mathfrak q}^{s}\ket{\brm{\omega}^{k}}\right\|_{2}\varepsilon_{\rm QSVT}+\varepsilon_{\rm QSVT}^{2}\right)
\nonumber\\
&= \frac{1}{\beta^{2}}\cdot \frac{n}{4\kappa_{\mathfrak m}^{2}}\left(\frac{2\sqrt{n_{\mathfrak q}}\,K_{\mathfrak q}^{s}}{{\cal N}(\brm{\omega}^{k})}\varepsilon_{\rm QSVT}+\varepsilon_{\rm QSVT}^{2}\right)=: \Delta_{\rm QSVT}.
\end{align}
We then perform amplitude estimation with precision $\varepsilon_{\rm AE}$ and failure probability $\delta$ to estimate $\hat{p}$ such that 
\begin{align}\label{eq:eq_epsilon_AE}
    |\hat p-p|\le \varepsilon_{\rm AE}.
\end{align}
We next compare $\hat{p}$ with $K_{\mathfrak{q}}$ via quantum arithmetics. Let $\widetilde{K}_{\mathfrak{q}}$ be the value that is encoded in the quantum register. Denote
    \begin{align}\label{eq:epsilon_enc}
        \tau_{\rm ideal} = \frac{n\,n_{\mathfrak q}}{4\kappa_{\mathfrak m}^{2}\beta^{2}{\cal N}(\brm{\omega}^{k})^{2}}K_{\mathfrak q}^{2},\quad\tau= \frac{n\,n_{\mathfrak q}}{4\kappa_{\mathfrak m}^{2}\beta^{2}{\cal N}(\brm{\omega}^{k})^{2}}\widetilde{K}_{\mathfrak q}^{2},\quad\text{ with }|\tau-\tau_{\rm ideal}|\le \varepsilon_{\rm enc}
    \end{align}
    By the promise $|K_\mathfrak{q}^{s}-K_\mathfrak{q}|\ge\Delta\ge0$ and $K_\mathfrak{q}^{s},K_\mathfrak{q}>0$, we have $\big|(K_{\mathfrak{q}}^{s})^{2}-K_{\mathfrak{q}}^2\big|=\big|K_{\mathfrak{q}}^{s}-K_{\mathfrak{q}}\big|\big(K_{\mathfrak{q}}^{s}+K_{\mathfrak{q}}\big)\ge\Delta^{2}$. Therefore the \emph{ideal} separation between $p_{\rm ideal}$ (from \eqref{eq:pideal_vs_Ks}) and $\tau_{\rm ideal}$ is at least
    \begin{align}
        |p_{\rm ideal}-\tau_{\rm ideal}|\ge \frac{n\,n_{\mathfrak q}}{4\kappa_{\mathfrak m}^{2}{\cal N}(\brm{\omega}^{k})^{2}\beta^{2}}\Delta^{2}.
    \end{align}
    From \Cref{eq:Delta_QSVT_clean,eq:eq_epsilon_AE,eq:epsilon_enc} and by triangle inequality, the above inequality yields
    \begin{align}
        |\hat{p}-\tau|\geq\frac{n\,n_{\mathfrak q}}{4\kappa_{\mathfrak m}^{2}{\cal N}(\brm{\omega}^{k})^{2}\beta^{2}}\Delta^{2}-\left(\Delta_{\rm QSVT}+\varepsilon_{\rm AE}+\varepsilon_{\rm enc}\right)
    \end{align}
    Then, it suffices to enforce
    \begin{align}\label{eq:NPL_error_sum}
        \Delta_{\rm QSVT}+\varepsilon_{\rm AE}+\varepsilon_{\rm enc}\le\frac{n\,n_\mathfrak{q}}{4\,\kappa_{\mathfrak{m}}^{2}\,{\cal N}(\brm{\omega}^{k})^{2}\,\beta^{2}}\,\Delta^{2},
    \end{align}
    which guarantees that the sign of $\widehat{p}-\tau$ matches the sign of $K_{\mathfrak q}^{s}-K_{\mathfrak q}$ and hence prevents misclassification {with probability $1-\delta$}. Finally, we can set
    \begin{align}\label{eq:NPL_error_budget}
        \varepsilon_{\rm AE}=\varepsilon_{\rm enc}=\frac{n\,n_\mathfrak{q}}{12\,\kappa_{\mathfrak{m}}^{2}\,{\cal N}(\brm{\omega}^{k})^{2}\,\beta^{2}}\,\Delta^{2},\quad\text{and}\quad\varepsilon_{\rm QSVT}=\frac{\kappa_{\mathfrak{m}\,\sqrt{n}_{q}}}{\mathcal{N}(\brm{\omega}^{k})}\,\min\left\{\frac{\Delta^{2}}{6\,K_{\mathfrak{q}}^{s}},\frac{\Delta}{12}\right\}
    \end{align}
    to ensure \Cref{eq:NPL_error_sum} holds. The value $\varepsilon_{\rm QSVT}$ is obtained by solving \Cref{eq:pideal_vs_Ks} with $\Delta_{\rm QSVT}=\frac{n\,n_\mathfrak{q}}{12\,\kappa_{\mathfrak{m}}^{2}\,{\cal N}(\brm{\omega}^{k})^{2}\,\beta^{2}}\,\Delta^{2}$.
    
    The complexity of this classification is then given by (i) first  implementing QSVT step, as stated in \Cref{lem:BE_of_NPL_projection}. The block encoding $\brm{U}_{\brm{B}_{\mathfrak{m}}}$ given in \Cref{lem:pue_boundary_correctness} requires calling $O(1)$ times to $\brm{O}_{m_{k}}$ and $\brm{O}_{m_{k-1}}$ for $\mathfrak{m}=k$, or to $\brm{O}_{m_{k}}$ and $\brm{O}_{m_{k+1}}$ for $\mathfrak{m}=k+1$, with $O(n)$ gates from \Cref{eq:pue_boundary_operator_Dirac}. Thus, in total, the QSVT subroutine calls 
    \begin{align}
        d_{\mathfrak{m}}=O\left(\kappa_{\mathfrak{m}}^{2}\log{\left(\frac{\mathcal{N}(\brm{\omega}^{k})}{\min\{\Delta^{2}/(6\,K_{\mathfrak{q}}^{s}),\Delta/12\}}\left(\frac{n}{n_{\mathfrak{q}}}\right)^{1/2}\right)}\right)
    \end{align}
    times to simplex oracle $\brm{O}_{m_{p}}$ for $p=\{k-1,k,k+1\}$. Next, the amplitude estimation~\cite{Brassard2000AmplitudeAmpli} step calls
    \begin{align}\label{eq:repetition_npl_amplitude_est}
        r_{\mathfrak{m}}=O\left(\frac{\kappa_{\mathfrak{m}}^{2}\,{\cal N}(\brm{\omega}^{k})^{2}\,\beta^{2}}{n\,n_\mathfrak{q}\,\Delta^{2}}\,\log\left(\frac{1}{\delta}\right)\right)
    \end{align}
    times to the QSVT step and $\UOmega$. Thus, up to the second step, the query complexity to the oracle $\brm{O}_{m_{p}}$ is $O(d_{\mathfrak{m}}\,r_{\mathfrak{m}})$, whereas it is $O(r_{\mathfrak{m}})$ to $\UOmega$. These two steps dominate the algorithm's query complexity. 
    
    The encoding step of the coupling constant $K_\mathfrak{q}$ with $\log_{2}{(1/\varepsilon_{\rm enc})}$-bit precision requires $O({\rm polylog}(\varepsilon_{\rm enc}^{-1}))=O({\rm polylog}(r_{\mathfrak{m}})))$ gates~\cite{Mitarai2019Analog-Digital}. And, the comparison of the amplitude estimation result $\hat{x}$ with the encoded coupling constant $\widetilde{K}_\mathfrak{q}$ by {running inequality testing} via quantum arithmetic~\cite{Cuccaro2004ACircuit} incurs $O({\rm polylog}(\varepsilon_{\rm AE}^{-1})))=O({\rm polylog}(r_{\mathfrak{m}})))$ gates~\cite{Cuccaro2004ACircuit}. Thus, the additional gate complexity (outside $\Om$ and $\UOmega$), including the Dirac operator gate complexity (used in the QSVT step), is then given by $O(n\,d_{\mathfrak{m}}\,r_{\mathfrak{m}})$. This concludes the proof.
\end{proof}

\section{Efficient state preparation for the simplicial node-aggregated natural frequency}\label{app:node-aggregated_KM}

\begin{figure}
    \centering
    \includegraphics[width=1\linewidth]{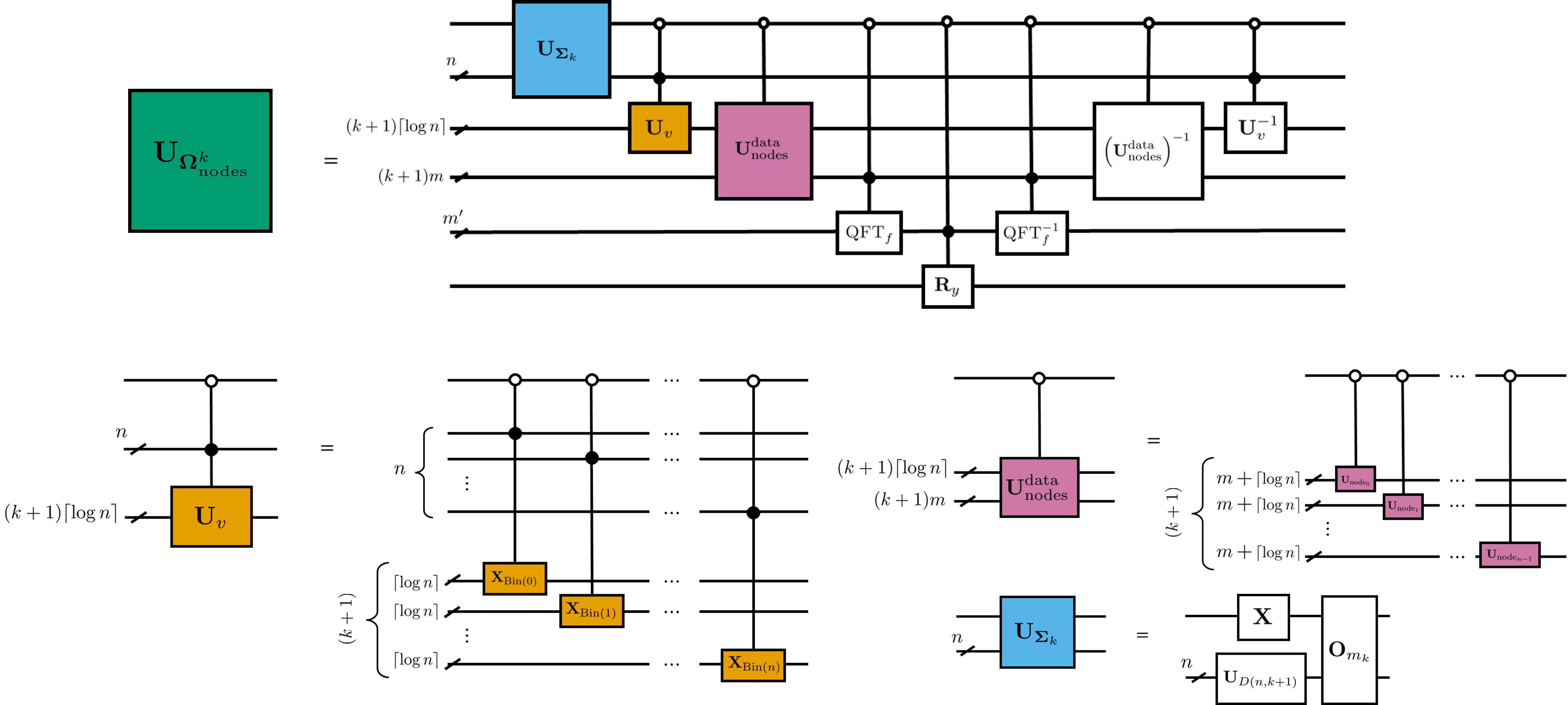}
    \caption{(Top panel) The quantum circuit for probabilistic state preparation unitary $\brm{U}_{\brm{\Omega}^{k}}=\brm{U}_{\brm{\Omega}^{k}_{\mathrm{nodes}}}$ preparing simplicial node-aggregated natural frequency. The gate complexity of the quantum circuit is $O(n+nk+|\mathcal{E}|)$ gates  (with additional polylogarithmic gates) to prepare SKM input that is node aggregated and resides on clique complexes $\mathcal{K}_{n}(G)$. (Bottom panel) Subroutines used in the state preparation unitary $\brm{U}_{\brm{\Omega}^{k}_{\rm nodes}}$. Application of control multi-$\mathrm{NOT}$ (or $\mathrm{CNOT}_{\lceil\log{n}\rceil}$) that build $\brm{U}_{v}$ that writes binary strings of all nodes $j\in[n]$; (Top right) The construction of QRAM/QROM $\brm{U}_{\rm nodes}^{\rm data}$ encoding node natural frequencies $\{\omega^{0}_{j}\}_{j\in[n]}$; (Bottom right) Specification of $\USigma$ if the simplicial phase is defined on clique complexes.}
    \label{fig:prob_theta_node_aggregated}
\end{figure}

In real-world scenarios, the presence and definition of higher-order interactions are not always evident a priori, as such relationships may be implicit within the data or emerge only after appropriate aggregation or inference. In many cases, the data comes from nodes. Thus, a natural way to define higher-order network data is by aggregating data from nodes. As commonly implemented in a sensor network framework~\cite{Rajagopalan2006DataAggregation} estimations of some local node data are collected and computed to obtain global data. The function that takes as input node-aggregated data is called an \textit{aggregate function} $f$. 

\begin{figure}[t!]
  \centering
  \fbox{
    \parbox{1\textwidth}{

\begin{algorithm}[H]
\caption{Simplicial Node–Aggregated Natural-Frequency State Preparation $\brm{U}_{{\brm{\Omega}^{k}_{\rm nodes}}}$}
\label{alg:node-aggregated-state}
\begin{algorithmic}[1]
\renewcommand{\algorithmicrequire}{\textbf{Input:}} 

\renewcommand{\algorithmicensure}{\textbf{Output:}} 

\Require\Statex\begin{itemize}[leftmargin=*, nosep]
    \item Node natural frequencies $\{\omega^0_j\}_{j\in[n]}$
    \item QRAM/QROM access $\brm{U}_{{\rm node}_{j}}$ encoding $\ket{\omega^{0}_{j}}$
    \item Oracle $\brm{U}_{\brm{\Sigma}_k}$ preparing $\ket{\brm{\Sigma}_{k}}$
    \item A QFT-admissible function $f:\mathbb{R}^{k+1}\to\mathbb{R}$ satisfying \Cref{def:QFT-multi-functions} 
    \item Bound $\Lambda_f$ in \Cref{eq:Lambda_f_def} and Lipschitz bound $L_f$ in \Cref{eq:L_f_def}
    \item Fixed-point precisions $m$ and $m'\le m-\log{L_f}$
\end{itemize}
\Ensure A quantum state $\ket{\overline{\brm{\Omega}}^k_{\mathrm{nodes}}}$ as defined in \Cref{eq:approximate_simplicial_node-aggregated_state}

\Statex
\Statex\textbf{Registers:}
\State $S$: simplex register ($n$ qubits) 
\State $A$: ancilla register from $\brm{U}_{\brm{\Sigma}_k}$ ($a_k$ qubits)
\State $V$: node register ($(k+1)\lceil\log{n\rceil}$ qubits)
\State $W$: node‑data work register ($(k+1)\lceil\log m\rceil$ qubits)
\State $W_f$: quantum arithmetic work register ($\lceil\log m'\rceil$ qubits)
\State $F$: success flag register (a single qubit register)

\Statex
\Statex \textbf{Main Algorithm:} \State Intialize $\ket{0}_S\ket{0}_A\ket{0}_V\ket{0}_W\ket{0}_{W_f}\ket{0}_F$
\State Apply $\brm{U}_{\brm{\Sigma}_k}$ on $(S,A)$ to prepare {$\ket{\brm{\Psi}_{1}}$ in}  \Cref{eq:initial_node_aggregated_state}\label{line:USigma}
\State Conditioned on $A=\ket{0}^{\otimes a_k}$, apply $\brm{U}_v$ defined in \Cref{eq:U_v} to $(S,V)$ to prepare {$\ket{\brm{\Psi}_{2}}$ in}  \Cref{eq:vertex_index_extract_state}\label{line:nodewrite}
\State Apply $\brm{U}^{\rm data}_{\rm nodes}:=\bigotimes_{i=0}^{k}\brm{U}^{\rm data}_{\mathrm{node}_i}$, where $\brm{U}^{\rm data}_{\mathrm{node}_i}$ defined in \Cref{eq:U_node_i}  acting on $(V,W)$ to prepare {$\ket{\brm{\Psi}_{3}}$ in} Eq.~\eqref{eq:node_data_load_state} \label{line:nodeload}
\State Apply $\textsc{QFT}_{f}$ defined in \Cref{eq:QFT_f} computing $f_i$ from registers $(W_0,\dots,W_k)$ into register $W_f$ with $m'$-bit precision\label{line:qft}
\State Apply $\brm{U}_{\rm amp}$ as defined in \Cref{lem:digita-analog_encoding} to $(W_f,F)$ to prepare {$\ket{\brm{\Psi}_{4}}$ in} Eq.~\eqref{eq:amplitude_transduce_state} \label{line:amptrans}
\State Uncompute registers $V,W,W_f$ by reversing operations in \Cref{line:nodewrite,line:nodeload,line:qft} to output $\ket{\overline{\brm{\Omega}}^{k}_{\rm nodes}}$
\label{line:uncompute}
\end{algorithmic}
\end{algorithm}
}

    }
\end{figure}

\begin{definition}[Simplicial node-aggregated function]\label{def:simplicial_aggregated_function}
    Let $\{\omega^0_j\}_{j\in[n]}$ be node natural frequencies. For each $k$-simplex $\sigma_k^i=\{v^{i,0},\dots,v^{i,k}\}$, define the simplicial node-aggregated natural frequency component
    \begin{equation*}
        (\omega^k_{\rm nodes})_i := f\!\big(\omega^0_{i,0},\dots,\omega^0_{i,k}\big),
        \qquad \forall i\in[n_k].
    \end{equation*}
    We call $f$ a simplicial node-aggregated function.
\end{definition}

To keep the end-to-end state preparation efficient, we require that $f$ be implementable by a (standard) quantum arithmetic (e.g., QFT-based add/multiply~\cite{Cuccaro2004ACircuit}), and we must make explicit the scaling needed for digital-to-analog conversion~\cite{Mitarai2019Analog-Digital}, which we call amplitude transduction.
\begin{definition}[QFT-admissible multivariable functions]\label{def:QFT-multi-functions}
A function $f:\mathbb{R}^{k+1}\to\mathbb{R}$ is \textit{QFT-admissible} if there exists a reversible quantum circuit that, given $(k+1)$ fixed-point inputs (each encoded with $m= O(1)$ bits), computes a fixed-point encoding of $f$ with $m'= O(1)$ bits of precision using $\mathrm{poly}(k)$ one- and two-qubit gates and only basic arithmetic operations (addition/subtraction/multiplication), e.g., via QFT-based arithmetic~\cite{Cuccaro2004ACircuit}. Moreover, we assume there exists $\mathfrak{c}= O(1)$ such that the number of gates is $\mathrm{poly}(k)=O(k^{\mathfrak{c}})$.
\end{definition}
Digital-to-analog conversion for multivariable functions requires bounded magnitudes. We therefore assume a known bound $\Lambda_f>0$ such that
\begin{equation}
\label{eq:Lambda_f_def}
    \big|f(\omega^0_{0},\dots,\omega^0_{k})\big|\le \Lambda_f
    \quad \text{for all } (\omega^0_{0},\dots,\omega^0_{k})\in\mathbb{R}^{k+1}.
\end{equation}
We also define a Lipschitz constant that we will use for precision propagation:
\begin{equation}
\label{eq:L_f_def}
    L_f \ :=\ \sum_{t=0}^{k}\ \sup_{\omega^0\in\mathbb{R}}\left|\frac{\partial f}{\partial \omega^0_t}(\omega^0)\right|.
\end{equation}
Note that, for the QFT-admissible function{s} $f$ we consider here (e.g., affine maps and low-degree polynomials of bounded inputs), $L_f$ is finite.

We now define a simplicial node-aggregated natural frequency state
\begin{eqnarray}\label{eq:node-aggregated signal state}
    \lvert\brm{\omega}^{k}_{\mathrm{nodes}}\rangle=\frac{1}{\mathcal{N}(\brm{\omega}^{k}_{\mathrm{nodes}})}\sum_{i\in\left[n_{k}\right]}\,f(\omega^{0}_{i,0},\cdots,\omega^{0}_{i,k})\ket{\sigma_{k}^{i}},
\end{eqnarray}
with $\mathcal{N}^2(\brm{\omega}^{k}_{\mathrm{nodes}}):=\sum_{i\in[n_{k}]}\Big(f(\omega^{0}_{i,0},\cdots,\omega^{0}_{i,k})\Big)^{2}$. We prepare such a state using the probabilistic state preparation unitary denoted as $\brm{U}_{\brm{\Omega}^{k}_{\mathrm{nodes}}}$ such that
\begin{eqnarray}
    \ket{\brm{\Omega}^{k}_{\mathrm{nodes}}}:=\frac{1}{\beta_{\rm nodes}}\,\lvert\brm{\omega}^{k}_{\mathrm{nodes}}\rangle\ket{0}^{\otimes b_{\rm nodes}}+\ket{\perp_{\brm{\Omega}_{\rm nodes}^k}},\label{eq:node-aggregated signal state plus garbage}
\end{eqnarray}
where $(\brm{1}\otimes\bra{0}^{\otimes b_{\rm nodes}})\ket{\perp_{\brm{\Omega}_{\rm nodes}}}=0$ and the rescaling factor is
\begin{equation}
\label{eq:beta_nodes_corrected}
    \beta_{\rm nodes}:=\mu_{k}\,\Lambda_f\,\frac{\sqrt{n_k}}{\mathcal{N}(\brm{\omega}^{k}_{\mathrm{nodes}})}.
\end{equation}
We explain the procedure in \Cref{alg:node-aggregated-state} and the gate complexity of $\brm{U}_{\brm{\Omega}^{k}_{\mathrm{nodes}}}$ based on the two following lemmas. \begin{lemma}[(Sparse) quantum digital-analog conversion~\cite{Mitarai2019Analog-Digital}]\label{Quantum digital-analog conversion} Let $\ket{\brm{\psi}}=1/\mathcal{N}(\psi)\,\sum_{i\in[N_{\psi}]}\psi_{i}\ket{v_{i}}$ be an $n$-qubit quantum state with $N_{\psi}<2^{n}$, and $\ket{\widetilde{\psi}_{i}}$ be an $m$-bit precision binary representation of $\psi_{i}$ in an $m$-qubit register. Given an $(n+m)$-qubit unitary $\brm{U}^{\mathrm{data}}$ that prepares $1/\sqrt{N_{\psi}}\sum_{i\in[N_{\psi}]}\ket{v_{i}}\ket{\widetilde{\psi}_{i}}$, there exists a {unitary $\brm{U}_{\rm amp}$} that outputs 
\begin{eqnarray}\label{lem:digita-analog_encoding}
    \frac{\mathcal{N}(\psi)}{\sqrt{N_{\psi}}}\ket{\widetilde{\brm{\psi}}}\ket{0}^{\otimes (m+1)}+\sqrt{1-\frac{1}{N_{\psi}^2}}\ket{\perp},
\end{eqnarray} {such that $\|\ket{\widetilde{\psi}}-\ket{\psi}\|_2\leq O(2^{-m}\,\mathcal{N}(\psi)/N_{\psi})$} and $\big(\brm{1}_{(n+m)}\otimes\bra{0}\big)\ket{\perp}=0${. This unitary $\brm{U}_{\mathrm{amp}}$ uses} $O(1)$ controlled applications of $\brm{U}^{\mathrm{data}}$ and $O(\mathrm{poly}(m))$ single- and two-qubit gates. 
\end{lemma}
\noindent The above state preparation only uses controlled-single qubit rotation $\brm{R}_{y}$ gates.

\begin{lemma}[Simplicial node-aggregated natural frequency state preparation]\label{lem:simplicial node data encoding}
Given node natural frequencies $\{\omega^{0}_{j}\}_{j\in[n]}$, a simplicial node aggregated function $f$ (in \Cref{def:simplicial_aggregated_function}) satisfying \Cref{def:QFT-multi-functions} with a known bound $\Lambda_f$ in \Cref{eq:Lambda_f_def}, fixed-point precision $m$, and access to $\brm{U}_{\brm{\Sigma}_{k}}$, Algorithm~\ref{alg:node-aggregated-state} prepares \begin{align}\label{eq:approximate_simplicial_node-aggregated_state}
    \ket{\overline{\brm{\Omega}}^{k}_{\mathrm{nodes}}}:=\frac{1}{\overline{\beta}_{\rm nodes}}\lvert\overline{\brm{\omega}}^{k}_{\rm nodes}\,\rangle\ket{0}^{\otimes (a_k+1)}+\ket{\perp_{\overline{\brm{\Omega}}_{\rm nodes}^k}},
\end{align} 
with
\begin{align*}
    \ket{\overline{\brm{\omega}}^{k}_{\mathrm{nodes}}}:=\frac{1}{\mathcal{N}(\overline{\brm{\omega}}^{k}_{\mathrm{nodes}})}\sum_{i\in[n_k]}\,\overline{f}_i\,\lvert\sigma_{k}^{i}\rangle,\quad\overline{\beta}_{\rm nodes}:=\frac{\mu_k\,\Lambda_f\,\sqrt{n_k}}{\mathcal{N}(\overline{\brm{\omega}}^{k}_{\mathrm{nodes}})},\quad\mathcal{N}^2(\overline{\brm{\omega}}^{k}_{\mathrm{nodes}}):=\sum_{i\in[n_k]}\overline{f}_i^2
\end{align*}
such that $\big|\overline{f}_i-f(\omega^{0}_{i,0},\cdots,\omega^{0}_{i,k})\big|\leq O(2^{-m})$, or equivalently, $\big\|\ket{\overline{\brm{\omega}}^{k}_{\mathrm{nodes}}}-\ket{\brm{\omega}^{k}_{\mathrm{nodes}}}\big\|_2\leq O(2^{-m}\,\overline{\beta}_{\rm nodes})$. The implementation uses one call to $\brm{U}_{\brm{\Sigma}_{k}}$ with $O(nk\log n+k^{\mathfrak{c}})$ additional gates, {where $\mathfrak{c} = O(1)$ is as defined in \Cref{def:QFT-multi-functions}}. The total number of ancillas used in the implementation is $O\big(a_k+k(\log{n}+m)+\log{L_f}\big)$ qubits with $L_f$ is a Lipschitz constant defined in \Cref{eq:L_f_def}.
\end{lemma}

In this state preparation protocol, we use five qubit registers that are described as follows: register $S$ is simplex register ($n$ qubits), register $A$ is the ancilla register from $\brm{U}_{\brm{\Sigma}_k}$ ($a_k$ qubits), register $V$ is a node register ($(k+1)\lceil\log{n\rceil}$ qubits), register $W$ is a node‑data work register ($(k+1)\lceil\log m\rceil$ qubits), register $W_f$ is quantum arithmetic work register ($\lceil\log{(m-\log{L_f})}\rceil$ qubits), and register $F$ is success flag register (a single qubit register). {In the following protocol, we also construct unitaries
\begin{align}
    &\brm{U}_{v}:\ket{\sigma_k^i}_S\ket{0}_V\mapsto \ket{\sigma_k^i}_S\bigotimes_{j=0}^{k}\ket{\mathrm{Bin}(v^{i,j})}_{V_j},\label{eq:U_v}\\
    &\brm{U}_{{\rm nodes}_j}^{\rm data}:\ket{\mathrm{Bin}(v^{i,j})}_{V_j}\ket{0}_{W_j}^{\otimes m}\mapsto\ket{\mathrm{Bin}(v^{i,j})}_{V_j}\ket{\widetilde{\omega}_{j}^{0}}_{W_j}^{\otimes m},\quad\text{where $|\widetilde{\omega}^0_j-\omega^0_j|\le 2^{-m},
    \;\forall j\in[n]$,}\label{eq:U_node_i}\\
    &\textsc{QFT}_{f}:\bigotimes_{j\in[k]}\ket{\widetilde{\omega}^{0}_{i,j}}_{V_j}\ket{0}_{W_f}\mapsto\bigotimes_{j\in[k]}\ket{\widetilde{\omega}^{0}_{i,j}}_{V_j}\ket{\overline{f}_i}_{W_f},\quad\text{where $|\overline{f}_i-\widetilde{f}_i|\leq 2^{-m'}$ and $\widetilde{f}_i=f(\widetilde{\omega}_{i,0}^{0},\cdots,\widetilde{\omega}_{i,k}^{0})$,}\label{eq:QFT_f}\\
    &\brm{U}_{\rm amp}:\ket{\overline{f}_i}_{W_f}\ket{0}\mapsto y_i\ket{\overline{f}_i}_{W_f}\ket{0}+\sqrt{1-y_i^2}\ket{\perp}.\quad\text{where $y_i:=\frac{\overline{f}_i}{\Lambda_f}\in[-1,1].$}
\end{align} 
}

\begin{proof}[Proof of \Cref{lem:simplicial node data encoding}]
Start from $\ket{0}_{S}\ket{0}_{A}$, applying $\brm{U}_{\brm{\Sigma}_k}$ yields
\begin{align}\label{eq:initial_node_aggregated_state}
    \ket{\Psi_1}=\frac{1}{\mu_k\sqrt{n_k}}\sum_{i\in[n_k]} \ket{\sigma_k^i}_{S}\ket{0}^{\otimes a_k}_{A} + \ket{\perp_{\Sigma_{k}}}, \quad\textbf{where $\sigma_{k}^{i}=\{v^{i,0},\cdots,v^{i,k}\}$}
\end{align}
where $(\brm{1}\otimes\bra{0}^{\otimes a_k})\ket{\perp_{\Sigma_k}}=0$. 

Next, conditioned on $A=\ket{0}^{\otimes a_k}$, a unitary $\brm{U}_v$ writes the ordered vertex list $(v^{i,0},\dots,v^{i,k})$ into register $V$, i.e., $\brm{U}_v:\ket{\sigma_k^i}_S\ket{0}_V\mapsto \ket{\sigma_k^i}_S\bigotimes_{j=0}^{k}\ket{\mathrm{Bin}(v^{i,j})}_{V_j}$. This is done by applying controlled multi-$\mathrm{X}$ gates to registers $S,V$. Explicitly~\cite{Low2024tradingtgatesdirty},
\begin{align}\label{eq:vertex_index_extraction}
    \brm{U}_{v}:=\bigotimes_{i\in[n]}\big(\ket{1}\bra{1}_{i}\otimes\brm{X}_{\mathrm{Bin}(i)}\big)
\end{align}
such that conditioned on $\ket{1}$ of the $i$-th qubit in register $S$, $\brm{X}_{\mathrm{Bin}(i)}\ket{0}_{V}=\ket{\mathrm{Bin}(i)}_{V}$. This step creates
\begin{align}\label{eq:vertex_index_extract_state}
    \ket{\Psi_2}=\frac{1}{\mu_{k}\,\sqrt{n_{k}}}\sum_{i\in\left[n_{k}\right]}\left(\ket{\sigma_{k}^{i}}_{S}\left(\bigotimes_{j\in[k+1]}\ket{\mathrm{Bin}\big(v^{i,j}\big)}_{V}\right)\right)\ket{0}^{\otimes a_{k}}_{A}+\ket{\perp_{\Sigma_{k}}}.
\end{align}

Afterwards, a unitary $\brm{U}_{{\rm nodes}_j}^{\rm data}$, that is implementable via QRAM~\cite{Giovannetti2008QRAM} or QROM~\cite{Babbush2018EncodingComplexity}, loads the (approximate with $m$-bit precision) node data $\widetilde{\omega}_{j}^{0}$ for $v^{i,j}$ into $W$, i.e., $\brm{U}_{{\rm nodes}_j}^{\rm data}:\ket{\mathrm{Bin}(v^{i,j})}_{V_j}\ket{0}_{W_j}^{\otimes m}\mapsto\ket{\mathrm{Bin}(v^{i,j})}_{V_j}\ket{\widetilde{\omega}_{j}^{0}}_{W_j}^{\otimes m}$. For all $n$ node data, we can load them coherently by applying a unitary
\begin{align}\label{eq:node_data_loading}
    \brm{U}_{{\rm nodes}}^{\rm data}=\bigotimes_{i\in[n]}\big(\ket{\mathrm{Bin}(i)}\bra{\mathrm{Bin}(i)}\otimes\brm{U}_{{\rm nodes}_i}^{\rm data}\big)
\end{align}
The state is then given by
\begin{eqnarray}\label{eq:node_data_load_state}
    \ket{\Psi_3}=\frac{1}{\mu_{k}\,\sqrt{n_{k}}}\sum_{i\in\left[n_{k}\right]}\left(\ket{\sigma_{k}^{i}}_{S}\left(\bigotimes_{j\in[k+1]}\ket{\mathrm{Bin}\big(v^{i,j}\big)}_{V}\ket{\widetilde{\omega}_{j}^{0}}_{W}\right)\right)\ket{0}^{\otimes a_{k}}_{A}+\ket{\perp_{\Sigma_{k}}}.
\end{eqnarray}

Let
\begin{equation}
    f_i:=f(\omega_{i,0}^{0},\cdots,\omega_{i,k}^{0})\quad\text{and}\quad \widetilde{f}_i:=f(\widetilde{\omega}_{i,0}^{0},\cdots,\widetilde{\omega}_{i,k}^{0}),
\end{equation}
and let $\overline{f}_i$ denote the $m'$-bit output of $\textsc{QFT}_{f}$ on the approximate node data
$(\widetilde{\omega}_{i,0}^{0},\cdots,\widetilde{\omega}_{i,k}^{0})$, so that
\begin{equation}\label{eq:m'-bit_precision}
    \left|\overline{f}_i-\widetilde{f}_i\right|\le 2^{-m'}.
\end{equation}
In the next step, $\textsc{QFT}_{f}$ computes $\overline{f}_i$ into register $W_f$ using quantum arithmetic~\cite{Cuccaro2004ACircuit}, namely additions, subtractions, and multiplications on the node data $(\widetilde{\omega}_{i,0}^{0},\cdots,\widetilde{\omega}_{i,k}^{0})$ stored in register $W$. Applying {unitary $\brm{U}_{\rm amp}$} of \Cref{lem:digita-analog_encoding} to the scaled values $y_i:=\overline{f}_i/\Lambda_f\in[-1,1]$ converts them into amplitudes on the success branch $F=\ket{0}$, yielding
\begin{align}\label{eq:amplitude_transduce_state}
    \ket{\Psi_4}
    &=
    \frac{1}{\mu_{k}\,\Lambda_f\,\sqrt{n_{k}}}
    \sum_{i\in\left[n_{k}\right]}
    \left(
        \overline{f}_i\ket{\sigma_{k}^{i}}_{S}
        \left(
            \bigotimes_{j=0}^{k}
            \ket{\mathrm{Bin}\big(v^{i,j}\big)}_{V_j}
            \ket{\widetilde{\omega}_{i,j}^{0}}_{W_j}
        \right)
        \ket{\overline{f}_i}_{W_f}
    \right)
    \ket{0}^{\otimes a_k}_{A}\ket{0}_{F}
    +
    \ket{\perp_{\brm{\Omega}^{k}_{\rm nodes}}},
\end{align}
where $(\brm{1}_{|S+A+V+W+W_f|}\otimes\bra{0}_{F})\ket{\perp_{\brm{\Omega}^{k}_{\rm nodes}}}=0$, with $|S+A+V+W+W_f|$ the total size of registers $S,A,V,W,W_f$. Uncomputing $\textsc{QFT}_{f}$, $\brm{U}^{\rm data}_{\rm nodes}$, and $\brm{U}_{v}$ resets $V,W,W_f$ to $\ket{0}$ on the success branch without changing the amplitudes in $S$. Thus, we obtain 
\begin{align}
    \frac{1}{\mu_{k}\,\Lambda_f\,\sqrt{n_{k}}}\sum_{i\in\left[n_{k}\right]}\overline{f}_i\ket{\sigma_{k}^{i}}_{S}
    \ket{0}^{\otimes a_k}_{A}\ket{0}_{F}+\ket{\perp_{\brm{\Omega}^{k}_{\rm nodes}}}&=\frac{1}{\overline{\beta}_{\rm nodes}}\ket{\overline{\brm{\omega}}^{k}_{\rm nodes}}
    \ket{0}^{\otimes a_k}_{A}\ket{0}_{F}+\ket{\perp_{\brm{\Omega}^{k}_{\rm nodes}}},
\end{align}
which is $\ket{\overline{\brm{\Omega}}_{\rm nodes}^{k}}$.

{We now analyze the error propagation from the approximate node data to the encoded node-aggregated data. From \Cref{eq:U_node_i}, each node value is encoded with error at most $2^{-m}$. Then from \Cref{eq:m'-bit_precision}, for any simplex $\sigma_k^i$, the triangle inequality together with the multivariable mean value theorem gives
\begin{align}\label{eq:multivariable_error}
    |\overline{f}_i-f_i|
    &\le
    \left|\overline{f}_i-\widetilde{f}_i\right|+\left|\widetilde{f}_i-f_i\right| \le 2^{-m'}+\sum_{j=0}^{k}\sup_{\omega^0}\left|\frac{\partial f}{\partial \omega_j^0}\right|\,|\widetilde{\omega}_{i,j}^{0}-\omega_{i,j}^{0}|\le
    2^{-m'}+L_f\,2^{-m}.
\end{align}
Therefore, choosing $m'\le m-\lceil\log_2 L_f\rceil$ ensures that the propagated error from node data is at most the arithmetic error, and hence
\begin{equation}
    |\overline{f}_i-f_i|\le O(2^{-m}).
\end{equation}
Thus, the amplitudes in Eq.~\eqref{eq:amplitude_transduce_state} approximate the ideal coefficients $f_i$ entrywise with error $O(2^{-m})$, and the uncomputation step does not change this bound. Applying \Cref{lem:digita-analog_encoding} to $y_{i}$, we have
\begin{align}
    \big\|\ket{\overline{\brm{\omega}}^{k}_{\rm nodes}}-\ket{\brm{\omega}^{k}_{\rm nodes}}\big\|_2\leq O(2^{-m}\,\overline{\beta}_{\rm nodes}),
\end{align}
where $\beta_{\rm nodes}$ is given in \Cref{eq:beta_nodes_corrected}.}

Lastly, we evaluate the cost of this algorithm as follows. First, the algorithm invokes $\brm{U}_{\brm{\Sigma}_k}$ once. Next, the unitary $\brm{U}_v$ scans the $n$ simplex bits and writes the $(k+1)$ vertex indices. {This step costs $O(nk\log n)$ gates, since it requires $(k+1)\log n$ controlled multi-$\mathrm{X}$ operations, each of which can be implemented using $O(n)$ gates~\cite{Low2024tradingtgatesdirty}.} The data loader produces $n$ coherent loads for $n$-node data, which, in general, costs $O(n)$ gates~\cite{Giovannetti2008QRAM,Babbush2018EncodingComplexity}. The reversible arithmetic $\textsc{QFT}_f$ and the amplitude transduction cost $\mathrm{poly}(k)=k^{\mathfrak{c}}$ gates by \Cref{def:QFT-multi-functions} and \Cref{lem:digita-analog_encoding}. Uncomputation doubles the number of reversible subroutines while preserving asymptotic scaling. In total, the algorithm uses a single query to $\USigma$ with additional  $O(nk\log{n}+\mathrm{poly}(k))=O(nk\log{n}+k^{\mathfrak{c}})$ gates. The ancilla cost is dominated by storing $V$ ($O((k+1)\log n)$), $W$ ($O((k+1)m)$), and $W_f$ ($O(m+\log{L_f})$), plus $a_k$ and the single flag qubit $F$.
\end{proof}

\section{Proofs of \Cref{sec:quantum_advantage_regimes}}\label{app:proof_quantum_advantage_regimes}

This appendix instantiates the oracle/query complexities from \Cref{thm:main-T1,thm:no-phase-locking} into end-to-end gate counts for the concrete instances used in \Cref{sec:quantum_advantage_regimes}. We count the number of gates for the quantum routines and the floating-point arithmetic operations for the classical baselines. Unless stated otherwise, we treat the precision parameters $(\epsilon,\delta,\Delta)$ as fixed constants and absorb polylogarithmic dependencies (including those arising from QSVT polynomial degree) into
$\tilde{O}(\cdot)$.

We work with a clique complex $\mathcal{K}_n(G)$ built from an underlying graph $G=(\mathcal V,\mathcal E)$. As in the QTDA literature~\cite{berry2023analyzing,McArdle2022AQubits}, the simplex membership oracle $\Om$ in \Cref{def:membership-oracle} can be implemented using either the edge set $\mathcal E$ or the missing-edge set $\mathcal E^c$ for all $p=\{k-1,k,k+1\}$. 
We denote by
\begin{align}\label{eq:membership_oracle_cost}
    \mathrm{G}_{\Om}=\widetilde{O}\big(\min\{|\mathcal E|,|\mathcal E^c|\}\big)
\end{align}
an upper bound on the gate complexity of one call to $\Om$. {In general, we denote by $\mathrm{G}_{\brm{U}}$ the complexity of implementing unitary $\brm{U}$ for simplicial complex and initial data in question.} Then, the following lemma holds.
\begin{lemma}[Gate complexity of $\USigma$ for all $k$]\label{lem:usigma_complexity}
    Given a clique complex $\mathcal{K}_{n}(G)$, there exists a quantum circuit constructing $\USigma$ to prepare state in the form of \Cref{def:prob-sigma} with
    \begin{align}
    \label{eq:mu_k_Task1}
        \mu_{k}^{2}=\frac{\binom{n}{k+1}}{n_{k}},
    \end{align}
    for all $k<n$ using $O(nk+\mathrm{G}_{\brm{O}_{m_k}})$ gates.
\end{lemma}
\begin{proof}
    To estimate the gate count in $\USigma$, one of the simplest ways is to start with a uniform superposition of all possible $k$-simplices on $n$ nodes, that is, basis states with $(k+1)$- bit Hamming weight. This refers to a Dicke state: $\ket{D{(n,k+1)}}=1/\binom{n}{k+1}^{1/2}\Big(\sum_{i\in\left[\binom{n}{k+1}\right]}\ket{\sigma_{k}^{i}}\Big)$. The recent development of Dicke state preparation unitary, denoted as $\brm{U}_{D(n,k+1)}$, requires only $O(nk)$ gates~\cite{Bartschi2019DeterministicStates,Bartschi2022Short-DepthPreparation}. We then construct $\USigmaLU$ based on $\brm{U}_{D(n,k+1)}$ as
    \begin{align}
        \USigmaLU=\mathrm{C}_{\Pi_{k\pm1}^{\perp}}\mathrm{NOT}\,\big(\brm{1}\otimes\brm{U}_{D(n,(k+1)\pm1)}\big).
    \end{align}
    Combining the gate cost of $\brm{U}_{D(n,k+1)}$ and $\Om$ (see \Cref{eq:membership_oracle_cost}) in $\mathrm{C}_{\Pi_{k\pm1}^{\perp}}\mathrm{NOT}$ (see \Cref{eq:controlled_k_projection}), the total number of gates for constructing $\USigmaLU$ is $O(nk+\mathrm{G}_{\brm{O}_{m_k}})$ gates.
\end{proof}
\begin{proof}[Proof of \Cref{cor:order_parameter_est_clique}]\label{proof:cor_Task1}
    To obtain the gate complexity of the \Cref{alg:T1}, we start by estimating the number of gates in $\UTheta$. We assume that there exists a probabilistic state preparation unitary $\UTheta$ with the number of gates $\mathrm{G}_{\brm{\Theta}^{k}}$ preparing $\ket{\brm{\Theta}^{k}}$ with $\alpha,a= O(1)$. 
    
    Now, we estimate the gate count of \Cref{alg:T1} for clique-dense complexes. In this instance, $n_{k}/\binom{n}{k+1}$ is a constant or $n_{k}=\Theta\big(\binom{n}{k+1}\big)$ and $|\mathcal{E}|=n_{1}=O(n^{2})$. Such properties allow us to simplify the complexity statement in \Cref{thm:main-T1}. Recall that {\Cref{alg:T1} invokes} $\UTheta$ and $\brm{O}_{m_{k\pm1}}$ a total of $\tilde{O}(\gamma_1 r_1)$ times, makes $O(r_1)$ calls to $\USigmaLU$, and uses an additional $\tilde{O}((n+a)\gamma_1 r_1)=\tilde{O}(n\gamma_1 r_1)$ gates, since $a= O(1)$. In clique-dense complexes, we have \begin{align}\label{eq:gamma_1_task1}
        \gamma_{1}=\sqrt{n}\,\alpha\,{\cal N}(\brm{\theta}^{k})=O\left(n\,\binom{n}{k+1}\right)^{1/2}
    \end{align}
    since ${\cal N}(\brm{\theta}^{k})\le\pi\sqrt{n_k}= O\big(\binom{n}{k+1}^{1/2}\big)$ and for $\alpha= O(1)$. {Together with \Cref{lem:usigma_complexity}}, we then obtain
    \begin{align}\label{eq:r_1_Task1}
        r_{1}=O\left(\frac{\binom{n}{k+2}}{n_{k+1}+n_{k-1}}\right)= O(1)
    \end{align}
    {for all $n$ and $k$, since $n_{k\pm1}=\Theta\big(\binom{n}{(k+1)\pm1}\big)$ in dense complexes}. Thus, from \Cref{eq:gamma_1_task1,eq:r_1_Task1}, \Cref{lem:usigma_complexity}, and the fact that $\mathrm{G}_{\brm{O}_{m_{p}}}=O(n^2)$ gates for clique-dense complexes, the gate complexity of \Cref{alg:T1} is given by
    \begin{align}\label{eq:T_Q^1}
        T_{\rm Q}^{(1)}&=O\left(\left(\mathrm{G}_{\brm{\Theta}^{k}}+\mathrm{G}_{\brm{O}_{m_k}}+n\right)\,\gamma_1\,r_1+\left(nk\log{n}+\mathrm{G}_{\brm{O}_{m_{k\pm1}}}\right)r_1\right)\nonumber\\
        &=O\left(\left(\mathrm{G}_{\brm{\Theta}^{k}}+n^2\right)\,\left(n\,\binom{n}{k+1}\right)^{1/2}\right).
    \end{align} 
    
    On the other hand, the computational cost of exactly computing $\mathrm{R}(\brm{\theta}^{k})$ classically is dominated by the cost of sparse matrix-vector multiplication of $\bk\brm{\theta}^{k}$ and $\bkk^{\rm T}\brm{\theta}^{k}$ and summing all components $(\brm{\theta}^{k}_{[\pm]})_{i}$ for all $i\in[n_{k\pm1}]$. This cost is given by
    \begin{align}
        T_{\rm C}^{(1)}&=\Theta\big(\mathrm{nnz}(\bk)+\mathrm{nnz}(\bkk) + n_{k-1}+n_{k+1}\big)\nonumber\\
        &=\Theta\big((k+1)\,n_k + (k+2)\,n_{k+1}\big)
    \end{align}
    arithmetic operations. In the clique-dense complexes and large $n$ and small $k$ regime, this becomes
    \begin{align}\label{eq:T_C^1}
        T_{\rm C}^{(1)}=\Theta\left(n\,\binom{n}{k+1}\right)
    \end{align}
    operations. Then, the proof of \Cref{cor:order_parameter_est_clique} follows immediately by dividing \Cref{eq:T_C^1} by \Cref{eq:T_Q^1}.
\end{proof}

\begin{proof}[Proof of \Cref{cor:no-phase-locking_complete_multipartite_graphs}]\label{proof:cor_task2}
    In this case, we invoke the efficient state preparation for simplicial node-aggregate natural frequency described in \Cref{lem:simplicial node data encoding}. Recall that the algorithm makes a single call to $\USigma$ and $O(nk\log{n}+\mathrm{poly}(k,m))$ additional gates. Thus, from \Cref{lem:simplicial node data encoding,lem:usigma_complexity}), the number of gates constructing $\brm{U}_{\brm{\Omega}^{k}_{\rm nodes}}$ is 
    \begin{align}\label{eq:uomega_nodes_complexity}
        \mathrm{G}_{\brm{\Omega}_{\rm nodes}^{k}}=O\left(nk\log(n)+k^{\mathfrak{c}}+\mathrm{G}_{\brm{O}_{m_{k}}}\right).
    \end{align}
    Note that we also have that $\mu_{k}^{2}=\binom{n}{k+1}/n_{k}$ from \Cref{eq:mu_k_Task1}.  
    
    Now, we can simplify some terms in the complexity statement of \Cref{thm:no-phase-locking}. Recall that \Cref{alg:T2} makes $\tilde{O}(r_{\frak m})$ calls to $\brm{U}_{\brm{\Omega}^k_{\mathrm{nodes}}}$, $\tilde{O}(r_{\frak m} d_{\frak m})$ calls to $\brm{O}_{m_p}$, and uses $\tilde{O}(n r_{\frak m} d_{\frak m})$ additional gates. First, observe that $\beta_{\rm nodes}\,\mathcal{N}(\brm{\omega}^{k}_{\rm nodes})=\binom{n}{k+1}^{1/2}$ and $\mathcal{N}\big(\brm{\omega}^{k}_{\rm nodes}\big)= O\big( \sqrt{n_{k}}\big)$. This yields 
    \begin{align}\label{eq:r_m_and_d_m_Task2}
        r_{\mathfrak{m}}=\frac{\kappa_{\mathfrak{m}}^{2}\,\binom{n}{k+1}}{n\,n_\mathfrak{q}},\quad d_{\mathfrak{m}}=\kappa_{\mathfrak{m}}^{2}\log{\left(\sqrt{\frac{nn_{k}}{n_{\mathfrak{q}}}}\right)}
    \end{align}
    with $\delta,\Delta$ is constant, and $K_{\mathfrak{q}}^{s}= O(1)$. Therefore, from \Cref{eq:uomega_nodes_complexity,eq:r_m_and_d_m_Task2}, the gate complexity of \Cref{alg:T2} for an arbitrary clique-complex is given by
    \begin{align}\label{eq:T_Q_Task2_clique}
        T_{\rm Q}^{(2)}&=\tilde{O}\left(\left(\mathrm{G}_{\brm{\Omega}_{\rm nodes}^{k}}+\left(\mathrm{G}_{\brm{O}_{m_{k\pm1}}}+n\right)\,d_{\frak m}\right)\,r_{\frak m}\right)\nonumber\\
        &=O\left(\left(nk\log(n)+\left(k^{\mathfrak{c}}+\mathrm{G}_{\brm{O}_{m_{p}}}+n\right)\,\kappa_{\mathfrak{m}}^{2}\log{\left(\sqrt{\frac{nn_{k}}{n_{\mathfrak{q}}}}\right)}\right)\,\frac{\kappa_{\mathfrak{m}}^{2}\,\binom{n}{k+1}}{n\,n_\mathfrak{q}}\right),
    \end{align}
    where we write the gate complexity $\mathrm{G}_{\brm{O}_{m_{k}}}+\mathrm{G}_{\brm{O}_{m_{k\pm1}}}\log{(nn_k/n_{\frak q})^{1/2}}=O(\mathrm{G}_{\brm{O}_{m_{p}}}\log{(nn_k/n_{\frak q})^{1/2}})$.
    
    We now focus on the case where the SKM dynamics is defined on a clique complex whose underlying graph is a balanced multipartite graph, a graph consisting of $(k+1)$ groups of $m$ nodes, with all edges connecting nodes across different groups. Since there are no $(k+1)$-simplices in this complex, we only consider the case of $\mathfrak{q}=k-1$, or equivalently $\mathfrak{m}=k$. In this case, we have 
    \begin{align}\label{eq:number_of_simplices_multipartite}
        n_{k}=m^{k+1},\quad n_{k-1}=k\,m^{k}.
    \end{align}
    {Since} this complex family does not have $(k+1)$ simplices{,} the $(k+1)$-th boundary matrix {is trivial, i.e.} $\bkk=\brm{0}^{n_{k}\times n_{k+1}}${,} and the smallest nonzero singular value of $\lambda_{\rm min}^{(k)}$ of $\hl=\bk^{\rm T}\bk+\bkk\bkk^{\rm T}=\bk^{\rm T}\bk$ is the square of the smallest nonzero singular value $\zeta_{\rm min}^{(k)}$ of the boundary operator $\bk$. {In} Ref.~\cite{berry2023analyzing}, it is shown that $\lambda_{\rm min}^{(k)}=\sqrt{m}$. Thus, we have
    \begin{align}\label{eq:condition_number_multipartite}
        \kappa_{k}\le \sqrt{\frac{n}{\zeta_{\min}^{(k)}}}=\sqrt{(k+1)}.
    \end{align}
    
    Now, we can further simplify \Cref{eq:r_m_and_d_m_Task2} using \Cref{eq:number_of_simplices_multipartite,eq:condition_number_multipartite} and Stirling approximation {to show that}
    \begin{align}\label{eq:stirling_Task2}
        r_{k}=\Theta\left( \frac{e^{k}}{k^{3/2}}\right),\quad d_{k}=O(k\log{m})
    \end{align}
    for large $n$ and small $k$. Following Ref.~\cite{berry2023analyzing}, we constructs $\Om$ with 
    \begin{align}\label{eq:simplex_membership_multipartite}
        {G}_{\brm{O}_{m_p}}=O(|\mathcal{E}^{c}|)=O(m^2k)
    \end{align}
    gates. Then, substituting \Cref{eq:stirling_Task2,eq:simplex_membership_multipartite} to \Cref{eq:T_Q_Task2_clique}, we obtain
    \begin{align}\label{eq:T_Q^2}
        T_{\rm Q}^{(2)}=O\left(\left(m^2k^{1/2}+k^{(\mathfrak{c}-1/2)}\right)\,e^{k}\,\log{m}\right).
    \end{align}

    The best approximate classical approach via iterative solver to output $K_{k\pm1}^{s}$ with precision $\varepsilon$ has the computational cost $T_{\rm C}^{(2)}= \Theta\Big(\mathrm{nnz}(\brm{L}_{k\mp1}^{[\pm]})\, \sqrt{\kappa\big(\brm{L}_{k\mp1}^{[\pm]}\big)}\, \log(1/\Delta)\Big)$ arithmetic operations~\cite{Sachdeva2014FasterAlgorithm,Musco2024StableLanczos}. First, in $\mathcal{K}_{n}(G_{m,k})$, we have $\kappa\big(\brm{L}_{k-1}^{[+]}\big)=(\zeta_{\max}(\bk)^{2}/n)\kappa_{k}^{2}=\kappa_{k}^{2}\le k+1$~\cite{berry2023analyzing}, where $\zeta_{\max}(\bk)$ is the maximum singular value of $\bk$, {that is upper bounded by $\sqrt{n}$~\cite{Horak2013SpectraComplexes}}. Next, let $d_{i}^{u}$ be the the {number of $(k+1)$-simplices containing} $\sigma_{k}^{i}$ and $n_{k-1}^{u}$ be the number of $(k-1)$ simplices that belong at least one $k$-simplex, i.e., has $d_{i}^{u}>0$. The number of nonzero elements $\mathrm{nnz}(\brm{L}_{k-1}^{[+]})=n_{k-1}^{u}+k\,(k+1)\,n_{k}$. Also, we have $n_{k-1}^{u}=n_{k-1}$ and $\mathrm{nnz}(\brm{L}_{k-1}^{[+]})$ becomes $k\,m^{k}+k(k+1)\,m^{k+1}=(k+1)m^{k}(1+mk)= O(k^{2}\,m^{k+1})$. Setting $\Delta$ to be constant, we have
    \begin{align}\label{eq:T_C^2}
        T_{\rm C}^{(2)}=\Theta\left(k^{5/2}\,m^{k+1}\right).
    \end{align}
    Then, the quantum-classical cost ratio follows immediately from the stated \Cref{eq:T_C^2} and \Cref{eq:T_Q^2}.
\end{proof}

\end{document}